\documentclass[10pt]{article}
\usepackage{fullpage,graphicx,subfigure,mathdots,mathpazo,color}
\usepackage{amsmath,amscd,tikz,mathrsfs}
\usepackage[normalem]{ulem}
\usepackage{amsmath}
\usepackage{setspace}

\usepackage{epsfig,amsmath,graphicx,amssymb,overpic}

\usepackage{graphicx}
\usepackage{dcolumn}
\usepackage{bm}
\usepackage{graphicx}
\usepackage{subfigure}
\usepackage{epsfig,amsmath,graphicx,amssymb,overpic,cite}

\usepackage{graphicx}
\usepackage{dcolumn}
\usepackage{bm}
\usepackage{graphicx}
\usepackage{subfigure}
\usepackage{amsthm}

\newcommand{\ii}{\mathrm{i}}
\newcommand{\ee}{\mathrm{e}}

\newcommand{\T}{\mathrm{T}}

\newtheorem{theorem}{Theorem}
\newtheorem{lemma}{Lemma}
\newtheorem{prop}{Proposition}

\newtheorem{conjecture}{Conjecture}
\newtheorem{remark}{Remark}

\newtheorem{assumption}{Assumption}
\newcommand{\bsm}{\boldsymbol}

\usepackage{graphicx}      
\usepackage{titlesec}
\def\bee{\begin{eqnarray}}
\def\ene{\end{eqnarray}}
\def\bes{\begin{subequations}}
\def\ees{\end{subequations}}

\usepackage{color}

\begin{document}

\baselineskip=13pt
\renewcommand {\thefootnote}{\dag}
\renewcommand {\thefootnote}{\ddag}
\renewcommand {\thefootnote}{ }

\pagestyle{plain}

\begin{center}
\baselineskip=16pt \leftline{} \vspace{-.3in} {\Large \bf Higher-order vector Peregrine solitons and asymptotic estimates for \\ the multi-component nonlinear Schr\"odinger equations} \\[0.2in]
\end{center}

\begin{center}
Guoqiang Zhang$^{\,\rm a}$,\, Liming Ling$^{\,\rm b,*}$, and  Zhenya Yan$^{\,\rm c,\,d,\dag}$
\footnote{$^{*}${\it Email address}: linglm@scut.edu.cn (corresponding author)}
\footnote{$^{\dag}${\it Email address}: zyyan@mmrc.iss.ac.cn (corresponding author)}
\\[0.1in]
{\it \small $^{\rm a}$Department of Mathematical Sciences, Tsinghua University, Beijing 100084, China \\
$^{\rm b}$School of Mathematics, South China University of Technology, Guangzhou 510640, China \\
$^{\rm c}$Key Laboratory of Mathematics Mechanization, Academy of Mathematics and Systems Science, \\ Chinese Academy of Sciences, Beijing 100190, China \\
 $^{\rm d}$School of Mathematical Sciences, University of Chinese Academy of Sciences, Beijing 100049, China} \\
 (Date:\,\, \today)
\end{center}

\vspace{-0.05in}

{\baselineskip=13pt

\begin{tabular}{p{16cm}}
 \hline \\
\end{tabular}

\vspace{-0.18in}

\begin{abstract} \small \baselineskip=12pt

We first report the first- and higher-order vector Peregrine solitons (alias rational rogue waves) for the any multi-component NLS equations based on the loop group theory, an explicit $\left(n+1\right)$-multiple eigenvalue of a characteristic polynomial of degree $(n+1)$ related to the condition of Benjamin-Feir instability, and inverse functions. Particularly, these vector rational rogue waves are parity-time symmetric for some parameter constraints. A systematic and effective approach is proposed to study the asymptotic behaviors of these vector rogue waves
such that the decompositions of rogue waves are related to the so-called governing polynomials $\mathcal{F}_\ell(z)$,
which pave a powerful way in the study of vector rogue wave structures of the multi-component integrable systems.
The vector rogue waves with maximal amplitudes can be determined via the parameter vectors, which is interesting and useful in the multi-component physical systems.

\vspace{0.1in} \noindent {\it Keywords:} Multi-component NLS equations, Lax pair, Darboux transform, Benjamin-Feir instability, higher-order vector
Peregrine solitons, parity-time symmetry, asymptotic estimates

\end{abstract}

\vspace{-0.05in}
\begin{tabular}{p{16cm}}
  \hline \\
\end{tabular}

\vspace{-0.15in}

\tableofcontents

\section{Introduction}


\subsection{The $n$-component nonlinear Schr\"odinger equations}

The focusing cubic nonlinear Schr\"odinger (NLS) equation~\cite{nls67,nls67b}
\bee\label{nls}
 {\rm i}q_t+\frac12 q_{xx}+|q|^2q=0,\quad  q:\, \mathbb{R}^2\to \mathbb{C},
\ene
where $q=q(x,t)$ stands for the complex field, and the subscript denotes the partial derivative, can be used to describe the nonlinear wave phenomena in many fields such as nonlinear optics, Bose-Einstein condensates, deep ocean, plasmas physics, and even finance~\cite{nls68,nlsb1,nlsb2,nlsb3,nlsb4,nlsb5,nlsb6,frw}. In fact, it is also a fundamental and important nonlinear evolution partial differential equation in the applied mathematics and mathematical physics~\cite{nlsbook1,nlsbook2}. It is a completely integrable equation, and can be solved via the inverse scattering transform~\cite{nlsist} and other approaches to admit multi-soliton solutions, double-periodic solutions, breathers and algebro-geometric solutions (see, e.g., Refs.~\cite{ist,dt} and references therein). Moreover, it also possesses the multi-Hamiltonian structures and infinity many conversation laws~\cite{ist}. In the past decade, it was paid more attention again. The main reason is that it was found to possess the novel fundamental rogue wave~\cite{Peregrine1983} and multi-rogue wave solutions~\cite{rwbook17,nail09a,nail09b,Guo2012} (also called multi rogons~\cite{rogon}). In 1983 Peregrine~\cite{Peregrine1983} first found its novel and fundamental RW solution (alias Peregrine soliton) in the form of rational function
\bee
 q(x,t)=\left[1-\frac{4(1+2it)}{4x^2+4t^2+1}\right]e^{\mathrm{i}t},
\ene
which leads to $q \sim e^{{\rm i}t}$ as $|x|,\, |t|\to \infty$, and ${\rm max}|q|=3$, which is three times the non-zero background ($|q|\to 1$ as $|x|,\, |t|\to \infty$). The pioneering work did not draw too much attention until Solli, {\it et al}. first observed the optical RW phenomenon in 2007~\cite{orw}, and Akhmediv, {\it et al.} rediscovered the Peregrine soliton and its higher-order extensions via the modified Darboux transformation in 2009~\cite{nail09a,nail09b}. In particular, the Peregrine soliton has been verified to well agree with the RW experiments of nonlinear optics~\cite{ps}, deep water tank~\cite{wrw} and multi-component plasma~\cite{prw}. The physical mechanisms of RW generation contain many reasons
~\cite{Kharif2003,Pelinovsky2016}, but the usual RW generation is frequently related with the Benjamin-Feir instability (also modulation instability (MI))~\cite{mi,mi2}.


Except for the focusing NLS equation, other integrable nonlinear wave equations were also verified to possess the RW solutions, such as the Hirota equation~\cite{hirota}, Ablowitz-Latik equation~\cite{hirota,al}, Sasa-Satsuma equation~\cite{ss,ss2}, Gerdjikov-Ivanov equation~\cite{gi},  Davey-Stewartson-I equation~\cite{ds1}, derivative NLS equation~\cite{dnls,dnls2}, quintic NLS equation~\cite{qnls}, fifth-order NLS equation~\cite{5nls}, etc.  Even if the RWs of the two-component\cite{guo2011,Ling2014,yan2011,Tao2012,Li2013,Baronio2012,He2013,Baronio2014,wen2015,wen2016,zhang2017,zhang2017b,zhang2019,
Bludov2010,Zhao2012,Chen2013,Ling2016,Zhao2014,Zhao2016,Chen2015,LingZ-19} and three-component\cite{Zhao2013,Baronio2013,zhang2018,zhang2018b} nonlinear wave equations were obtained, but it is still an open issue for the RWs of the $n\, (n>3)$-component integrable nonlinear partial differential equations (PDEs). A key and hard point is that these solutions will correspond to the Riemann surface with sheets more than $3$, which is hard to solve in a closed form. In other words, one needs
to find explicitly  a multiple root of a polynomial of degree $(n+1)$ with $n\geq 4$, which is in general hard to be studied.

 In this paper, we would like to answer this question by studying the focusing $n$-component NLS ($n$-NLS) equations~\cite{Ablowitz2004}
 \begin{equation}\label{n-NLS}
\ii \mathbf{q}_t+\frac{1}{2}\mathbf{q}_{xx}+\left\|\mathbf{q}\right\|_2^2\mathbf{q}=\mathbf{0},\quad x,\, t\in \mathbb{R},
\end{equation}
where $\mathbf{q}(x,t)=\left(q_1(x,t), q_2(x,t),\cdots, q_n(x,t)\right)^{\T}$ is the complex vector field, the subscripts stands for the partial derivatives with respect to the variables, and $\left\|\cdot\right\|_2$ is the standard Euclidean norm. System (\ref{n-NLS}) contains the special nonlinear wave equations, such as the single NLS equation ($n=1$)~\cite{nls67,nls67b}, the Manakov system ($n=2$)~\cite{manakov}, and other multi-component NLS equations~\cite{nnls1,nnls2,nnls3,nnls4}. The $n$-NLS equations \eqref{n-NLS} are completely integrable and admit the $(n+1)\times(n+1)$
Zaharov-Shabat spectral problem~\cite{Ablowitz2004}
\begin{equation}\label{laxs}
\Phi_x=\mathbf{U}(\lambda;x,t)\Phi,\quad \mathbf{U}(\lambda; x, t)=\ii\left(\lambda\sigma_3+\mathbf{Q}\right),\quad\sigma_3={\rm diag}\left(1,-{\mathbb I}_n\right),\quad {\bf Q}=\begin{pmatrix}
\mathbf{0}_{1\times n} &{\bf q}^{\dag} \\
{\bf q} &\mathbf{0}_{n\times n} \\
\end{pmatrix},
\end{equation}
and the associate evolution part is
\begin{equation}\label{laxt}
 \Phi_{t}=\mathbf{V}(\lambda;x,t)\Phi,\quad \mathbf{V}(\lambda;x,t)=\ii\left(\lambda^2\sigma_3+\lambda {\bf Q}-\frac{1}{2}\sigma_3{\bf Q}^2-\frac{\ii}{2}\sigma_3 {\bf Q}_x\right),
\end{equation}
where $\Phi=\Phi(\lambda;x,t)$ is the complex matrix eigenvalue, $\dag$ denotes the conjugate transpose, and $\lambda \in\mathbb{C}$ is a spectral parameter. The compatibility condition $\Phi_{xt}=\Phi_{tx}$, i.e.,  the zero-curvature equation $\mathbf{U}_t-\mathbf{V}_x+[\mathbf{U}, \mathbf{V}]=0$,  of Eqs.~(\ref{laxs}) and (\ref{laxt}) leads to system (\ref{n-NLS}). In this paper we would like to study the $n$-NLS equations (\ref{n-NLS}) via the equivalent Lax pair ~(\ref{laxs}) and (\ref{laxt}), rather than the system itself.



We would like to focus on the rational RW solutions of the $n$-NLS equations (\ref{n-NLS}) on the degenerated MI. The vector rational RW solutions correspond to the $(n+1)$-sheeted Riemann surface $(\chi,\lambda)$:
\begin{equation}\label{eq:riemann-surface}
\chi-2\lambda-\sum_{i=1}^n\frac{a_i^2}{\chi+b_i}=0,\quad a_j>0,\,\,\, b_j\in\mathbb{R},\qquad b_i\neq b_j\,\, \text{ if } \,\, i\neq j,
\end{equation}
and the branch points are determined by the following algebraic equations
\begin{equation}\label{eq:multi-root}
1+\sum_{i=1}^n\frac{a_i^2}{(\chi+b_i)^2}=0,
\end{equation}
where the above algebraic equation \eqref{eq:multi-root} has $n$ pairs of non-real roots. Further, by the Riemann-Hurwitz formula, we find that the genus of Riemann surface is $g=-(n+1)+1+B/2=0$, where the parameter $B$ denotes the branching number with the value $2n$. The algebraic equation corresponds to the rational RW solutions for the $n$-NLS equations. Unluckily, the algebraic equation \eqref{eq:multi-root} can in general not be solved in a closed form. One way to solve it is using the numerical roots instead of the closed form solution. But the numerical solution can not be  used to analyze the general rules for those solutions. Then obtaining a closed formal solution for the equation \eqref{eq:multi-root} is benefit to the analysis for the dynamics of rational RW solutions involving the parameters. In this work, we find that  equation \eqref{eq:multi-root} can be solved under some special parameter setting, in which the roots are a pair $n$-times multiple one.

The decomposition of multi-solitons at $t\to\pm\infty$ is well known for us since 70's in the last century \cite{Faddeev}. The RW solutions have the enormous complicated dynamical structures by choosing the distinct parameters, in which the decomposition of RWs was observed about a decade ago. Until now, there was no a powerful tool to analyze the decomposition of RWs. Most of previous works are based on the RW
profiles~\cite{KedzioraAA-11,KedzioraAA-13}, in which the conjectures on the spatial-temporal structures of high-order RWs for the scalar NLS equations were proposed\cite{BilmanLM-20,Ling2017}. The results on the numbers of peaks for the fusion and fission type high-order RWs for the scalar NLS equations were solved in~\cite{HeZWPF-13}. The systematic analysis for the second-order RWs of the scalar NLS equation was performed in the recent literature \cite{BilmanM-19}. To the best of our knowledge, a systematic tool to analyze the dynamics of RWs is still a long-standing problem on the studying of RWs, which will be an important promotion to the further studies of RWs in the both theoretic and applied aspects. We find that the high-order RWs can be decomposed into several lower-order RWs, in which the locus of lower-order RWs can be approximately determined  by the roots of some algebraic equations. In this work, we try to propose an approach to analyze the high-order RWs quantitatively, which may  prove the
conjectures~\cite{KedzioraAA-11,KedzioraAA-13} for the lower-order RW cases rigorously.  When we finished the work, a recent work on the pattern of scalar rogue waves by using the Hirota method was posed on the arxiv \cite{Yang-20}, in which the rogue waves for scalar NLS equation is found to be relative to the Yablonskii-Vorob\'ev polynomials remarkably. It is hopeful that the vector ${\cal PT}$ rogue wave in current work also has the intimate relationship with the Yablonskii-Vorob\'ev polynomials as well, which will need to further study due to the complexity of multi-sheet algebraic Riemann surface.

\subsection{The main results}

The main results of this paper are summarized as follows:

 \begin{itemize}

\item{} We first find a new explicit $(n+1)$-multiple root of the characteristic polynomial of degree $n+1$ generated from the Lax pair (\ref{laxs}) and (\ref{laxt}) with the initial plane-wave solutions;

\item{} Based on the explicit $(n+1)$-multiple root, a family of first-order vector rational RWs of the $n$-NLS equation (\ref{n-NLS}) is explicitly found, and admits the parity-time (${\cal PT}$) symmetric structure under some parameter constraints;

\item{} The first-order vector rogue waves with ${\cal PT}$-symmetry can be classified via the degree $\ell$ of the polynomial $\mathbf{A}\bsm{\beta}$. When $\ell=1$, the firs-order vector RWs become the fundamental RWs containing three types of RWs, i.e., bright, dark, and four-pated RWs. However when $\ell\geq 2$, the dynamical behaviors of vector RWs as some parameters tend to infinity can be modulated by the new introduced {\it  governing polynomials} $\mathcal{F}_\ell$'s. The leading term of asymptotic formula is the linear superposition of some fundamental RWs, whose central locations are related to roots of $\mathcal{F}_\ell$'s;

\item{} An new techinique of inverse function is proposed and used to generate the explicit formula of higher-order vector rogue waves of the $n$-NLS equation (\ref{n-NLS});

\item{} Two kinds of maximal amplitude constants $\mathcal{GA}[N], \mathcal{A}_j[N]$ of   $N$th-order vector rogue wave solutions with $\bsm{\beta}=\bsm{\eta}$ and $\bsm{\beta}=\bsm{\eta}_j$ are studied.

\end{itemize}

The organization of this paper is as the following: In section \ref{sec2}, we derive
the generalized Darboux transformation for the $n$-NLS equations by the loop group method~\cite{Terng1998}.
In section \ref{sec3}, we construct the first-order vector rational RW solutions for the $n$-NLS equations
by finding an $(n+1)$-multiple root of a characteristic polynomial of degree $n+1$ and Darboux transformation. Based on the degree of the polynomial $\mathbf{A}\boldsymbol{\beta}$, the first-order RWs are classified into $n$ cases. The dynamical behaviors of first-order RWs are also studied by performing an asymptotic analysis method and the characteristics of dynamics are related to the so-called governing polynomial $\mathcal{F}_\ell$. High-order RWs can be decomposed into lower-order RWs under the sense of the large enough parameters. Besides, two kinds of maximal amplitude $\mathcal{GA}[1], \mathcal{A}_j[1]$ are presented and studied.
In section \ref{sec4}, we presented an explicit formula of high-order RWs for $n$-NLS equations. The asymptotic analysis method can also be used to study the dynamical behaviors of high-order RWs. Moreover, two kinds of maximal amplitude $\mathcal{GA}[N], \mathcal{A}_j[N]$ for arbitrary $N\in \mathbb{N}^+$  are presented and studied. Section 5 gives the conclusions and discussions.

\section{Preliminaries}\label{sec2}

In this section we recall the basic constructions of one-fold and multi-fold Darboux transformation\cite{Terng1998} for the multi-component NLS equations (\ref{n-NLS}) in terms of the Darboux matrix with simple poles, and then present a general multi-fold Darboux transformation by means of the Darboux matrix with $k$ high-order poles.

\subsection{Darboux transform with simple poles}

 A rigorous inverse scattering analysis to the spectral problem \eqref{laxs} under the nonzero plane-wave background:
\begin{equation}
q_j\sim a_j\ee^{\ii b_j x\pm \theta_{\pm}},\qquad x\to\pm\infty,
\end{equation}
is still an open question until now, where $a_j$, $b_j$'s are arbitrary non-zero real parameters. Partial results can be recalled in \cite{KrausBK-15}.
To avoid the rigorous inverse scattering analysis on the non-vanishing background, we assume that:
\begin{assumption}
	If $\mathbf{q}(x,t)\in \mathbf{C}^{\infty}(\mathbb{R}^2)\cup\mathbf{L}^{\infty}(\mathbb{R}^2)$, then there exists a sectional analytic matrix:
	\begin{equation}
	\mathbf{M}(\lambda;x,t):\mathbb{R}^2\times(\mathbb{C}/(\Sigma\cup \mathbb{R}))\to {\rm SU}(n+1,\mathbb{C}),
	\end{equation}
	where the contour $\Sigma$ is cut of some Riemann surface, such that
	\begin{itemize}
		\item $\mathbf{M}(\lambda;x,t)$ is meromorphic for $\lambda\in \mathbb{C}/(\Sigma\cup \mathbb{R})$,
\item $\mathbf{M}(\lambda;x,t)=\mathbb{I}+\mathbf{M}_1(x,t)\lambda^{-1}+\mathbf{M}_2(x,t)\lambda^{-2}+\mathcal{O}(\lambda^{-3})$ in the neighborhood of $\infty$,
		\item $\left(\frac{\partial}{\partial x}\Phi(\lambda;x,t)\right)\Phi^{-1}(\lambda;x,t)=\ii\left(\lambda\sigma_3+\mathbf{Q}\right)$, where $\Phi(\lambda;x,t)=\mathbf{M}(\lambda;x,t)\ee^{\ii\lambda(x+\lambda t)\sigma_3}\mathbf{M}^{-1}(\lambda;0,0)$ is holomorphic in $\lambda\in\mathbb{C}$.
	\end{itemize}
\end{assumption}
Actually, we just consider the construction of exact analytic solutions, which implies that the above assumptions will be satisfied automatically.
Further, we find the symmetric property of the holomorphic function.

\begin{lemma}\label{lemma2} $\Phi(\lambda;x,t)\Phi^{\dag}(\lambda^*;x,t)=\mathbb{I}_{n+1}$ if the holomorphic function $\Phi(\lambda;x,t)$ in Eqs.~(\ref{laxs}) and (\ref{laxt}) satisfies $\Phi(\lambda;0,0)=\mathbb{I}_{n+1}$.
\end{lemma}
\begin{proof} See Appendix A.
\end{proof}
On the other hand, by the asymptotic behavior of $\Phi(\lambda;x,t)$, we can set
\bee\label{phi}
\Phi(\lambda;x,t)=\mathbf{M}(\lambda;x,t)\ee^{\ii\lambda(x+\lambda t)\sigma_3}\mathbf{M}^{-1}(\lambda;0,0),
\ene
where $\mathbf{N}(\lambda;x,t)\equiv\mathbf{M}(\lambda;x,t)\ee^{\ii\lambda(x+\lambda t)\sigma_3}$ satisfies the following Riemann-Hilbert problems:
\begin{equation}
\mathbf{N}_+(\lambda;x,t)=\mathbf{N}_-(\lambda;x,t)\mathbf{J}(\lambda),\qquad \lambda\in\Sigma,
\end{equation}
 where $\mathbf{J}(\lambda)$ is a certain jump matrix
and $\Sigma$ is a certain contour. And the analytic matrix $\mathbf{N}(\lambda;0,0)=\mathbf{M}(\lambda;0,0)$ also satisfies the same jump condition. Thus the matrix function $\Phi(\lambda;x,t)$ is holomorphic in the complex plane $\mathbb{C}.$ So are the matrices of $\Phi_x(\lambda;x,t)$ and $\Phi_t(\lambda;x,t)$. By the sectional analytic matrix $\mathbf{M}(\lambda;x,t)$, the integrable hierarchy can be constructed:

\begin{lemma}\label{lemma1}
	Let $\Theta(\lambda;x,t)=\mathbf{M}(\lambda;x,t)\sigma_3\mathbf{M}^{-1}(\lambda;x,t)$, then we have the expansions in the neighorhood of $\infty$
		\begin{equation}
	\Theta(\lambda;x,t)=\sigma_3+\sum_{j=1}^{\infty}\Theta_j(x,t)\lambda^{-j},
	\end{equation}
where $\Theta_j$'s are the differential polynomials of $\mathbf{Q}$ and $\sigma_3$.
\end{lemma}
\begin{proof} See Appendix B.
\end{proof}

Define the following evolutionary hierarchy:
\begin{equation}
\mathbf{V}_j(\lambda;x,t):=[\ii\, \lambda^j\Theta(\lambda;x,t)]_+,
\end{equation}
where the subscript $_+$ means taking the principal part of the corresponding expansions.  Lemma~\ref{lemma1} can generate the first two terms
\begin{equation}
\Theta_1=\mathbf{Q}\equiv [\mathbf{M}_1,\sigma_3], \quad
\Theta_2=-\frac{1}{2}\sigma_3(\ii\mathbf{Q}_x+\mathbf{Q}^2).
\end{equation}
Therefore the first non-trivial flow $\mathbf{V}_2(\lambda;x,t)$ is just $\mathbf{V}(\lambda;x,t)$ given by Eq.~(\ref{laxt}).

Then we consider [cf. Eq.~(\ref{phi})]
\begin{equation}
\Phi_x(\lambda;x,t)\Phi^{-1}(\lambda;x,t)=\mathbf{M}_x(\lambda;x,t)\mathbf{M}^{-1}(\lambda;x,t)+\ii\,\lambda \mathbf{M}(\lambda;x,t)\sigma_3\mathbf{M}^{-1}(\lambda;x,t),
\end{equation}
and
\begin{equation}
\Phi_t(\lambda;x,t)\Phi^{-1}(\lambda;x,t)=\mathbf{M}_t(\lambda;x,t)\mathbf{M}^{-1}(\lambda;x,t)+\ii\,\lambda^2 \mathbf{M}(\lambda;x,t)\sigma_3\mathbf{M}^{-1}(\lambda;x,t),
\end{equation}
which implies the Lax pair \eqref{laxs} and \eqref{laxt} by the Liouville theorem and Lemmas~\ref{lemma2} and \ref{lemma1}.  Now we would like to  proceed to construct the Darboux matrix by calculating the residue.

Let $\Phi(\lambda;x,t)$ be the analytic matrix function solution of Lax pair \eqref{laxs} and \eqref{laxt} in the complex plane $\mathbb{C}$. Choosing any point $\lambda_1\in \mathbb{C}^+$ (i.e., the upper half complex plane), we consider the following gauge transformation
\begin{equation}\label{eq:guage}
\Phi^{[1]}(\lambda;x,t)=\mathbf{T}_1(\lambda;x,t)\Phi(\lambda;x,t)\mathbf{T}_1^{-1}(\lambda;0,0),\quad
\mathbf{T}_1(\lambda;x,t)\equiv\mathbb{I}+\frac{\mathbf{P}_1(x,t)}{\lambda-\lambda_1^*},
\end{equation}
which requires the following conditions:
\begin{itemize}
	\item  the determinant $\oint_{C}{\rm d}\ln\det(\mathbf{T}_1(\lambda;x,t))=2\pi\mathrm{i}\left(N-P\right)=0$, where $N$ denotes the number of zeros and $P$ denotes the number of poles,
	\item ${\rm rank}(\mathbf{P}_1(x,t))=1$, $\mathbf{P}_1^2(x,t)=\alpha\mathbf{P}_1(x,t)$, where $\alpha$ is a undetermined constant,
	\item The same formal Lax pair $\Phi^{[1]}_x=\mathbf{U}^{[1]}\Phi^{[1]}$ and $\Phi^{[1]}_t=\mathbf{V}^{[1]}\Phi^{[1]}$, where $\mathbf{U}^{[1]}=\mathbf{U}\big|_{\mathbf{Q}=\mathbf{Q}^{[1]}},\,  \mathbf{V}^{[1]}=\mathbf{V}\big|_{\mathbf{Q}=\mathbf{Q}^{[1]}}$,

	\item the normalization condition $\Phi^{[1]}(\lambda;x,t)\ee^{-\ii\lambda(x+\lambda t)\sigma_3}\to \mathbb{I}$ as $\lambda\to\infty.$
\end{itemize}

By the Abel formula and the above third requirement, we know that the determinant $\det(\mathbf{T}_1(\lambda;x,t))$ satisfies
$\partial_x\det(\mathbf{T}_1(\lambda;x,t))=\partial_t\det(\mathbf{T}_1(\lambda;x,t))=0,$
which deduces that
\begin{equation}
\det(\mathbf{T}_1(\lambda;x,t))=1+\frac{\alpha}{\lambda-\lambda_1^*}.
\end{equation}
Thus the Darboux matrix and its inverse can be rewritten as
\begin{equation}
\mathbf{T}_1(\lambda;x,t)=\mathbb{I}+\frac{\alpha}{\lambda-\lambda_1^*}\frac{\varphi_1(x,t)\psi_1(x,t)}{\psi_1(x,t)\varphi_1(x,t)},\qquad
\mathbf{T}_1^{-1}(\lambda;x,t)=\mathbb{I}-\frac{\alpha}{\lambda-\lambda_1^*+\alpha}\frac{\varphi_1(x,t)\psi_1(x,t)}{\psi_1(x,t)\varphi_1(x,t)},
\end{equation}
where $\varphi_1(x,t)$ is a column vector and $\psi_1(x,t)$ is a row vector. Therefore the residues for $\Phi^{[1]}(\lambda;x,t)$ at $\lambda=\lambda_1^*$ and $\lambda=\lambda_1^*-\alpha$ are given by
\bee
\begin{array}{rl}
\underset{\lambda=\lambda_1^*}{\rm Res}\Phi^{[1]}(\lambda;x,t)\!\!&=\displaystyle \alpha \frac{\varphi_1(x,t)\psi_1(x,t)}{\psi_1(x,t)\varphi_1(x,t)}\Phi(\lambda_1^*;x,t)
\left(\mathbb{I}-\frac{\varphi_1(0,0)\psi_1(0,0)}{\psi_1(0,0)\varphi_1(0,0)}\right)=0, \vspace{0.1in} \\
\underset{\lambda=\lambda_1^*-\alpha}{\rm Res}\Phi^{[1]}(\lambda;x,t)\!\!&=\displaystyle -\alpha \left(\mathbb{I}-\frac{\varphi_1(x,t)\psi_1(x,t)}{\psi_1(x,t)\varphi_1(x,t)}\right)\Phi(\lambda_1^*-\alpha;x,t)\frac{\varphi_1(0,0)\psi_1(0,0)}{\psi_1(0,0)\varphi_1(0,0)}=0.
\end{array}
\ene
Furthermore, by the linear algebra we know that
\begin{equation}\label{eq:residue}
\begin{split}
\psi_1\Phi(\lambda_1^*;x,t)\left(\mathbb{I}-\frac{\varphi_1(0,0)\psi_1(0,0)}{\psi_1(0,0)\varphi_1(0,0)}\right)=0, \quad
\left(\mathbb{I}-\frac{\varphi_1(x,t)\psi_1(x,t)}{\psi_1(x,t)\varphi_1(x,t)}\right)\Phi(\lambda_1^*-\alpha;x,t)\varphi_1(0,0)=0.
\end{split}
\end{equation}

By the Lemma \ref{lemma2}, we require that
$\left(\Phi^{[1]}(\lambda^*;x,t)\right)^{\dag}\Phi^{[1]}(\lambda;x,t)=\mathbb{I},$
which implies the sufficient condition
\begin{equation}\label{eq:suff}
\mathbf{T}_1(\lambda;x,t)\mathbf{T}_1^{\dag}(\lambda;x,t)=\mathbb{I}.
\end{equation}
Further, the condition $\alpha=\lambda_1^*-\lambda_1$ and $\varphi_1(x,t)=\psi_1^{\dag}(x,t)=\Phi(\lambda_1;x,t)\varphi_1(0,0)$ can deduce the above sufficient condition \eqref{eq:suff} and the residue conditions \eqref{eq:residue}, which means the gauge transformation \eqref{eq:guage} is well defined.

\begin{lemma}[Terng and Uhlenbeck\cite{Terng1998}]\label{thm1}
	Suppose $\mathbf{q}(x,t)$ is a solution of the multi-component NLS equations \eqref{n-NLS}, then the Darboux matrix
	 \begin{equation}
	 \mathbf{T}_1(\lambda;x,t)=\mathbb{I}-\frac{\lambda_1-\lambda_1^*}{\lambda-\lambda_1^*}\mathbf{P}_1(x,t),\quad \mathbf{P}_1(x,t)=\frac{\varphi_1(x,t)\varphi_1^{\dag}(x,t)}{\varphi_1^{\dag}(x,t)\varphi_1(x,t)},
	 \end{equation}
	where $\varphi_1(x,t)=\Phi(\lambda_1;x,t)\mathbf{v}_1$ with $\mathbf{v}_1$ being a constant column vector, can make the new holomorphic function \eqref{eq:guage} satisfy the same Lax pair \eqref{laxs} and \eqref{laxt} by replacing the potential function $\mathbf{Q}(x,t)$ with
	\begin{equation}\label{eq:back}
	\mathbf{Q}^{[1]}(x,t)=\mathbf{Q}(x,t)+(\lambda_1-\lambda_1^*)[\sigma_3,\mathbf{P}_1(x,t)],
	\end{equation}
which is just the one-fold Darboux transformation.
\end{lemma}
\begin{proof}
Consider the matrix function
\begin{equation}
\Phi^{[1]}(\lambda;x,t)=\mathbf{M}^{[1]}(\lambda;x,t)\ee^{\ii\lambda(x+\lambda t)\sigma_3}\left(\mathbf{M}^{[1]}(\lambda;0,0)\right)^{-1},\quad\mathbf{M}^{[1]}(\lambda;x,t)=\mathbf{T}_1(\lambda;x,t)\mathbf{M}(\lambda;x,t).
\end{equation}
Since $\mathbf{T}_1(\lambda;x,t)$ can be expanded in the neighborhood of $\infty$:
\begin{equation}
\mathbf{T}_1(\lambda)=\mathbb{I}-(\lambda_1-\lambda_1^*)\mathbf{P}_1(x,t)\lambda^{-1}+O(\lambda^{-2})
\end{equation}
then we have
\begin{equation}
\mathbf{M}^{[1]}(\lambda;x,t)=\mathbf{T}_1(\lambda)\mathbf{M}(\lambda;x,t)=\mathbb{I}
+\left[\mathbf{M}_1(x,t)-(\lambda_1-\lambda_1^*)\mathbf{P}_1(x,t)\right]\lambda^{-1}+O(\lambda^{-2}),
\end{equation}
which implies the B\"acklund (or Darboux) transformation \eqref{eq:back}. Through the Liouville theorem,  we have
\begin{equation}
\begin{split}
\left(\frac{\partial}{\partial x}\Phi^{[1]}(\lambda;x,t)\right)\left(\Phi^{[1]}(\lambda;x,t)\right)^{-1}&=\left(\frac{\partial}{\partial x}\mathbf{T}_1(\lambda;x,t)\right)\mathbf{T}_1^{-1}(\lambda;x,t)+\mathbf{T}_1(\lambda;x,t)\mathbf{U}(\lambda;x,t)\mathbf{T}_1^{-1}(\lambda;x,t) \\
&=\mathbf{U}^{[1]}(\lambda;x,t), \\
\end{split}
\end{equation}
and
\begin{equation}
\begin{split}
\left(\frac{\partial}{\partial t}\Phi^{[1]}(\lambda;x,t)\right)\left(\Phi^{[1]}(\lambda;x,t)\right)^{-1}&=\left(\frac{\partial}{\partial t}\mathbf{T}_1(\lambda;x,t)\right)\mathbf{T}_1^{-1}(\lambda;x,t)+\mathbf{T}_1(\lambda;x,t)\mathbf{V}(\lambda;x,t)\mathbf{T}_1^{-1}(\lambda;x,t) \\
&=\mathbf{V}^{[1]}(\lambda;x,t),
\end{split}
\end{equation}
which prove that the holomorphic function $\Phi^{[1]}(\lambda;x,t)$ solves the Lax pair \eqref{laxs} and \eqref{laxt} by replacing the potential function $\mathbf{Q}(x,t)$ with Eq.~\eqref{eq:back}.
\end{proof}


\subsection{Darboux transform with $k$ high-order poles}

Here we extend the Darboux matrix with a general case containing $k$ high-order poles for the $n$-NLS equations to present the Darboux transformation:

\begin{theorem}\label{thm2}
	Suppose that $\mathbf{q}(x,t)\in \mathbf{L}^{\infty}(\mathbb{R}^2)\cup \mathbf{C}^{\infty}(\mathbb{R}^2)$ is a smooth solution and  $\Phi(\lambda;x,t)$ is an analytic matrix solution in the whole complex plane $\mathbb{C}$. Then the Darboux transformation
	\begin{equation}\label{eq:darboux}
	\mathbf{T}_N(\lambda;x,t)=\mathbb{I}-\mathbf{Y}_N\mathbf{M}^{-1}\mathbf{D}\mathbf{Y}_N^{\dag},
	\end{equation}
converts the Lax pair \eqref{laxs} and \eqref{laxt} into the new one by replacing the potential functions
	\begin{equation}\label{eq:back-N}
	\mathbf{q}^{[N]}(x,t)=\mathbf{q}(x,t)-2 \mathbf{Y}_{N,2}\mathbf{M}^{-1}\mathbf{Y}_{N,1}^{\dag},
	\end{equation}
where $\mathbf{D}=\mathrm{diag}\left(\mathbf{D}_1, \mathbf{D}_2, \cdots, \mathbf{D}_k\right)$, $\mathbf{D}_s=\left(\frac{\mathrm{He}(i-j+1)}{\left(\lambda-\lambda_s^*\right)^{i-j+1}}\right)_{0\le i, j\le r_s-1}, s=0, 1, \cdots, k$, $\mathrm{He}(x)$ is the heaviside function [$ \mathrm{He}(x)=1$ for $x>0$ and zero otherwise], $\mathbf{M}=\left(\mathbf{M}_{i, j}\right)_{1\le i, j \le k}$, $\mathbf{M}_{i, j}=\left(M_{i, j}^{[s, \ell]}\right)_{\substack{0\le s\le r_i-1\\  0\le \ell\le r_j-1}}$ with
\begin{gather}
M_{i, j}^{[s, \ell]}=\sum_{p=0}^s\sum_{m=0}^\ell \binom{p+m}{m}\frac{\left(-1\right)^m}{\left(\lambda_j-\lambda_i^*\right)^{p+m+1}}\left(\varphi_i^{[s-p]}\right)^\dag \varphi_j^{[\ell-m]},
\end{gather}
the subscript $\mathbf{Y}_{N,1}$ denotes the first row vector of $\mathbf{Y}_N$, the subscript $\mathbf{Y}_{N,2}$ the second up to $n+1$-th row vector,
\begin{gather}
\mathbf{Y}_N =\left(\varphi_1^{[0]},\varphi_1^{[1]},\cdots,\varphi_1^{[r_1-1]}, \varphi_2^{[0]},\varphi_2^{[1]},\cdots,\varphi_2^{[r_2-1]},\cdots, \varphi_k^{[0]},\varphi_k^{[1]},\cdots,\varphi_k^{[r_k-1]} \right),\quad N=\sum_{i=1}^kr_i,
\end{gather}
and   $\varphi_i^{[j]}=\frac{1}{j!}\left(\frac{\rm d}{{\rm d}\lambda}\right)^{j}\Phi(\lambda; x,t)\mathbf{v}(\lambda)|_{\lambda=\lambda_i}$, $1\le i\le k, 0\le j\le r_i-1$.
\end{theorem}

\begin{proof}
The general Darboux matrices can be constructed by the successive iterations of elementary Darboux transformation. Suppose that we iterate the Darboux transformation at $\lambda=\lambda_i$ for $r_i$ times. The Darboux matrices can be obtained by the above-mentioned Theorem \ref{thm1}. We merely need to rewrite the general Darboux matrices by the above compact form \eqref{eq:darboux}.

By the form of poles, assume that the general Darboux matrix can be written as the form
\begin{equation}
\mathbf{T}_N(\lambda;x,t)=\mathbb{I}-\sum_{l=1}^{k}\sum_{i=1}^{r_l}\frac{1}{(\lambda-\lambda_l^*)^i}\sum_{j=1}^{r_l-i+1}|z_l^{[j+i-1]}(x,t)\rangle\langle y_l^{[j]}(x,t)|,
\end{equation}
where the notations $|\cdot \rangle, \langle \cdot |$ represent the row and column vector respectively, and $|\cdot \rangle=(\langle \cdot |)^{\dag}$.
By the formula
\begin{equation}
\mathbf{T}_N(\lambda;x,t)\mathbf{T}_N^{\dag}(\lambda^*;x,t)=\mathbb{I}_{n+1},
\end{equation}
we have the residues:
\begin{equation}
\underset{\lambda=\lambda_i}{\rm Res}\left((\lambda-\lambda_i)^l\mathbf{T}_N(\lambda;x,t)\mathbf{T}_N^{\dag}(\lambda^*;x,t)\right)=0,\quad l=0,1,\cdots, r_i-1;\quad i=1,2,\cdots, k,
\end{equation}
which imply that
\begin{equation}\label{eq:kernel1}
\sum_{j=1}^{l}\left[\frac{1}{(j-1)!}\left(\frac{\rm d}{{\rm d} \lambda}\right)^{j-1}\mathbf{T}_N(\lambda;x,t)|_{\lambda=\lambda_i}\right]\left|y_i^{[l-j+1]}(x,t)\right\rangle=0.
\end{equation}

On the other hand, through the holomorphic function
\begin{equation}
\mathbf{T}_N(\lambda;x,t)\Phi(\lambda;x,t)\mathbf{T}_N^{-1}(\lambda;0,0)=\mathbf{T}_N(\lambda;x,t)\Phi(\lambda;x,t)\mathbf{T}_N^{\dag}(\lambda^*;0,0)
\end{equation}
we have
\begin{equation}
\underset{\lambda=\lambda_i}{\rm Res}\left((\lambda-\lambda_i)^l\mathbf{T}_N(\lambda;x,t)\mathbf{T}_N^{\dag}(\lambda^*;x,t)\right)=0,\qquad l=0,1,\cdots, r_i-1;\qquad i=1,2,\cdots, k,
\end{equation}
which can yield
\begin{multline}\label{eq:kernel2}
\sum_{j=1}^{l}\left(\frac{1}{(j-1)!}\left(\frac{\rm d}{{\rm d} \lambda}\right)^{j-1}\mathbf{T}_N(\lambda;x,t)|_{\lambda=\lambda_i}\right)       \\
\times\left(\sum_{s=1}^{l-j+1} \frac{1}{(s-1)!}\left(\frac{\rm d}{{\rm d} \lambda}\right)^{s-1}\Phi(\lambda;x,t)|_{\lambda=\lambda_i}\left|y_i^{[l-j+2-s]}(0,0)\right\rangle\right)=0.
\end{multline}

Combining Eq.~\eqref{eq:kernel1} with Eq.~\eqref{eq:kernel2}, we arrive at
\begin{equation}
\left|y_i^{[m]}\right\rangle=\sum_{s=1}^{m} \frac{1}{(s-1)!}\left(\frac{\rm d}{{\rm d} \lambda}\right)^{s-1}\Phi(\lambda;x,t)|_{\lambda=\lambda_i}\left|y_i^{[m+1-s]}(0,0)\right\rangle.
\end{equation}
Meanwhile, by the linear algebra and Eq.~\eqref{eq:kernel1}, the row vectors $|z_i^{[m]}(x,t)\rangle$ can be solved exactly, which deduces the formula \eqref{eq:darboux}.
The B\"acklund transformation \eqref{eq:back} between old and new potential functions can be deduced by the following equations:
\begin{equation}
\frac{\partial}{\partial x}\mathbf{T}_N(\lambda;x,t)+ \mathbf{T}_N(\lambda;x,t) \mathbf{U}(\lambda;x,t)=\mathbf{U}^{[N]}(\lambda;x,t)\mathbf{T}_N(\lambda;x,t),
\end{equation}
through the expansions at the neighborhood of $\infty$.
\end{proof}

\begin{remark}
The general Darboux matrix and the corresponding B\"acklund transformation can be rewritten with the other formulation due to the Taylor expansion on the spectral parameters. Here we just write it in a compact and uniform form.
\end{remark}

A direct application for the above theorem is to construct new exact solutions by the B\"acklund transformation. For the vector NLS equations, we can obtain the multi-solitons on the zero background, multi-breathers on the non-vanishing background and so on. In the current work, we focus on the rational RW solutions at the branch point with a genus zero curve. Actually, the above formula \eqref{eq:back-N} is flexible,  we do not need to confirm the fixed formula to construct the exact solutions. Thus when we confront the explicit problems, the suitable modifications on the formulas is better to represent them in a more compact form.

\section{Fundamental vector rational RWs: classification, asymptotics and maximal amplitude}\label{sec3}

In this section, the vector rational RW solutions for the $n$-NLS equations will be explicitly constructed by the theorems in previous section. Here we mainly focus on the vector rational RW solutions on the branch points of multi-sheeted Riemann surface with maximal multiple root. Actually, other distinct types of exact solutions o the $n$-NLS equations can also be constructed by the Darboux transformation. Recently, we gave a short paper~\cite{zlyk2020} for the first-order vector RWs of the $n$-NLS equations, and we here give the detailed discussion about the general first-order vector RWs and their properties.

\subsection{The formula of first-order vector RW solutions with $\left(n+1\right)$-multiple root}

To study the vector rational RWs of the $n$-NLS system, we start from its general non-zero plane-wave (PW) solution
\begin{gather}\label{plane}
\mathbf{q}^{\mathrm{bg}}=\left(q_1^{\mathrm{bg}}, q_2^{\mathrm{bg}}, \cdots, q_n^{\mathrm{bg}}\right)^{\rm T},\quad
q_j^{\mathrm{bg}}=a_j\mathrm{e}^{\mathrm{i}\theta_j},\quad \theta_j=b_jx-\left(\frac{1}{2}\,b_j^2-\left\|\mathbf{a}\right\|_2^2\right)t,\,\,\, \,\, j=1,2,...,n,
\end{gather}
where $\mathbf{a}=\left(a_1, a_2, \cdots, a_n\right)^\mathrm{T},\,\,a_j\not=0,\, b_j\in\mathbb{R}$, the phase and group velocities for the PW of the $j$-th component are $\frac{1}{2}\,b_j-\left\|\mathbf{a}\right\|_2^2/b_j$ and $b_j$, respectively. The MI analysis for the PW \eqref{plane} is given in the {\bf Appendix C} by the method of squared eigenfunction. Actually, the MI analysis for the plane-wave solution can also be done by the linear algebra. Through the squared eigenfunction method, we can understand that the consistence between the RW generation and MI is originated from the solutions of Lax pair.

To apply the above-mentioned Darboux transformation with the initial non-zero PW background (\ref{plane}), one needs to solve the varying-coefficient Lax pair (\ref{laxs}) and (\ref{laxt}) with $\mathbf{Q}=\mathbf{Q}^{\mathrm{bg}}=\mathbf{Q}\big|_{{\bf q}={\bf q}^{\rm bg}}$:
\begin{gather}\label{initial}
\left\{\begin{aligned}
\Phi_x&=\mathbf{U}^{\mathrm{bg}}\left(\lambda; x, t\right)\Phi,\quad \mathbf{U}^{\mathrm{bg}}(\lambda; x, t):=\mathrm{i}\left(\lambda\sigma_3+\mathbf{Q}^{\mathrm{bg}}\right) , \\
\Phi_t&={\mathbf V}^{\mathrm{bg}}\left(\lambda; x, t\right)\Phi,\quad
\mathbf{V}^\mathrm{bg}\left(\lambda; x, t\right):={\rm i}\lambda\left(\lambda\sigma_3+\mathbf{Q}^{\mathrm{bg}}\right)
+\frac{1}{2}\sigma_3\left[\mathbf{Q}^{\mathrm{bg}}_x-\mathrm{i}\left(\mathbf{Q}^{\mathrm{bg}}\right)^2\right],
\end{aligned}\right.
\end{gather}
which can be reduced to the system of linear constant-coefficient PDEs
\begin{equation}\label{constant}
\Psi_x=\mathrm{i}\mathbf{H}\Psi,\quad \Psi_t=\mathrm{i}\left[\frac{1}{2}\mathbf{H}^2+\lambda\mathbf{H}-\left(\sum_{j=1}^na_j^2+\frac{\lambda^2}{2}\right)\mathbb{I}_{n+1}\right]\Psi
\end{equation}
via the gauge transformation
$\Phi=\mathbf{G}\Psi$ with $\mathbf{G}=\mathrm{diag}\left(1, \mathrm{e}^{\mathrm{i}\theta_1}, \cdots,\mathrm{e}^{\mathrm{i}\theta_n}\right)$, where
\begin{gather*}
\mathbf{H}=\begin{pmatrix}
\lambda& \mathbf{a}^\dagger  \\
\mathbf{a}&-\lambda\mathbb{I}_n+\mathbf{b}
\end{pmatrix},\quad \mathbf{b}=\mathrm{diag}\left(b_1, b_2, \cdots, b_n\right).
\end{gather*}

To study vector rational RW solutions of system (\ref{n-NLS}) we require that the coefficient matrix $\mathbf{H}$ in system (\ref{constant}) must admit the multiple eigenvalues.
Here suppose that $\mathbf{H}$ has an $(n+1)$-multiple eigenvalue at $\lambda=\lambda_0$. In the following, we will firstly derive the existence condition of the $(n+1)$-multiple eigenvalue.
For convenience, we define $\mathbf{H}_0:=\mathbf{H}\left|\right._{\lambda=\lambda_0}$ and a polynomial of  degree $(n+1)$ as
\bee\label{Dx}
\begin{array}{rl}
D(z):=& \mathrm{det}\left[\left(z-\lambda_0\right)\mathbb{I}_{n+1}-\mathbf{H}_0\right] \vspace{0.1in}\\
 =&\displaystyle\left(z-2\lambda_0\right)\prod_{j=1}^n\left(z+b_j\right)-\sum_{j=1}^n\left[a_j^2 \prod_{k\ne j}\left(z+b_k\right)\right],
\end{array}
\ene
where $z-\lambda_0$ is regarded as the eigenvalue of $\mathbf{H}_0$.

Then one has the following lemma.
\begin{lemma}\label{root-condition}
Given the real-valued amplitudes $a_j$'s ($a_j\not=0$) and wave-numbers $b_j$'s in the PW background (\ref{plane}), the polynomial $D(z)$ has an (n+1)-multiple root $z=z_0$ if and only if 
$b_k\not=b_j\, (k\not=j)$ and
\begin{gather}\label{ab-relation}
a_j^2=-\frac{\left(b_j+z_0\right)^{n+1}}{\prod_{k\ne j}\left(b_j-b_k\right)}, \quad z_0=\frac{1}{n+1}\left(2\lambda_0-\sum_{k=1}^n b_k\right).
\end{gather}
\end{lemma}
\begin{proof} (Necessity $\Rightarrow$) Since $z=z_0$ is the $(n+1)$-multiple root of $D(z)$, that is,
\bee
\label{dchi}
D(z)=(z-z_0)^{n+1}.
\ene
Matching the $\mathcal{O}\left(z^n\right)$ term in Eqs.~(\ref{Dx}) and (\ref{dchi}) makes sure that the second expression of Eq.~(\ref{ab-relation}) holds.
One can perform the reduction to absurdity to verify that $b_j's$ are mutually different. Without loss of generality, one can suppose $b_1=b_2=\cdots=b_m$, $m\ge 2$ and $b_1, b_{m+1}, b_{m+2}, \cdots, b_n$ are mutually different.
In this case, one obtains
$D(z)=p(z)\left(z+b_1\right)^{m-1}$,
where
\begin{gather*}
p(z)=\left(z-2\lambda_0\right)\prod_{j=m}^n\left(z+b_j\right)-\sum_{j=1}^m a_j^2\prod_{k=m+1}^n\left(z+b_k\right)-\sum_{j=m+1}^n a_j^2\prod_{\substack{k=m\\k\ne j}}^n\left(z+b_k\right).
\end{gather*}
Since $D(z)$ has an $(n+1)$-multiple root, then it follows that
$p(z)=\left(z+b_1\right)^{n-m+2}$.
Then one deduces $\sum_{j=1}^m a_j^2=0$, which is a contradiction. At last, one deduces
$a_j^2=-\left(b_j+z_0\right)^{n+1}/\prod_{k\ne j}\left(b_j-b_k\right)$.

(Sufficiency $\Leftarrow$) Define the polynomial
\begin{gather*}
\widehat{p}(z):=\left(z-2\lambda_0\right)\prod_{j=1}^n\left(z+b_j\right)-\sum_{j=1}^n a_j^2 \prod_{k\ne j}\left(z+b_k\right)-\left(z-z_0\right)^{n+1}.
\end{gather*}
One can obtain that the degree of $\widehat{p}(z)$ is less than $n$ and $\widehat{p}(-b_j)=0$. Since $b_j's$ are mutually different, it follows that $D(z)=\left(z-z_0\right)^{n+1}$. Thus the Lemma follows.
\end{proof}

In Lemma \ref{root-condition}, one establishes the sufficient and necessary condition of the $(n+1)$-multiple root for $D(z)$. Next, one will seek the suitable and mutually different real $b_j's$ and non-zero real $a_j's$ such that the relation (\ref{ab-relation}) holds. To the end, one can discuss it in two cases: i) $\Im(\lambda_0)=0$; ii) $\Im(\lambda_0)\ne 0$.


\begin{lemma}
For the case $\Im\left(\lambda_0\right)=0$, there do not exist suitable and mutually different real $b_j's$ and non-zero real $a_j's$ such that the relation (\ref{ab-relation}) holds, that is, $\mathbf{H}_0$ has no an $(n+1)$-multiple eigenvalue for the case.
\end{lemma}

\begin{proof}
Without loss of generality, one can assume $b_1>b_2>\cdots>b_n$.
\begin{enumerate}
\item [1)] As $n$ is odd, one has
\begin{gather*}
a_1^2=-\frac{\left(b_1+z_0\right)^{n+1}}{\prod_{k=2}^n\left(b_1-b_k\right)}<0,
\end{gather*}
which is a contradiction.
\item [2)] As $n$ is even, one can verify this in the following two cases:
\begin{itemize}
\item As $b_1+z_0>0$, one derives that
\begin{gather*}
a_1^2=-\frac{\left(b_1+z_0\right)^{n+1}}{\prod_{k=2}^n\left(b_1-b_k\right)}<0,
\end{gather*}
which is a constradiction.
\item As $b_1+z_0\le 0$, it follows directly that $b_2+z_0<0$. Then one derives that
\begin{gather*}
a_2^2=-\frac{\left(b_2+\chi_0\right)^{n+1}}{\prod_{k\ne2}\left(b_2-b_k\right)}<0,
\end{gather*}
which is a contradiction.
\end{itemize}
\end{enumerate}
\end{proof}

\begin{lemma}\label{prop-jutixuanqu}
As $\Im\left(\lambda_0\right)\ne0$,  the parameters $a_j's$ and $b_j's$  can be  chosen as
\begin{gather}\label{jutixuanqu}
a_j=\zeta\csc\omega_j, \quad b_j=\zeta\cot\omega_j-2\Re\left(\lambda_0\right),\quad \zeta=\frac{2\Im\left(\lambda_0\right)}{n+1}, \quad \omega_j=\frac{j\pi}{n+1},
\end{gather}
such that the relation (\ref{ab-relation}) holds.
\end{lemma}
\begin{proof}
To find the suitable $a_j's$ and $b_j's$ satisfying the relation (\ref{ab-relation}), it is a natural idea that one sets $b_j+z_0$ in the right hand side of relation (\ref{ab-relation}) $(n+1)$-multiple root. By the symmetry of the potential $\mathbf{q}$ in the $n$-NLS equation (\ref{n-NLS}), one can suppose that
\begin{gather*}
2\lambda_0+nb_j-\sum_{k\ne j}b_k=\rho_j\mathrm{e}^{i\omega_j},\quad j=1, 2, \cdots, n,
\end{gather*}
which lead to
$\rho_j=2\Im\left(\lambda_0\right)\csc\omega_j, \quad b_j=\zeta\cot\omega_j-2\Re\left(\lambda_0\right).$
It is noted that $b_1, b_2, \cdots, b_n$ are mutually different. Substituting $b_j's$ into relation (\ref{ab-relation}) derives
$a_j=\pm \zeta\csc\varphi_j$. Without loss of generality, one here chooses $a_j=\zeta\csc\varphi_j$. This completes the proof of the lemma.
\end{proof}

\begin{remark} It follows from Lemma \ref{prop-jutixuanqu} that each family of parameters $(a_j,\, b_j)$ of the component $q_j^{\rm bg}$ is located on the hyperbola on the $(a_j,\, b_j)$-plane
\bee\label{abr}
  \frac{a_j^2}{\zeta^2}-\frac{\left(b_j+2\Re(\lambda_0)\right)^2}{\zeta^2}=1,\quad j=1,2,...,n.
\ene
with center $(0,\, -2\Re\left(\lambda_0\right))$, vertices $(\pm \zeta,\, -2\Re\left(\lambda_0\right))$, asymptotes $b_j=\pm x-2\Re\left(\lambda_0\right)$, and slopes of asymptotes $\pm 1$.

\end{remark}

\begin{remark}
Lemma \ref{prop-jutixuanqu} gives the explicit existence condition for the $(n+1)$-multiple eigenvalue of matrix $\mathbf{H}\left(\lambda_0\right)$. In fact, it  is also an existence condition for the vector rational RW solutions of system (\ref{n-NLS}).
\end{remark}

With Eq.~(\ref{jutixuanqu}),  one can simplify the $(n+1)$-multiple eigenvalue $z_0$ and $\mathbf{B}:=\mathbf{H}_0-\left(z_0-\lambda_0\right) \mathbb{I}_{n+1}$ as
\begin{gather}\label{BA}
\mathbf{B}=\begin{pmatrix}
 \mathrm{i}n\zeta &\mathbf{a}^\dagger \vspace{0.05in}\\
\mathbf{a}&-z_0 \mathbb{I}_n-\mathbf{b}
\end{pmatrix},\quad
z_0=2\Re\left(\lambda_0\right)+\mathrm{i}\zeta=\frac{2(\lambda_0+n\Re\left(\lambda_0\right))}{n+1}.
\end{gather}

Based on the Darboux transformation technique in Theorem \ref{thm2} and the existence condition (\ref{jutixuanqu}), one can deduce a family of
analytical vector rational RW solutions for the $n$-NLS equation (\ref{n-NLS}).

\begin{theorem}\label{dingli-RW}
The formula for the first-order vector rational rogue wave solutions $\mathbf{q}^{[1]}=(q_1^{[1]}(x,t), \cdots, q_n^{[1]}(x,t))^{\rm T}$ of the $n$-NLS equation (\ref{n-NLS}) is found as
\begin{gather}\label{RW}
q_j^{[1]}(x,t)=q_j^{\mathrm{bg}}\left(1-\frac{2\mathrm{i}\left(n+1\right)\zeta}{a_j}
\frac{\left(\mathbf{A}_{j+1}\bsm{\beta}\right)\left(\mathbf{A}_1\bsm{\beta}\right)^\dagger}{\left(\mathbf{A}\bsm{\beta}\right)^\dagger
\left(\mathbf{A}\bsm{\beta}\right)}\right),\quad j=1,\ldots,n,
\end{gather}
and then the intensity of the vector potential $\mathbf{q}^{[1]}$ is
\begin{gather}\label{RW-mf}
\|\mathbf{q}^{[1]}\|_2^2=\left\|\mathbf{a}\right\|_2^2+\frac{\partial^2}{\partial x^2} \ln \left\|\mathbf{A}\bsm{\beta}\right\|_2^2,
\end{gather}
where the matrix polynomial $\mathbf{A}(x,t)$ is defined as
\begin{gather}
\label{AB}
\mathbf{A}(x,t)=\sum_{s=0}^n\sum_{k=0}^{\lfloor n/2 \rfloor}\frac{\mathrm{i}^{s+k}\mathbf{B}^{s+2k}}{2^k\, s!\, k!}\left(x+z_0t\right)^s t^k,
\end{gather}
with $\lfloor\cdot\rfloor$ standing for the integer part and ${\bf B}$ given by Eq.~(\ref{BA}), and the constant column vector $\bsm{\beta}$ is linearly independent with
\begin{gather*}
\bsm{\xi}_0=\left(1, \frac{a_1}{z_0+b_1}, \frac{a_2}{z_0+b_2}, \cdots, \frac{a_n}{z_0+b_n}\right)^\mathrm{T},
\end{gather*}
\end{theorem}
\begin{proof}
Since $z_0-\lambda_0$ is the eigenvalue of $\mathbf{H}_0$, hence $\mathrm{det}\,(\mathbf{B})=\mathrm{det}(\mathbf{H}_0-\left(z_0-\lambda_0\right) \mathbb{I}_{n+1})=0$, which infers $\mathrm{rank}\,\mathbf{B}\le n$. From the fact $\mathrm{det}\left(z_0\mathbb{I}_n+\mathbf{b}\right)\ne 0$, one deduces $\mathrm{rank}\,\mathbf{B}\ge n$. Therefore, one can determine $\mathrm{rank}\,\mathbf{B}=n$ such that
one knows $\mathrm{rank}\,\mathbf{B}^k=n+1-k\, (1\le k\le n+1)$.
Then one can obtain the fundamental matrix solution for (\ref{constant}) with $\lambda=\lambda_0$ as
\begin{gather}
\Psi_p=\mathbf{A}\mathrm{e}^{\mathrm{i}\varpi},\quad
\varpi=\left(z_0-\lambda_0\right)\left(x+\lambda_0 t\right)\left[\frac{1}{2}\left(z_0-\lambda_0\right)^2-\left(\sum_{j=1}^na_j^2+\frac{\lambda_0^2}{2}\right)\right]t.
\end{gather}
 Then the Lax pair ~(\ref{laxs}) and (\ref{laxt}) with $\mathbf{Q}=\mathbf{Q}\big|_{{\bf q}={\bf q}^{\rm bg}}$ and $\lambda=\lambda_0$ has the fundamental solution
\begin{gather}\label{Phi-p}
\Phi_p=\mathbf{GA}\bsm{\beta}\,\mathrm{e}^{\mathrm{i}\varpi}.
\end{gather}
 From the definition of $\mathbf{B}$, one can obtain $\mathbf{B}\bsm{\beta}_0=\mathbf{0}$. Since $\mathrm{rank}\,\mathbf{B}=n$, then as $\bsm{\beta}$ is linearly independent with $\bsm{\beta}_0$, $\mathbf{A}\bsm{\beta}$ is the polynomial of at least degree 1.   Recalling the Darboux transformation for the $n$-NLS equation, the formula (\ref{RW}) follows.

Define 
the Darboux transformation $\mathbf{T}_p$
\begin{gather}
\label{darm}
\mathbf{T}_p:=\mathbb{I}_{n+1}+\frac{\bsm{\mathcal{L}}}{\lambda-\lambda_0^*},\quad \bsm{\mathcal{L}}:=\frac{-\mathrm{i}\left(n+1\right)\zeta\left(\Phi_p\bsm{\beta}\right)
\left(\Phi_p\bsm{\beta}\right)^\dag}{\left(\Phi_p\bsm{\beta}\right)^\dag\left(\Phi_p\bsm{\beta}\right)}.
\end{gather}
By matching the term $\mathcal{O}\left(\lambda_0^{-1}\right)$ in $\mathbf{T}_{p, x}+\mathbf{T}_{p}\mathbf{U}^{\mathrm{bg}}=\mathbf{U}^{[1]}\mathbf{T}_p$, one yields
$\bsm{\mathcal{L}}_x=\mathrm{i}\left(\mathbf{Q}^{[1]}\mathbf{A}-\mathbf{A}\mathbf{Q}^\mathrm{bg}\right)$.
By matching the term $\mathcal{O}\left(1\right)$ in $\mathbf{T}_{p, t}+\mathbf{T}_{p}\mathbf{V}^{\mathrm{bg}}=\mathbf{V}^{[1]}\mathbf{T}_p$, one yields
$\Big(\mathbf{Q}^{[1]}\Big)^2=\Big(\mathbf{Q}^\mathrm{bg}\Big)^2+\mathrm{i}\Big(\mathbf{Q}^\mathrm{bg}-\mathbf{Q}^{[1]}\Big)_x+2\sigma_3\Big(\mathbf{Q}^{[1]}\bsm{\mathcal{L}}-\bsm{\mathcal{L}}\mathbf{Q}^\mathrm{bg}\Big)$. Then one derives that
$\|\mathbf{q}^{[1]}\|_2^2=\left\|\mathbf{a}\right\|_2^2-2\mathrm{i}\left(\bsm{\mathcal{L}}_x\right)_{1, 1}$ and  $|q_j^{[1]}|^2=\left|a_j\right|^2+2\mathrm{i}\left(\bsm{\mathcal{L}}_x\right)_{j+1, j+1}$. Together with $\left(\mathbf{\varPhi}_p^\dag\mathbf{\varPhi}_p\right)=-\left(n+1\right)\zeta\mathbf{\varPhi}_p^\dag\sigma_3\mathbf{\varPhi}_p$, then the formula (\ref{RW-mf}) follows. This completes the proof.
\end{proof}

\begin{remark} The analytical vector rational RW solutions (\ref{dingli-RW}) of the $n$-NLS system (\ref{n-NLS}) were never reported for ${\bf H}$ with an $(n+1)$-multiple root before. Besides, by the intensity expression of the vector potential $\mathbf{q}^{[1]}$, one can pose the mass conservation law
\begin{gather}
\int_{-\infty}^{+\infty}\left(\|\mathbf{q}^{[1]}\|_2^2-\left\|\mathbf{a}\right\|_2^2\right)\mathrm{d}x=0.
\end{gather}
\end{remark}

\subsection{${\cal PT}$-symmetric vector rational RW solutions}

\begin{prop}\label{prop-PT}
The vector rational RW $\mathbf{q}^{[1]}(x,t)$ given by Eq.~(\ref{RW}) is of the parity-time-reversal (${\cal PT}$) symmetric structure:
\begin{gather}\label{PT}
\mathbf{q}^{[1]}(x, t)={\cal PT}\mathbf{q}^{[1]}(x, t)={\cal P}\mathbf{q}^{[1]}(x, -t)^*,
\end{gather}
if $\Re(\lambda_0)=0$ and $\bsm{\beta}=\mathbf{J}\bsm{\beta}^*$, where the matrix $\mathbf{J}$ and the parity operator ${\cal P}$ are defined as
\begin{gather*}
\mathbf{J}=
\begin{pmatrix}
1& \mathbf{0}_{1\times n} \vspace{0.1in}\\
\mathbf{0}_{n\times 1} & -\mathcal{P}
\end{pmatrix}, \quad
{\cal P}=
\begin{pmatrix}
&&1\\
&\iddots&\\
1&&
\end{pmatrix}_{n\times n},\quad {\cal P}^2=\mathbb{I}_{n},
\end{gather*}
and the time-reversal operator ${\cal T}$ is:\, $t\to -t,\,\ {\rm i}\to -{\rm i}$.
\end{prop}

\begin{proof}
As $\Re(\lambda_0)=0$, it is easy to know that the plane-wave solution (\ref{jutixuanqu}) is parity-time-reversal symmetric, i.e. $\mathbf{q}^{\mathrm{bg}}(x, t)={\cal P}\mathbf{q}^\mathrm{bg}(x, -t)^*$,
from which one derives the symmetries of $\mathbf{U}^\mathrm{bg}(\lambda; x, t)=\mathbf{U}(\lambda; x, t)\big|_{{\bf Q}={\bf Q}^\mathrm{bg}}$ and $\mathbf{V}^\mathrm{bg}(\lambda; x, t)=
\mathbf{V}(\lambda; x, t)\big|_{{\bf Q}={\bf Q}^\mathrm{bg}}$ as
\begin{gather*}
\mathbf{U}^\mathrm{bg}(\lambda; x, t)=\mathbf{J}\mathbf{U}^\mathrm{bg}(-\lambda^*; x, -t)^*\mathbf{J},  \quad \mathbf{V}^\mathrm{bg}(\lambda; x, t)=-\mathbf{J}\mathbf{V}^\mathrm{bg}(-\lambda^*; x, -t)^*\mathbf{J}.
\end{gather*}

Since $\Phi_p(\lambda_0; x, t)$ in Eq. \eqref{Phi-p} is the fundamental vector solution of
\begin{gather}\label{initial-1}
\Phi_x=\mathbf{U}^\mathrm{bg}(\lambda_0; x, t)\Phi, \quad \Phi_t=\mathbf{V}^\mathrm{bg}(\lambda_0; x, t)\Phi
\end{gather}
thus so is $\mathbf{J}\Phi_p(-\lambda_0^*; x, -t)^*$.
From $\bsm{\beta}=\mathbf{J}\bsm{\beta}^*$, one obtains $\Phi_p(\lambda_0; 0, 0)=\mathbf{J}\Phi_p(-\lambda_0^*; 0, 0)^*$. It follows from the uniqueness of the solution of Eq.~(\ref{initial-1}) that
$\Phi_p(\lambda_0; x, t)=\mathbf{J}\Phi_p(-\lambda_0^*; x, -t)^*$.
In particular, one has $\Phi_p(\lambda_0; x, t)=\mathbf{J}\Phi_p(\lambda_0; x, -t)^*$ for $\Re(\lambda_0)=0$. Then the Darboux matrix $\mathbf{T}_p$ given by Eq.~(\ref{darm}) has the symmetry
$\mathbf{T}_p(\lambda; x, t)=\mathbf{J}\mathbf{T}_p(-\lambda^*; x, -t)^*\mathbf{J}$.
By the relation between the initial $\left(\mathbf{U}^\mathrm{bg}, \mathbf{V}^\mathrm{bg}\right)$ and transformed $\left(\mathbf{U}^{[1]}, \mathbf{V}^{[1]}\right)$
\begin{gather*}
\mathbf{U}^{[1]}(\lambda; x, t)=\left(\mathbf{T}_{p, x}(\lambda; x, t)+\mathbf{T}_p(\lambda; x, t)\mathbf{U}^\mathrm{bg}(\lambda; x, t)\right)\mathbf{T}_p(\lambda; x, t)^{-1}, \\
\mathbf{V}^{[1]}(\lambda; x, t)=\left(\mathbf{T}_{p, t}(\lambda; x, t)+\mathbf{T}_p(\lambda; x, t)\mathbf{V}^\mathrm{bg}(\lambda; x, t)\right)\mathbf{T}_p(\lambda; x, t)^{-1},
\end{gather*}
then one deduces
\begin{gather*}
\mathbf{U}^{[1]}(\lambda; x, t)=\mathbf{J}\mathbf{U}^{[1]}(-\lambda^*; x, -t)^*\mathbf{J}, \quad \mathbf{V}^{[1]}(\lambda; x, t)=-\mathbf{J}\mathbf{V}^{[1]}(-\lambda^*; x, -t)^*\mathbf{J},
\end{gather*}
which imply that Eq.~(\ref{PT}) holds. Thus we complete the proof.
\end{proof}

\begin{remark} In fact if $\mathbf{q}(x, t)$ is a solution of the $n$-NLS equation (\ref{n-NLS}), then so is $\mathbf{q}(x+2\Re(\lambda_0)t, t)\mathrm{e}^{-2\mathrm{i}\Re(\lambda_0)[x+\Re(\lambda_0)t]}$.  Without loss of generality, one can take $\Re\left(\lambda_{0}\right)=0$ hereafter.
\end{remark}


\subsection{The classification of vector rational RW solutions and dynamics}

The explicit formula (\ref{RW}) for the vector rational RWs are in fact the rational forms and $\mathbf{A}\bsm{\beta}$ is a polynomial of degree less than $n+1$. 
Given a vector parameter $\bsm{\beta}$, the degree of the polynomial $\mathbf{A}\bsm{\beta}$ can be $1$, $2$, $\cdots$, or $n$. In the following,
we will give a complete classification of the rational RW solution (\ref{RW}) by the degree of the polynomial $\mathbf{A}\bsm{\beta}$.

\begin{prop}\label{classification}
Let the matrix $\mathbf{S}=\left(\bsm{\xi}_0, \bsm{\xi}_1, \cdots, \bsm{\xi}_n\right)$ and $\bsm{\alpha}=(\alpha_0, \alpha_1, \cdots, \alpha_n)=\mathbf{S}^{-1}\bsm{\beta}$, where
\begin{gather}
\bsm{\xi}_j=\left(0, \frac{a_1}{\left(\mathrm{i}\zeta+b_1\right)^{j+1}}, \frac{a_2}{\left(\mathrm{i}\zeta+b_2\right)^{j+1}}, \cdots, \frac{a_n}{\left(\mathrm{i}\zeta+b_n\right)^{j+1}}\right)^\mathrm{T}, \quad j=1, 2, \cdots, n.
\end{gather}
Then the complete classification of the vector rational RW solution (\ref{RW}) can be determined  by the degree of the polynomial $\mathbf{A}\bsm{\beta}$, where $\mathbf{A}\bsm{\beta}$ is the polynomial of degree $j$ with respect to the variables $x$ and $t$ as $\alpha_j\ne 0$, $\alpha_{j+1}=\cdots=\alpha_n=0$\, ($1\le j\le n$).

\end{prop}

\begin{proof}
The vector parameter $\bsm{\alpha}$ is well-defined due to
\begin{gather*}
\mathrm{det}\,\mathbf{S}=\left(-\frac{1}{\zeta}\right)^{\frac{n\left(n+1\right)}{2}}\prod_{j=1}^n\left(\sin\omega_j\right)^j\ne 0.
\end{gather*}
Since $\mathrm{i}\zeta$ is the $(n+1)$-multiple root of $D(z)=0$, then one can deduce that
\begin{gather*}
\mathrm{i}\zeta-2\lambda_0-\sum_{j=1}^n\frac{a_j^2}{\mathrm{i}\zeta+b_j}=0, \quad 1+\sum_{j=1}^n\frac{a_j^2}{\left(\mathrm{i}\zeta+b_j\right)^2}=0, \quad \sum_{j=1}^n\frac{a_j^2}{\left(\mathrm{i}\zeta+b_j\right)^{k+1}}=0, \quad 2\le k\le n,
\end{gather*}
which derives that
$\mathbf{B}^j\bsm{\xi}_j \ne \mathbf{0}$ and $\mathbf{B}^{j+1}\bsm{\xi}_j = \mathbf{0}\,  (1\le j\le n)$.
Then we complete the proof.
\end{proof}

\begin{remark}
Proposition~\ref{classification} confirms the complete classification of the vector raiotnal RWs, and gives a rule to determine the degree of the polynomial $\mathbf{A}\bsm{\beta}$.
 \end{remark}

\begin{remark}
Give a parameter $\bsm{\beta}$, one can always confirm the degree of $\mathbf{A}\bsm{\beta}$ judged by $\bsm{\alpha}$. Besides,  the symmetric condition  $\bsm{\beta}=\mathbf{J}\bsm{\beta}^*$ is equivalent to $\alpha_j=(-1)^j\alpha_j^*$. The linear independence with $\bsm{\xi}_0$ leads to $\sum_{j=1}^n\left|\alpha_j\right|^2\ne 0$.
\end{remark}

For convenience,  we can redenote the first-order vector rational RWs (\ref{RW}) as $q_j^{{[1]}_\ell}$ for the $\mathbf{A}\bsm{\beta}$ of degree $\ell$ ($1\le\ell\le n$). Then we classify the general first-order vector RWs (\ref{RW}) into the $n$ cases: $q_j^{[1]_1}, q_j^{[1]_2}, \cdots, q_j^{[1]_n}$.
In the next example, we will confirm each component $q_j^{[1]_1}$ is a kind of the fundamental RW and the types of the structures are controlled  by $\omega_j$. For convenience, we introduce $n$ rational functions
\begin{gather}
R_j\left(x, t\right)=\frac{1}{a_j^2}\frac{2\mathrm{i}(b_jx-\zeta^2t)-1}{x^2+\zeta^2t^2+\zeta^2/4}, \quad j=1, 2, \cdots, n.
\end{gather}
and an $\left(n+1\right)\times\left(n+1\right)$  Jordan block  $\mathbf{J}_\mathrm{or}$
\begin{gather}
\mathbf{J}_\mathrm{or}=
\begin{pmatrix}
0& 1 &  & &  \\
 & 0 & 1 & & \\
  &  &   \ddots & \ddots &\\
  &  &      &  0   & 1  \\
  &&&&0
\end{pmatrix}.
\end{gather}

{\it Example 3.2.1.} \, For the simplest case: $\mathbf{q}^{[1]_1}$, without loss of generality, we choose $\alpha_1=\mathrm{i}$ and $\alpha_j=0\, (2\leq j\leq n)$ such that $\mathbf{A}\bsm{\beta}=\left(x+\mathrm{i}\zeta t+\alpha_0\right)\bsm{\xi}_0+\mathrm{i}\bsm{\xi}_1$ is a linear polynomial.  Then one can simplify $q_j^{[1]_1}$ as
\begin{gather}\label{1-RW}
q_j^{[1]_1}=q_j^\mathrm{bg}\left[1+R_j\left(x+\alpha_0+\frac{1}{2\zeta}, t\right)\right]\mathrm{e}^{-2\mathrm{i}\omega_j}, \quad j=1, 2, \cdots, n.
\end{gather}
It follows from the formula (\ref{1-RW}) that one can find that each component $q_j^{[1]_1}$ is the fundamental RW with the center located at $(x,t)=\left(-\alpha_0-1/(2\zeta),\,0\right)$.

\begin{remark}
It follows from Eq.~(\ref{abr}) that $a_j^2-b_j^2=\zeta^2$ for $\Re(\lambda_0)=0$, thus the formula (\ref{1-RW}) coincides with one in Ref.~\cite{Ling2017}, but Ref.~\cite{Ling2017} did not find an explicit multiple root $z_0$ for the general case $n\geq 4$. We here first find an explicit $(n+1)$-multiple root $z_0$ given by Eq.~(\ref{BA}), which is a critical point to find vector rational RWs of the $n$-NLS system.
\end{remark}


Based on the classification of fundamental RW types~\cite{Ling2017}, for the case of $(n+1)$-multiple root,
the dynamical structures of $q_j^{[1]_1}$
can be determined by $\omega_j$ as follows:
\begin{itemize}
\item As $\omega_j\in\left[\pi/3, 2\pi/3\right]$, the $j$-th component $q_j^{[1]_1}$ is eye-shaped RW (alias bright RW) with one hump and two valleys;

\item As $\omega_j\in\left(\pi/6, \pi/3\right)\bigcup\left(2\pi/3, 5\pi/6\right)$, the $j$-th component $q_j^{[1]_1}$ is four-petaled RW with
two humps and two valleys;

\item As $\omega_j\in\left(0, \pi/6\right] \bigcup \left[5\pi/6, \pi\right)$, the $j$-th component $q_j^{[1]_1}$ is anti-eye-shaped RW (alias dark RW) with two humps and one valley.

    \end{itemize}

\begin{figure}[!t]
\centering
\includegraphics[scale=0.4]{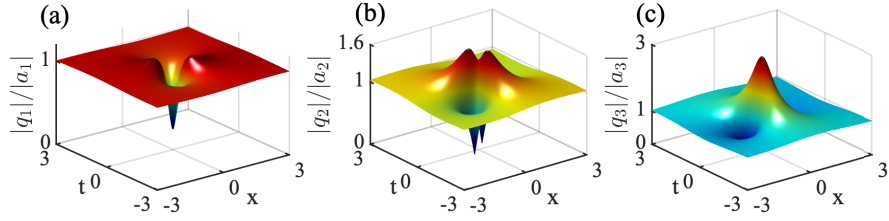}
\caption{Profiles of fundamental RW solutions of the $6$-NLS equation with $\zeta=1, \alpha_0=-1/2$ in Eq. (\ref{1-RW}).}
\label{RW1}
\end{figure}

With the judgement, one can determine the fundamental RW type of each component $q_j^{[1]_1}$ by $\omega_j$. For instance (based on the ${\cal PT}$ symmetry of the vector RWs in Proposition \ref{prop-PT}),

\begin{itemize}

\item{} For $n=1$, the fundamental RW for the scalar NLS equation is always the bright RW due to $\omega_1=\pi/2$;

\item{} For $n=2$, two fundamental RWs for the 2-NLS equation are both bright RWs due to $\omega_1=\pi/3,\, \omega_2=2\pi/3$;

\item{} For $n=3$, the fundamental RWs for $q_1^{[1]_1}$ and $q_3^{[1]_1}$ are both four-petaled RWs due to $\omega_1=\pi/4,\, \omega_3=3\pi/4$, and
    $q_2^{[1]_1}$ is the bright RW due to $\omega_2=\pi/2$;

\item{} For $n=4$, the fundamental RWs for $q_1^{[1]_1}$ and $q_4^{[1]_1}$ are both  four-petaled RWs due to $\omega_1=\pi/5,\, \omega_3=4\pi/5$, and
    $q_j^{[1]_1}$'s are all the bright RWs due to $\omega_j=j\pi/5\, (j=2,3,4)$;

\item{} For $n=5$, the fundamental RWs for $q_1^{[1]_1}$ and $q_5^{[1]_1}$ are both the dark RWs due to $\omega_1=\pi/6,\, \omega_5=5\pi/6$, and
    $q_j^{[1]_1}$'s are all the bright RWs due to $\omega_j=j\pi/6\, (j=2,3,4,5)$;

\item{} For $n\ge 6$ the fundamental RWs for the $n$-NLS equation admit three types of RWs. For example as $n=6$, $q_1^{[1]_1}$ and $q_6^{[1]_1}$ are both the dark RWs due to $\omega_1=\pi/7,\, \omega_6=6\pi/7$ (Fig. \ref{RW1}a),
$q_2^{[1]_1}$ and $q_5^{[1]_1}$ are both the four-petaled RWs due to $\omega_2=2\pi/7,\, \omega_5=5\pi/7$ (Fig. \ref{RW1}b), and
$q_3^{[1]_1}$ and $q_4^{[1]_1}$ are both the bright RWs due to $\omega_3=3\pi/7,\, \omega_4=4\pi/7$ (Fig. \ref{RW1}c).

\end{itemize}

\subsection{The asymptotics of vector rational RW solutions}

From the above propositions, the degree of the polynomial $\mathbf{A}\bsm{\beta}$ is decided by $\bsm{\alpha}$ and as $\mathbf{A}\bsm{\beta}$ is the polynomial of degree $1$, each component $q_j^{[1]_1}$ is a fundamental rogue wave. As the degree of $\mathbf{A}\bsm{\beta}$ is greater than $1$, each component  is no longer a fundamental rogue wave, which has complicated and abundant dynamical structures. In the following  two examples, we study the asymptotic behaviors of the rogue wave $q_j^{[1]_\ell}$ with  $\mathbf{A}\bsm{\beta}$ of degree $\ell$ ($\ge 2$). As some parameters tend to $\pm\infty$, the dynamical behavior of each component $q_j^{[1]_\ell}$ with  $\mathbf{A}\bsm{\beta}$ of degree $\ell$ ($\ge 2$) will be characterized by the leading term of asymptotic formula.

From the definition of $\mathbf{A}(x,t)$ in Eq. \eqref{AB}, it follows that $\mathbf{A}$ can be written as the form
\begin{gather}
\mathbf{A}(x,t)=\mathbb{I}_{n+1}-\tau_1(x,t)\mathbf{B}+\tau_2(x,t)\mathbf{B}+\cdots+\left(-1\right)^n \tau_n(x,t)\mathbf{B}^n,
\end{gather}
where $\tau_s(x,t)$'s are the polynomials of degree $s$ with respect to $x$ and $t$. By the straightforward computation,
then the formula of rogue wave solution $q_j^{[1]_\ell}$ in Eq. \eqref{RW} can be rewritten as
\begin{gather}\label{form-gai}
q_j^{[1]_\ell}=q_j^\mathrm{bg}\left[1-\frac{2\mathrm{i}\zeta}{b_j-\mathrm{i}\zeta}\frac{\left(n+1\right)L_jL_0^*-\mathbf{L}^\dag \mathbf{L}}{\mathbf{L}^\dag \mathbf{L}}\right]\mathrm{e}^{-2\mathrm{i}\omega_j},\quad j=1, 2, \cdots, n,
\end{gather}
where \, $\bsm{\alpha}=\left(\alpha_0, \alpha_1, \cdots, \alpha_\ell, 0, \cdots, 0\right)^\mathrm{T}$, \, $\mathbf{v}_\ell=\left(1, \tau_1, \tau_2, \cdots, \tau_\ell, 0, \cdots, 0\right)$  and
\begin{gather}\label{L}
\mathbf{L}=\sum_{s=0}^\ell\mathbf{v}_\ell\mathbf{J}_\mathrm{or}^s\bsm{\alpha}\bsm{\xi}_s,   \quad
L_j=\sum_{s=0}^\ell\frac{\mathbf{v}_\ell\mathbf{J}_\mathrm{or}^s\bsm{\alpha}}{\left(b_j+\mathrm{i}\zeta\right)^s}, \quad
L_0=\mathbf{v}_\ell\bsm{\alpha}.
\end{gather}
Note that $\bsm{\xi}_0^\dag\bsm{\xi}_0=n+1$. It follows from the form \eqref{form-gai} that  $q_j^{[1]_\ell}/q_j^\mathrm{bg}$ tends to $\mathrm{e}^{-2\mathrm{i}\omega_j}$ as  $|x|\to \infty$ or $|t|\to \infty$.

In the following, we study the asymptotic behaviors of vector rational RWs $\mathbf{q}^{[1]_\ell}$ with $\ell=2, 3, 4$, whose explicit formulas \eqref{form-gai}  are determined by Eq. \eqref{L} with
\begin{gather}\label{VL}
\tau_1= \mathrm{i}\zeta t-{\rm i}x, \,\, \tau_2=\frac{\mathrm{i}t-\left(x+\mathrm{i}\zeta t\right)^2}{2}, \,\, \tau_3=\frac{\left(x+\mathrm{i}\zeta t\right)t}{2}\!+\!\frac{\mathrm{i}\left(x+\mathrm{i}\zeta t\right)^3}{6}, \,\, \tau_4=\frac{\left(x\!+\!\mathrm{i}\zeta t\right)^4}{24}\!-\!\frac{2\mathrm{i}\left(x\!+\!\mathrm{i}\zeta t\right)^2t\!+\!t^2}{8}.
\end{gather}

\vspace{0.1in}
{\it Example 3.3.1.\, } For the case $\ell=2$, to study the asymptotic behaviors of $\mathbf{q}^{[1]_2}$, we take
\begin{gather}
\alpha_0=\kappa_0h^2+\frac{\kappa_{0, 1}h+\kappa_{0, 0}}{\zeta},\quad \alpha_1=\mathrm{i}\left(\kappa_1h+\frac{\kappa_{1, 0}}{\zeta}\right),\quad \alpha_2=-2,\quad \kappa_0, \kappa_1,\kappa_{0, 1}, \kappa_{0, 0}, \kappa_{1, 0}, h\in\mathbb{R}.
\end{gather}
Under the transformations $x=\widehat{x}+\Re\left(z\right)h,\, t=\widehat{t}+\Im\left(z\right)h/\zeta,\, z\in\mathbb{C}$, we derive
\begin{gather}
\frac{\mathbf{L}^\dag \mathbf{L}}{\left(n+1\right)h^4}=|z^2+\kappa_1z+\kappa_0|^2+\mathcal{O}\left(\frac{1}{h}\right), \quad h\to \infty,
\end{gather}
in which we call
\bee\label{f2}
\mathcal{F}_2(z)=z^2+\kappa_1z+\kappa_0
\ene
 the {\it governing polynomial}, whose roots are related to the central locations of separated fundamental RWs. Therefore we can control the central locations of separated RWs by the arbitrariness of $\kappa_0, \kappa_1$. Since  $\mathcal{F}_2(z)$ is a polynomial of degree $2$, it allows two simple roots or a double root.  For the two case,  we can analyze the asymptotic behaviors of rogue waves.

 \begin{figure}[!t]
\centering
\includegraphics[scale=0.4]{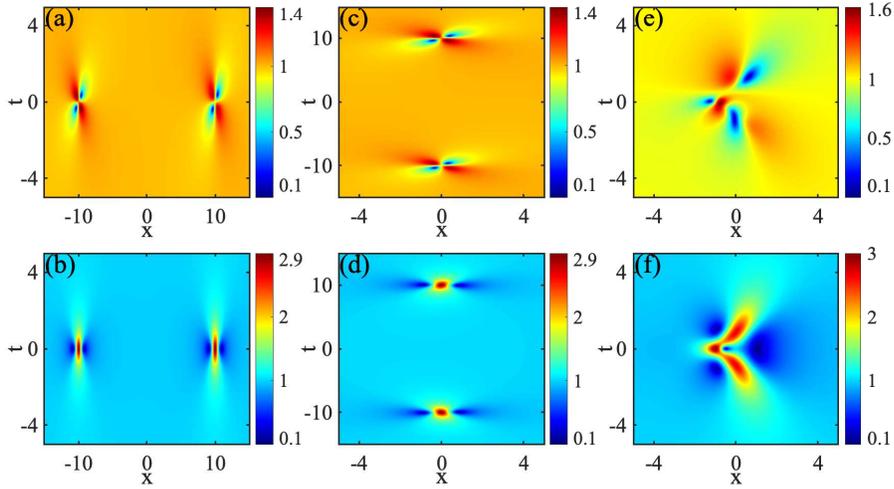}
\caption{Rogue wave solutions    $q_1^{[1]_2}/a_1$ (a, c, e) and  $q_2^{[1]_2}/a_2$ (b, d, f)  of $3$-NLS equation in Eqs.
\eqref{form-gai}-\eqref{VL} with parameters:  (a, b) $\zeta=1, \alpha_0=-100, \alpha_1=-\mathrm{i}, \alpha_2=-2$; (c, d)  $\zeta=1, \alpha_0=100, \alpha_1=0, \alpha_2=-2$; (e, f) $\zeta=1, \alpha_0=\alpha_1=0, \alpha_2=-2$.}
\label{RW2}
\end{figure}

\vspace{0.05in}
{\it Case  1.} For the two simple roots of $\mathcal{F}_2(z)$ and the corresponding two fundamental RWs arranging along the $x$-axis.
For instance as $\kappa_0=-1, \kappa_1=0$, the governing equation $\mathcal{F}_2(z)=0$ has two simple roots $z=\pm 1$. With the transformations $x=\widehat{x}+\left(-1\right)^\delta h,\, t=\widehat{t}$\,  ($\delta=0, 1)$, we derive the limitations
 \begin{gather}
 \lim_{h\to \infty}q_j^{[1]_2}=q_j^\mathrm{bg}\left[1+R_j\left(\widehat{x}+\frac{\left(-1\right)^\delta\kappa_{0, 0}+\kappa_{1, 0}+1}{2\zeta},\, \widehat{t}\right)\right]\mathrm{e}^{-2\mathrm{i}\omega_j}.
 \end{gather}
 Then we deduce the formula of  asymptotic behaviors
\begin{gather}
q_j^{[1]_2}=q_j^\mathrm{bg}\left[1+\sum_{\delta=0}^1R_j\left(x-x_\delta, t\right)\right]\mathrm{e}^{-2\mathrm{i}\omega_j}
+\mathcal{O}\left(\frac{1}{h}\right), \,\,\,\, h\to \infty,
\end{gather}
where $x_\delta=\left(-1\right)^\delta\left(h-\kappa_{0, 0}/(2\zeta)\right)-\left(\kappa_{1, 0}+1\right)/(2\zeta)$. The leading term of asymptotic formula is the linear superposition of two fundamental RWs, which can be used to characterize the dynamical behaviors. As the value of $h$ increases, each component $q_j^{[1]_2}$ will be separated into two fundamental RWs with centers located at $\left(x_\delta, 0\right)$. In particular, as $\kappa_{0, 0}=0, \kappa_{1, 0}=-1$, the centers of two fundamental RWs are located at $\left(\pm h, 0\right)$, respectively (Figs. \ref{RW2}(a, b)).

\vspace{0.05in}
{\it Case  2.} For the two simple roots of $\mathcal{F}_2(z)$ and corresponding two fundamental RWs arranging along the $t$-axis.
As $\kappa_0=1, \kappa_1=0$, the governing equation $\mathcal{F}_2(z)=0$ has the roots $z=\pm \mathrm{i}$. With the transformation $x=\widehat{x}, \, t=\widehat{t}+\left(-1\right)^\delta h/\zeta$\, ($\delta=0, 1$), we derive the limitations
 \begin{gather}
 \lim_{h\to+\infty}q_j^{[1]_2}=q_j^\mathrm{bg}\left[1+R_j\left(\widehat{x}+\frac{\kappa_{1, 0}}{2\zeta}, \widehat{t}-\frac{\left(-1\right)^\delta\kappa_{0, 0}}{2\zeta^2}\right)\right]\mathrm{e}^{-2\mathrm{i}\omega_j}.
 \end{gather}
 Then we deduce the formula of  asymptotic behavior
\begin{gather}
q_j^{[1]_2}=q_j^\mathrm{bg}\left[1+\sum_{\delta=0}^1R_j\left(x-x_\delta, t-t_\delta\right)\right]\mathrm{e}^{-2\mathrm{i}\omega_j}
+\mathcal{O}\left(\frac{1}{h}\right), \,\, h\to +\infty,
\end{gather}
where $x_\delta=-\kappa_{1, 0}/(2\zeta), \, t_\delta=\left(-1\right)^\delta\left(h+\kappa_{0, 0}/(2\zeta)\right)/\zeta$. The leading term of the formula of asymptotic behavior is the linear superposition of two fundamental RWs, located at $\left(x_\delta, t_\delta\right)$, which can characterize the dynamical structure. In particular, as $\kappa_{1, 0}=\kappa_{0, 0}=0$, two fundamental RWs are located at $\left(0, \pm h/\zeta\right)$ (see Figs. \ref{RW2}(c, d)).

\vspace{0.05in} {\it Case  3.} For a double root of $\mathcal{F}_2(z)$, to conveniently express the leading term of asymptotic formula, we define
\begin{gather}
p_{j, 2}(x,t)=-\frac{2\mathrm{i}\zeta}{b_j-\mathrm{i}\zeta}\frac{\left(n+1\right)L_jL_0^*-\mathbf{L}^\dag \mathbf{L}}{\mathbf{L}^\dag \mathbf{L}},\quad j=1, 2, \cdots, n,
\end{gather}
where $\mathbf{L}, L_j's, L_0$ are defined in Eqs. \eqref{L} and \eqref{VL} with $\ell=2$. In particular, as $\kappa_0=\kappa_1=0$, the governing polynomial $\mathcal{F}_2=z^2$ has a double root $z=0$. To obtain the rogue wave solution in the leading term of the asymptotic formula, $\kappa_{0, 0}=0$ is posed.  Then
\bee
q_j^{[1]_2}=q_j^\mathrm{bg}\left(1+p_{j, 2}(x,t)\right)\mathrm{e}^{-2\mathrm{i}\omega_j}
\ene
with $\alpha_0=\kappa_{0, 1}, \alpha_1=\mathrm{i}\kappa_{1, 0}, \alpha_2=-2\zeta$, which display the strong interaction (Figs. \ref{RW2}(e, f)). Similarly, for the other double roots of $\mathcal{F}_2$, we also have the formula of asymptotic behavior as $h\to +\infty$. For instance,  as $\kappa_0=1, \kappa_1=-2$, the governing polynomial $\mathcal{F}_2=\left(z-1\right)^2$, which has a double root $z=1$. Under the constraint $\kappa_{0, 0}$, we derive the formula of asymptotic behaviors of the RWs
\begin{gather}
q_j^{[1]_2}=q_j^\mathrm{bg}\left[1+p_{j, 2}\left(x-h, t\right)\right]\mathrm{e}^{-2\mathrm{i}\omega_j}+\mathcal{O}\left(\frac{1}{h}\right),\,\, h\to +\infty,
\end{gather}
where  the parameters of $p_{j, 2}$ are $\alpha_0=\kappa_{0, 1},\, \alpha_1=\mathrm{i}\kappa_{1, 0},\, \alpha_2=-2\zeta$.

\vspace{0.1in} {\it Example 3.3.2.} For the case $\ell=3$, to study the asymptotic behaviors of rogue waves of $\mathbf{q}^{[1]_3}$, we take
\begin{gather}
\begin{gathered}
\alpha_0=\kappa_0h^3+\sum_{s=0}^2\kappa_{0, s}h^{2-s}/\zeta,\quad \alpha_1=\mathrm{i}\kappa_1h^2+\sum_{s=0}^1\kappa_{1, s}h^{1-s}/\zeta,            \quad
 \alpha_2=-2\kappa_2h-2\kappa_{2, 0}/\zeta,\quad \alpha_3=-6\mathrm{i},
 \end{gathered}
\end{gather}
with $\kappa_0,\kappa_{0, s}, \kappa_1, \kappa_{1, s}, \kappa_2, \kappa_{2, 0},  h\in\mathbb{R}$.

By the transformation $x=\widehat{x}+\Re\left(z\right)h, t=\widehat{t}+\Im\left(z\right)h/\zeta$, we derive the expansions
\begin{gather}
\frac{\mathbf{L}^\dag \mathbf{L}}{\left(n+1\right)h^6}=\left|z^3+\kappa_2z^2+\kappa_1z+\kappa_0\right|^2+\mathcal{O}\left(\frac{1}{h}\right), \,\,\,\, h\to +\infty,
\end{gather}
which generates the governing polynomial of degree $3$ as
 \begin{gather} \label{f3}
\mathcal{F}_3(z)=z^3+\kappa_2z^2+\kappa_1z+\kappa_0,
\end{gather}
which allows three cases: three simple roots, a double root and one simple root, and a triple root.

\vspace{0.05in}
{\it Case 1.} For the three simple roots of $\mathcal{F}_3(z)$, to obtain  the line-typed structure along $t$-axis, we take $\kappa_1=-1, \kappa_2=\kappa_0=0$, with which the governing polynomial $\mathcal{F}_3$ has three real simple roots $z=0, \pm 1$.  To obtain  the line-typed structure along $x$-axis, we take $\kappa_1=1, \kappa_2=\kappa_0=0$, with which the governing polynomial $\mathcal{F}_3$ has three simple roots $z=0, \pm \mathrm{i}$.    To obtain  the triangle-typed structure, we take $\kappa_0=1, \kappa_1=\kappa_2=0$, with which the governing polynomial $\mathcal{F}_3$ has three  simple roots $\mathrm{e}^{\left(2s+1\right)\pi\mathrm{i}/3}, s=0, 1, 2$.  Then we deduce the formula of  asymptotic behavior
\begin{gather}
q_j^{[1]_3}=q_j^\mathrm{bg}\left[1+\sum_{s=0}^2R_j\left(x-x_s, t-t_s\right)\right]\mathrm{e}^{-2\mathrm{i}\omega_j}
+\mathcal{O}\left(\frac{1}{h}\right), \,\, h\to +\infty.
\end{gather}
with the linear superposition of three single rogue waves located at $(x_s, t_s)$ in the leading term, where
\begin{itemize}
\item $x_\delta=\left(-1\right)^\delta\left(h-\kappa_{1, 0}/2\zeta\right)-\left(\kappa_{0, 0}+\kappa_{2, 0}+1\right)/2\zeta, \, t_\delta=0, x_2=(2\kappa_{0, 0}-1)/2\zeta, \,t_2=0, \, \delta=0, 1$ for  the case $\kappa_1=-1, \kappa_2=\kappa_0=0$ (see Figs. \ref{RW3}(a, b));
\item $x_\delta=(\kappa_{0, 0}-\kappa_{2, 0}+2)/2\zeta, \, t_\delta=(-1)^\delta(h+\kappa_{1, 0}/2\zeta)/\zeta, \, x_2=(2\kappa_{0, 0}+1)/2\zeta, \, t_2=0, \, \delta=0, 1$ for the case $\kappa_1=1, \kappa_2=\kappa_0=0$ (see Figs. \ref{RW3}(c, d));
\item $x_s=(h+(2\kappa_{0, 0}-2\kappa_{1, 0}+3)/6\zeta)\cos[(2s+1)\pi/3]-\kappa_{2, 0}/3\zeta, \, t_s=(h+(2\kappa_{0, 0}+2\kappa_{1, 0}+3)/6\zeta)\cos[(2s+1)\pi/3]/\zeta, \, s=0, 1, 2$ for the case $\kappa_0=1, \kappa_1=\kappa_2=0$ (see Figs. \ref{RW3}(e, f)).
\end{itemize}

\begin{figure}[!t]
\centering
\includegraphics[scale=0.4]{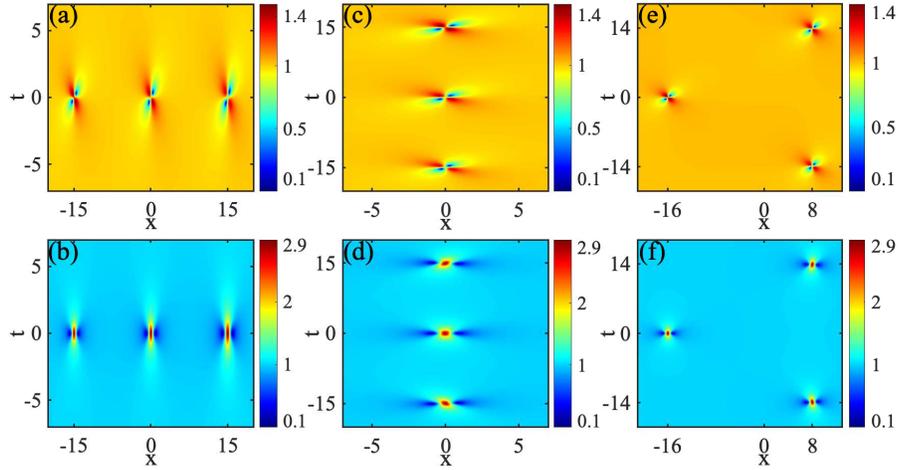}
\caption{Rogue wave solution $q_1^{[1]_3}/a_1$ (a, c, e) and  $q_2^{[1]_3}/a_2$ (b, d, f) of $3$-NLS equation in Eqs. \eqref{form-gai}, \eqref{L} and \eqref{VL} with parameters: (a, b) $\zeta=1, \alpha_0=225/2, \alpha_1=-225\mathrm{i}, \alpha_2=3, \alpha_3=-6\mathrm{i}$; (c, d)  $\zeta=1, \alpha_0=-225/2, \alpha_1=225\mathrm{i}, \alpha_2=-3, \alpha_3=-6\mathrm{i}$; (e, f) $\zeta=1, \alpha_0=-128, \alpha_1=\alpha_2=0, \alpha_3=-6\mathrm{i}$.}
\label{RW3}
\end{figure}

\vspace{0.05in}
{\it Case 2.} For the case of one simple root and one double root of $\mathcal{F}_3(z)$, we take $\kappa_2=1, \kappa_0=\kappa_1=0$, the governing polynomial $\mathcal{F}_3$ has a simple roots $z=-1$ and a double root $z=0$. The corresponding asymptotic formula is derived as \begin{gather}
q_j^{[1]_3}=q_j^\mathrm{bg}\left[1+p_{j, 2}\left(x, t\right)+R_j\left(x+h-\frac{2\kappa_{1, 0}-2\kappa_{2, 0}-1}{2\zeta}, t\right)\right]\mathrm{e}^{-2\mathrm{i}\omega_j}
+\mathcal{O}\left(\frac{1}{h}\right), \,\, h\to +\infty,
\end{gather}
where the parameters in $p_{j, 2}$ are $\alpha_0=-\kappa_{0, 1}/2\zeta, \alpha_1=-\mathrm{i}\kappa_{1, 0}/2\zeta$.

\vspace{0.05in}
{\it Case 3.}  For the case of a triple root of $\mathcal{F}_3(z)$, we take $\kappa_0=\kappa_1=\kappa_2=0$, the governing polynomial $\mathcal{F}_3$ has a triple root $z=0$. For the convenience of the representation of  the leading term of asymptotic formula, we define
\begin{gather}
p_{j, 3}=-\frac{2\mathrm{i}\zeta}{b_j-\mathrm{i}\zeta}\frac{\left(n+1\right)L_jL_0^*-\mathbf{L}^\dag \mathbf{L}}{\mathbf{L}^\dag \mathbf{L}},\quad j=1, 2, \cdots, n,
\end{gather}
where $\mathbf{L}, L_j's, L_0$ are defined in Eqs. \eqref{L} and \eqref{VL} with $\ell=3$.  To obtain the rogue wave solution in the leading term of the asymptotic formula, $\kappa_{0, 0}=\kappa_{0, 1}=\kappa_{1, 0}=0$ is posed.  Then $q_j^{[1]_3}=q_j^\mathrm{bg}\left(1+p_{j, 3}\right)\mathrm{e}^{-2\mathrm{i}\omega_j}$ with $\alpha_0=\kappa_{0, 2}, \alpha_1=\mathrm{i}\kappa_{1, 1}, \alpha_2=-2\kappa_{2, 0}, \alpha_3=-6\mathrm{i}\zeta$.

 \begin{figure}[!t]
\centering
\includegraphics[scale=0.4]{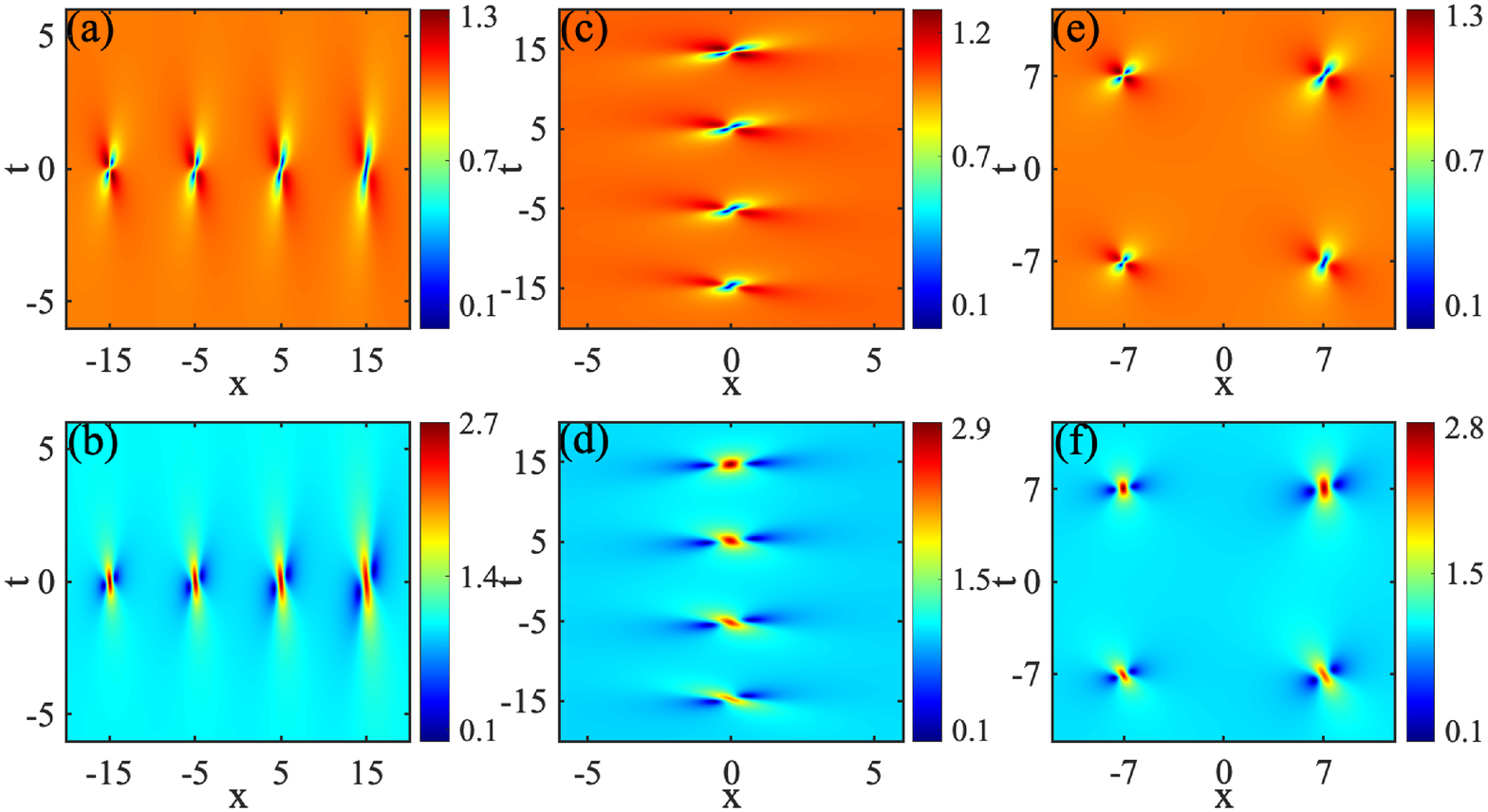}
\caption{Rogue wave solutions $q_1^{[1]_4}/a_1$ (a, c, e) and $q_2^{[1]_4}/a_2$ (b, d, f) of $4$-NLS equation in Eqs. \eqref{form-gai}, \eqref{L} and \eqref{VL} with parameters: (a, b) $\zeta=1, \alpha_0=5^49, \alpha_1=250\mathrm{i}, \alpha_2=500, \alpha_3=12\mathrm{i}, \alpha_4=24$; (c, d)  $\zeta=1,  \alpha_0=5^49, \alpha_1=0, \alpha_2=-500, \alpha_3=-24\mathrm{i}, \alpha_4=24$; (e, f) $\zeta=1, \alpha_0=10^4, \alpha_1=-300, \alpha_2=0, \alpha_3=-6\mathrm{i}, \alpha_4=24$.}
\label{RW4}
\end{figure}

\vspace{0.1in} {\it Example 3.3.3.} For the case $\ell=4$, to study the asymptotic behaviors of rogue wave, we take
\begin{gather}
\begin{gathered}
\alpha_0=\kappa_0h^4+\sum_{s=0}^3\kappa_{0, s}h^{3-s}/\zeta,\quad \alpha_1=\left(\kappa_1h^3+\sum_{s=0}^2\kappa_{1, s}h^{2-s}/\zeta\right)\mathrm{i},  \quad \alpha_2=-2\left(\kappa_2h^2+\sum_{s=0}^1\kappa_{2, s}h^{1-s}/\zeta\right),        \\
 \alpha_3=-6\left(\kappa_3h+\kappa_{3, 0}/\zeta\right)\mathrm{i},\quad \alpha_4=24,\quad \kappa_0,\kappa_{0, s}, \kappa_1, \kappa_{1, s}, \kappa_2, \kappa_{2, s},  \kappa_3, \kappa_{3, 0}, h\in\mathbb{R}.
 \end{gathered}
\end{gather}
By the transformation $x=\widehat{x}+\Re\left(z\right)h, t=\widehat{t}+\Im\left(z\right)h/\zeta$, we derive the expansions
\begin{gather}
\frac{\mathbf{L}^\dag \mathbf{L}}{\left(n+1\right)d^8}=\left|z^4+\kappa_3z^3+\kappa_2z^2+\kappa_1z+\kappa_0\right|^2+\mathcal{O}\left(\frac{1}{h}\right), \,\, h\to +\infty.
\end{gather}
Based on different types of multiple roots of the governing polynomial
 \begin{gather}\label{f4}
\mathcal{F}_4(z)=z^4+\kappa_3z^3+\kappa_2z^2+\kappa_1z+\kappa_0,
\end{gather}
we consider the following seven cases: four simple roots (Figs. \ref{RW4}(a, b, c, d, e,f)),
 a double root and two simple roots (Figs. \ref{RW4a}(a, b)), two  double roots (Figs. \ref{RW4a}(c, d)), a simple root and a triple root (Figs. \ref{RW4a}(e, f)), and a quadruple root. The details of asymptotic behaviors are addressed in {\bf Appendix D}.

Through the above analysis of some representative vector rational RWs, we have the following conjecture:

\begin{conjecture}
The asymptotic structures of $\mathcal{PT}$-symmetric vector rational rogue waves are related with some real coefficient polynomials.	
\end{conjecture}
Through the above theoretic analysis, we verify the above conjecture for the lower order cases. The systematic proof for above conjecture will be a future studying by the determinant solution or with the help of Riemann-Hilbert representation.

\begin{figure}[!t]
\centering
\includegraphics[scale=0.4]{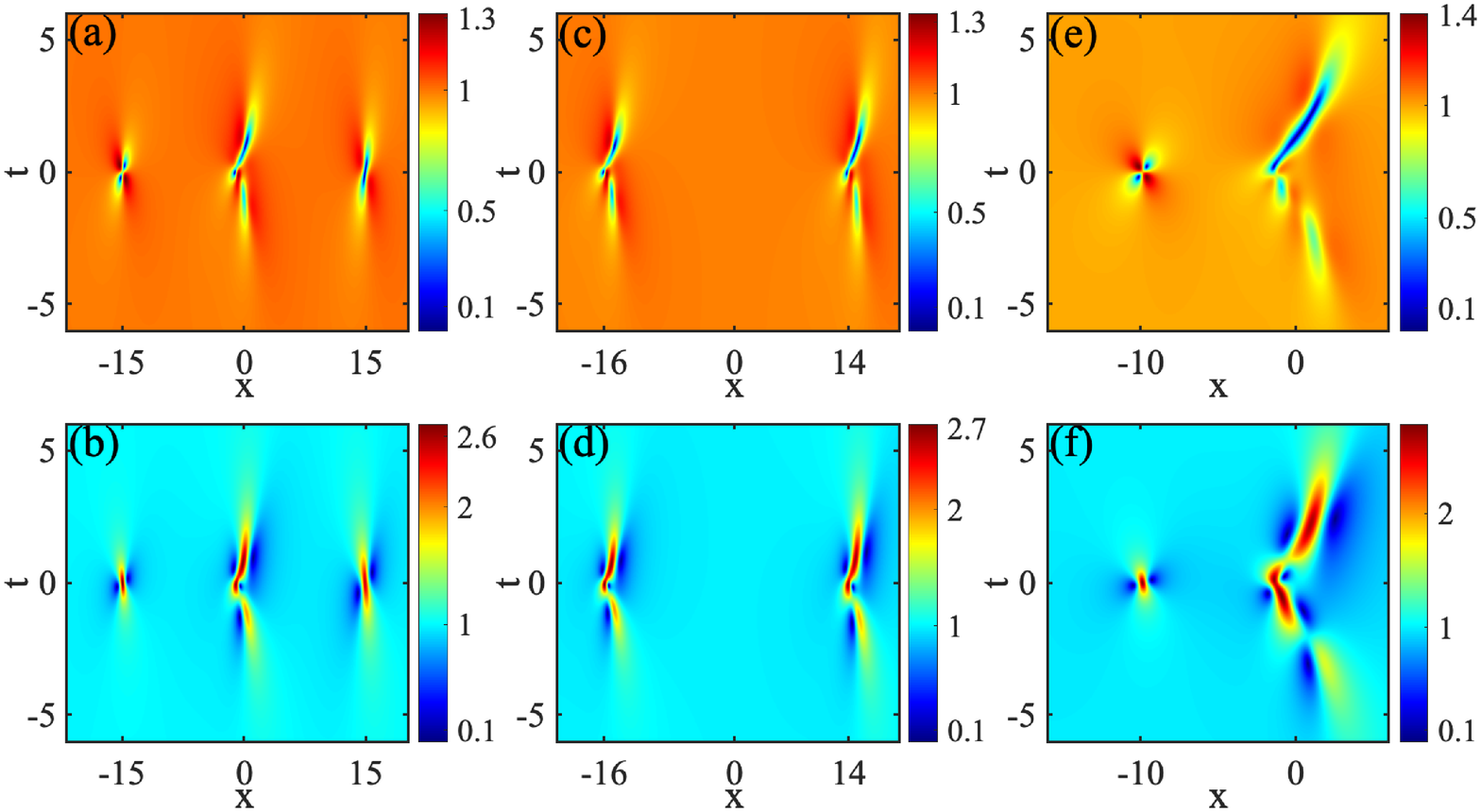}
\caption{Rogue wave solutions $q_1^{[1]_4}/a_1$ (a, c, e) and $q_2^{[1]_4}/a_2$ (b, d, f) of the $4$-NLS equation given in Eqs. \eqref{form-gai}, \eqref{L} and \eqref{VL} with parameters: (a, b) $\zeta=1, \alpha_0=10^4, \alpha_1=-300\mathrm{i}, \alpha_2=0, \alpha_3=-6\mathrm{i}, \alpha_4=24$; (c, d)  $\zeta=1, \alpha_0=15^4, \alpha_1=0, \alpha_2=900, \alpha_3=0, \alpha_4=24$; (e, f) $\zeta=1, \alpha_0=\alpha_1=\alpha_2=0, \alpha_3=-19\mathrm{i}, \alpha_4=8$.}
\label{RW4a}
\end{figure}

\subsection{Vector rational RWs with ultrahigh amplitudes}

In this subsection, we consider the vector rational RWs with the maximal amplitude. From the formula of the general first-order vector rogue waves (\ref{RW}), the parameter $\bsm{\beta}$ is related to the amplitude. Next, we will seek the suitable $\bsm{\beta}$ corresponding to the rogue wave with the maximal  amplitude. We will consider two kinds of cases: the maximal  average amplitude
and the maximal  amplitude of a certain component
as
\begin{gather}\label{GAA}
\mathcal{GA}[1]:=\max_{\bsm{\beta},x, t}\frac{\|\mathbf{q}^{[1]}\|_1}{\left\|\mathbf{a}\right\|_1}, \qquad
\mathcal{A}_j[1]:=\max_{\bsm{\beta}, x, t}\left|q_j^{[1]}/a_j\right|.
\end{gather}

\begin{prop}
Given the formula of  the vector RW solutions (\ref{RW}),  $\mathcal{GA}[1]$ can be attained  at $\bsm{\beta}=\bsm{\eta}$ and $(x, t)=(0, 0)$, and $\mathcal{A}_j[1]$ can be attained at $\bsm{\beta}=\bsm{\eta}_j$ and $(x, t)=(0, 0)$:
\begin{align}\label{GAA-asy}
\mathcal{GA}[1]&=1+\frac{\left(n+1\right)\sqrt{n}}{\sum_{j=1}^n \csc{\omega_j}}=\mathcal{O}\left(\frac{\sqrt{n}}{\ln n}\right), \quad \mathrm{as}\quad n\to\infty,  \\[0.05in]
\mathcal{A}_j[1]&=1+(n+1)\sin\omega_j
\end{align}
where $\bsm{\eta}=\left(\sqrt{n}, \mathrm{i}, \mathrm{i}, \cdots, \mathrm{i}\right)^\mathrm{T}$, $\bsm{\eta}_j=\left(1, 0, \cdots, 0 ,\mathrm{i}, 0, \cdots, 0\right)^\mathrm{T}$, and the non-zero $\mathrm{i}$ is the $j+1$-th entry of $\bsm{\eta}_j$.
\end{prop}
\begin{proof}
Given the $(n+1)$-dimensional vectors $\bsm{\gamma}=\left(\gamma_0, \gamma_1, \cdots, \gamma_n\right)$ and $\widehat{\bsm{\gamma}}=\left(\left|\gamma_0\right|, \left|\gamma_1\right|\mathrm{i}, \cdots, \left|\gamma_n\right|\mathrm{i}\right)$, one has the following inequality
\begin{gather*}
\left|a_j-\frac{2\mathrm{i}\left(n+1\right)\zeta\gamma_j\gamma_0^*}{\bsm{\gamma}^\dag\bsm{\gamma}}  \right|\le
\left|a_j-\frac{2\mathrm{i}\left(n+1\right)\zeta\left(\left|\gamma_j\right|\mathrm{i}\right)\left|\gamma_0\right|^*}{\widehat{\bsm{\gamma}}^\dag\widehat{\bsm{\gamma}}}  \right|,
\end{gather*}
from which one can claim that the maximal  value of $\sum_{j=1}^n\left|a_j-2\mathrm{i}\left(n+1\right)\zeta\gamma_j\gamma_0^*/(\bsm{\gamma}^\dag\bsm{\gamma}) \right|$ can be reached at the subset 
$\mathbb{C}_1:=\left\{\bsm{\gamma}=\left(\gamma_0, \gamma_1, \cdots, \gamma_n\right)\in \mathbb{C}^{n+1} \big| \, \gamma_0\in \mathbb{R}, \gamma_j\in \mathrm{i}\mathbb{R}, \gamma_0\ge 0, \left(\gamma_j\right)_\mathrm{I}\ge 0, 1\le j\le n\right\}$.

Given a vector $\mathbf{c}=\left(c_0, c_1\mathrm{i}, c_2\mathrm{i}, \cdots, c_n\mathrm{i} \right)\in\mathbb{C}_1$, where $c_j\ge 0\, (0\le j\le n)$, one can derive
\begin{gather}\label{ceq}
\sum_{j=1}^n\left|a_j+\frac{2\left(n+1\right)\zeta c_jc_0}{\mathbf{c}^\dag\mathbf{c}}  \right|=
\left(\sum_{j=1}^n \csc\varphi_j+\frac{2\left(n+1\right)c_0\sum_{j=1}^n c_j}{\sum_{j=0}^nc_j^2}\right)\left|\zeta\right|.
\end{gather}
The right-hand side of Eq.~(\ref{ceq}) is a multivariate function with respect to variables $c_0, c_1, \cdots, c_n$, and its maximal value is achieved at $c_0=c_1\sqrt{n}$,\, $c_j=c_1\, (1\le j\le n)$. Then the greatest value can be obtained at $\bsm{\eta}$. Note that $\mathbf{A}$ in Eq.~(\ref{RW}) is identity matrix at $(x, t)=(0, 0)$, thus one confirms $\mathcal{GA}[1]$ is attained at $\bsm{\beta}=\bsm{\eta}$, $(x, t)=(0, 0)$, and
\begin{gather*}
\mathcal{GA}[1]=\frac{\|\mathbf{q}^{[1]}(0, 0)\|_1}{\left\|\mathbf{a}\right\|_1}\bigg|_{\bsm{\beta}=\bsm{\eta}}=1+\frac{\left(n+1\right)\sqrt{n}}{\sum_{j=1}^n \csc{\omega_j}}.
\end{gather*}
By the estimate
\begin{gather*}
\frac{2}{\pi}\ln\left(\frac{n}{2}\right)\le\frac{\sum_{j=1}^n\csc\omega_j}{n+1}\le 1+\ln\left(\frac{n+1}{2}\right),
\end{gather*}
one derives
$\mathcal{GA}[1]=\mathcal{O}\left(\sqrt{n}/\ln n\right), \, n\to\infty$.

The case of $\mathcal{A}_j[1]$ can be verified by the estimate
\begin{gather*}
\mathcal{A}_j[1]=\left|q_j^{[1]}(0, 0)/a_j\right|_{\bsm{\beta}=\mathbf{c}}=1+\frac{2(n+1)c_0c_j\sin\omega_j}{\sum_{k=0}^nc_k^2}\le 1+\frac{2(n+1)c_0c_j\sin\omega_j}{c_0^2+c_j^2}\le 1+(n+1)\sin\omega_j,
\end{gather*}
 where  $\mathbf{c}=\left(c_0, c_1\mathrm{i}, \cdots, c_n\mathrm{i}\right)\in\mathbb{C}_1$. Then we complete the proof.
\end{proof}

\begin{figure}[!t]
\centering
\includegraphics[scale=0.4]{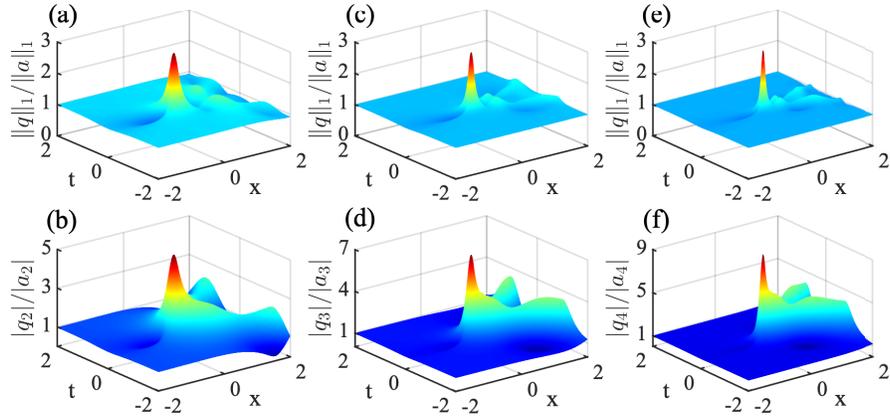}
\caption{Rogue wave solutions with $\mathcal{GA}[1]$ for (a) $3$-NLS  equation, (c) $5$-NLS  equation and (e) $7$-NLS  equation. Rogue wave solutions with $\mathcal{A}_{\frac{n+1}{2}}[1]$ for (b) $3$-NLS  equation, (d) $5$-NLS  equation and (f) $7$-NLS  equation.}
\label{fGAA}
\end{figure}

\begin{remark}
 $\bsm{\eta}$ satisfies the parity-time symmetric condition $\bsm{\eta}=\mathbf{J}\bsm{\eta}^*$, then the rogue wave $\mathbf{q}^{[1]}$ in Eq.  \eqref{RW} with $\bsm{\beta}=\bsm{\eta}$ is also the parity-time-reversal. That is to say the rogue wave with the greatest average amplitude is time-reversal.
\end{remark}

It is well known that the amplitude of the Peregrine soliton (rogue wave) for $n=1$ is $3$ times  the background. It follows from Eq. \eqref{GAA-asy} that the ratio of the amplitude of $\|\mathbf{q}^{[1]}\|_1$ and the background $\|\mathbf{a}\|_1$ is $\mathcal{O}(\sqrt{n}/\ln n)$, which tends to $\infty$ as $n\to\infty$:

\begin{itemize}

\item{} As $n=1$, one has $\mathcal{GA}[1]=3$ corresponds to the Peregrine soliton;

 \item{} As $2\le n\le 11$,  one has $2.8<\mathcal{GA}[1]<3$;

 \item{}
 As $n\ge 12$,  one has $\mathcal{GA}[1]>3$.

 \end{itemize}

 We display the dynamical structures of rogue waves with $\mathcal{GA}[1]$ for $3$-NLS equation in Fig.~\ref{fGAA}a,  $5$-NLS equation in
 Fig. \ref{fGAA}c, and $7$-NLS equation in Fig. \ref{fGAA}e, respectively.

For the function $\mathcal{A}_{j}[1]$, we have the following conclusions:

\begin{itemize}

\item{} As  $n$ is odd, for the rogue wave with maximal amplitude of the $j$-th component, $\bsm{\eta}_j=\mathbf{J}\bsm{\eta}_j^*$ iff $j=\frac{n+1}{2}$. That is, the vector rogue wave $\mathbf{q}^{[1]}$ with $\bsm{\beta}=\bsm{\eta}_{j_c}$ and $j_c=\left(n+1\right)/2$ is parity-time-reversal and the
$|q_{j_c}^{[1]}(0,0)/a_{j_c}|$ of the $j_c$-th component attains maximum $\mathcal{A}_{j_c}[1]=n+2$. We display the dynamical structures of rogue waves $|q_{j_c}^{[1]}(x,t)/a_{j_c}|$ for the $3$-NLS equation (Fig. \ref{fGAA}(b)),  $5$-NLS equation (Fig. \ref{fGAA}(d)), and $7$-NLS equation (Fig. \ref{fGAA}(f)), respectively.

\item{} As $n$ is even, the vector rogue wave $\mathbf{q}^{[1]}$ with $\bsm{\beta}=\bsm{\eta}_j$ is no longer parity-time-reversal. As both the time-reversal property and the greatest amplitude are considered, one can choose the rogue wave with $\bsm{\beta}=\left(1, 0, \cdots, 0, \mathrm{i}/\sqrt{2},  \mathrm{i}/\sqrt{2}, 0, \cdots, 0\right)^\mathrm{T}$, where the non-zero $ \mathrm{i}/\sqrt{2}$ is the $\left(n/2+1\right)$- and $\left(n/2+2\right)$-th entries. In this case, $\mathbf{q}^{[1]}$ is parity-time-reversal and
\begin{gather*}
\left|q_{\frac{n}{2}}^{[1]}(0, 0)/a_{\frac{n}{2}}\right|=\left|q_{\frac{n}{2}+1}^{[1]}(0, 0)/a_{\frac{n}{2}+1}\right|=1+\frac{\sqrt{2}}{2}(n+1)\sin\omega_{\frac{n}{2}}.
\end{gather*}

\end{itemize}

In the last section, we classify the general formula of the first-order rogue wave solution (\ref{RW}) into $n$ cases: $\mathbf{q}^{[1]_1}, \mathbf{q}^{[1]_2}, \cdots, \mathbf{q}^{[1]_n}$. We can clarify  the vector rogue waves with the maximal amplitude (\ref{GAA}) considered in this section correspond to the $n$-th case: $\mathbf{q}^{[1]_n}$. In fact, by the computation we can obtain
\begin{gather}
\mathrm{det}\left(\bsm{\xi}_0, \bsm{\xi}_1, \cdots, \bsm{\xi}_{n-1}, \bsm{\eta}\right)=\mathrm{i}^n\left(-\zeta^{-1}\right)^{\frac{n\left(n-1\right)}{2}}\left(\sqrt{n}+\sum_{k=1}^n\sin\omega_k\right)\prod_{s=1}^{n-1}\left(\sin\omega_{s+1}\right)^s\ne 0,
\end{gather}
and  for $j=1, 2, \cdots, n$,
\begin{gather}
\mathrm{det}\left(\bsm{\xi}_0, \bsm{\xi}_1, \cdots, \bsm{\xi}_{n-1}, \bsm{\eta}_j\right)=\mathrm{i}^n\left(-\zeta^{-1}\right)^{\frac{n\left(n-1\right)}{2}}\left(1+\mathrm{i}{\rm e}^{-{\rm i}\omega_j}\right)\prod_{s=1}^{n-1}\left(\sin\omega_{s+1}\right)^s\ne 0,
\end{gather}
which derive $\mathbf{B}^n\bsm{\eta}\ne\mathbf{0}$ and  $\mathbf{B}^n\bsm{\eta}_j\ne\mathbf{0}$. Then the rogue waves with $\bsm{\beta}=\bsm{\eta}$ and $\bsm{\beta}=\bsm{\eta}_j$ belong to the  case of $\ell=n$.

\section{Higher-order vector rational RWs: dynamics and maximal amplitude}\label{sec4}

\subsection{The inverse function method and an explicit formula via two determinants}

In this subsection, we present a novel technique of the inverse function to find an explicit formula of the $N$-th-order vector rational RW solutions of the $n$-NLS equation (\ref{n-NLS}) via the multi-fold Darboux transformation, where $N\in\mathbb{N}^+$. By the idea of the generalized Darboux transformation, we need to make the perturbation of the spectral parameter $\lambda_0$.

Let $\lambda=\lambda_0\left(1+\epsilon^{n+1}\right)$ with $\epsilon$ being a small perturbation parameter. Denote the  eigenvalues of the matrix $\mathbf{H}+\lambda\mathbb{I}_{n+1}$ by $\mu_0:=\mu_0(\epsilon), \ldots, \mu_{n}:=\mu_{n}(\epsilon)$, which satisfy the characteristic equation $|(\mu_j(\epsilon)-\lambda)\mathbb{I}_{n+1}-\mathbf{H}|=0$.
The main difficulty of the application of the generalized Darboux transformation to multi-component model is that $\mu_j(\epsilon)$'s can not be explicitly represented. By recalling the idea of the generalized Darboux transformation, we find that these explicit representations are in fact not necessary. We just need to solve the values $\mu_j$, $\mathrm{d}\mu_j/\mathrm{d}\epsilon$, $\mathrm{d}^2\mu_j/\mathrm{d}\epsilon^2$, $\cdots$ at $\epsilon=0$ explicitly. In next theorem, the explicit formula of the $N$-th order rogue wave solutions is presented.

\begin{theorem}\label{Nth-RW}
The explicit formula of the $N$-th order vector rational RW solutions for the $n$-NLS system (\ref{n-NLS}) is derived as
\begin{gather}\label{high-RW}
q_j^{[N]}(x,t)=q_j^{\mathrm{bg}}+2\frac{\mathrm{det}
\begin{pmatrix}
\mathbf{M}& \mathbf{Y}_1^\dagger \vspace{0.05in}\\
\mathbf{Y}_{j+1}&0
\end{pmatrix}
}{\mathrm{det}\,\mathbf{M}},\quad  j=1, 2, \cdots n
\end{gather}
with the intensity of the vector potential $\mathbf{q}^{[N]}$ being
\begin{gather}\label{RW-mf-N}
\|\mathbf{q}^{[N]}\|_2^2=\left\|\mathbf{a}\right\|_2^2+\frac{\partial^2}{\partial x^2} \ln (\mathrm{det}\,\mathbf{M}),
\end{gather}
where $\mathbf{Y}_j$ is the $j$-th rows of the matrix $\mathbf{Y}=\left(\widetilde{\Phi}_0, \widetilde{\Phi}_1, \cdots, \widetilde{\Phi}_{N-1}\right)$ with
$\widetilde{\Phi}_k=\sum_{s=0}^k\widehat{\Phi}_s\bsm{\gamma}_{k-s}\bsm{\beta}\, (k=0, 1, \ldots, N-1)$, $\bsm{\gamma}_0,  \bsm{\gamma}_1\cdots,  \bsm{\gamma}_{N-1}, \bsm{\beta}$ being $(n+1)$-dimensional free constant column vectors, and $\widehat{\Phi}_s$ are $(n+1)$-dimensional column vector as
\begin{gather}
\widehat{\Phi}_s=\left(\frac{\Phi_0^{\left(sn+s\right)}\left(0\right)}{\left(sn+s\right)!}\mathbf{R}_0, \frac{\Phi_0^{\left(sn+s+1\right)}\left(0\right)}{\left(sn+s+1\right)!}\mathbf{R}_1, \cdots, \frac{\Phi_0^{\left(sn+s+n\right)}\left(0\right)}{\left(sn+s+n\right)!}\mathbf{R}_n\right),
\end{gather}
with $\mathbf{R}_s=\left(g_0^s\left(z_0\right), g_1^s\left(z_0\right), \cdots, g_{n}^s\left(z_0\right)\right)^\dag,\, (s=0, 1, \ldots, n)$
and
\begin{gather}
g_k:z\mapsto \left|f(z)\right|^{\frac{1}{n+1}}\mathrm{e}^{\mathrm{i}\left[\mathrm{arg}\left(f(z)\right)+2k\pi\right]/(n+1)}, \,\, f(z)=2\lambda_0\prod_{s=1}^n\left(z+b_s\right),\,\, \mathrm{arg}\left(f(z)\right)\in\left(-\pi, \pi\right],\,\,k=0, \ldots, n,
\end{gather}
$\Phi_0^{\left(j\right)}\left(0\right)$ denotes the $j$-th-order derivative of matrix $\Phi_0\left(\epsilon\right)=\Big(\mathbf{h}\left(\mu_1\right), \mathbf{h}\left(\mu_2\right), \cdots, \mathbf{h}\left(\mu_{n+1}\right)\Big)$ at $\epsilon=0$ with
\begin{gather}
\mathbf{h}:
z\mapsto\mathbf{G}\left(1, \frac{a_1}{z+b_1}, \cdots, \frac{a_n}{z+b_n}\right)^\mathrm{T}
\mathrm{e}^{\mathrm{i}(z-\lambda)[x+(z+\lambda)t/2]-\mathrm{i}\left(\left\|\mathbf{a}\right\|_2^2+\lambda^2/2\right)t},
\end{gather}
$\mu_0:=\mu_0(\epsilon),\mu_1:=\mu_1(\epsilon),  \ldots, \mu_{n}:=\mu_{n}(\epsilon)$ are $n+1$ eigenvalues of the matrix  $\mathbf{H}+\lambda_0\left(1+\epsilon^{n+1}\right)\mathbb{I}_{n+1}$, satisfying
\begin{gather}\label{DD}
\mu_s\left(0\right)=z_0, \quad \mu_s^{\left(i\right)}\left(0\right)=\frac{{\rm d}^{i-1}g_s^i(z)}{{\rm d}z^{i-1}}\Big|_{z=z_0}, \quad i\in\mathbb{N}^+, \quad s=0, 1, \ldots, n,
\end{gather}
and the matrix $\mathbf{M}=\left(M_{m, k}\right)_{0\le m, k \le N-1}$  with
\begin{gather}
M_{m, k}=\frac{1}{\mathrm{i}\left(n+1\right)\zeta}\sum_{s=0}^{m}
\sum_{l=0}^{k}\left(-\frac{1}{2}\right)^{s+l}\binom{s+l}{l} \widetilde{\Phi}_{m-s}^\dag \widetilde{\Phi}_{k-l},  \,\,\,\, 1\le m, k\le N.
\end{gather}
\end{theorem}
\begin{proof}
The form of the Eq. (\ref{high-RW}) is  trivially derived by the generalized Darboux transformation in Theorem \ref{thm2}. For more details, one can also refer to Refs. \cite{Ling2014, Zhao2016, LingZ-19}. To present the explicit formula, one needs to yield the explicit and non-trivial vector solution $\widetilde{\Phi}$, which has explicit coefficients $\widetilde{\Phi}_k$ of Taylor expansion at $\epsilon=0$. Note that the characteristic polynomial of the matrix $\mathbf{H}+\lambda\mathbb{I}_{n+1}$ with respect to $z$ is
\begin{gather}
\mathrm{det}\left[\left(z-\lambda\right)\mathbb{I}_{n+1}-\mathbf{H}\right]\big|_{\lambda=\lambda_0\left(1+\epsilon^{n+1}\right)}=\left(z-z_0\right)^{n+1}-2\lambda_0\epsilon^{n+1}\prod_{j=1}^n\left(z+b_j\right).
\end{gather}
It follows from $\mathrm{det}\left[\left(\mu_s-\lambda\right)\mathbb{I}_{n+1}-\mathbf{H}\right]=0$ that one can derive
\begin{gather}\label{D0}
\epsilon=\frac{\mu_s-z_0}{g_s\left(\mu_s\right)},\quad s=0, 1, \ldots, n,
\end{gather}
from which one obtains $\mu_s\left(0\right)=z_0$. Taking the first-order derivative with respect to $\epsilon$ in Eq. (\ref{D0}) yields  $\mathrm{d}\epsilon/\mathrm{d}\mu_s=\left(1-\epsilon g_s'\left(\mu_s\right)\right)/g_s\left(\mu_j\right)$. By the derivative rule of the inverse function, one derives
\begin{gather}\label{D1}
\frac{\mathrm{d}\mu_s(\epsilon)}{\mathrm{d}\epsilon}=\frac{g_s\left(\mu_s\right)}{1-\epsilon g_s'\left(\mu_s\right)},
\end{gather}
which implies $\mu_s'\left(0\right)=g_s\left(z_0\right)$, where the prime denote the derivative with respect to $\epsilon$.
Taking the first-order derivative of Eq.~(\ref{D1}) with respect to $\epsilon$ yields
\begin{gather}\label{D2}
\frac{\mathrm{d}^2\mu_s}{\mathrm{d}\epsilon^2}=\frac{2g_s\left(\mu_s\right)g_s'\left(\mu_s\right)
+\epsilon\left[g_s^2\left(\mu_s\right)g_s''\left(\mu_s\right)-2g_s\left(\mu_s\right)g_s'^2(\mu_s)\right]}{\left(1-\epsilon g_s'\left(\mu_s\right)\right)^3},
\end{gather}
from which one deduces $\mu_s''\left(0\right)=\left(g_s^2\left(z\right)\right)'\left|\right._{z=z_0}$. Taking the derivative with respect to $\epsilon$ in Eq. (\ref{D2}), one can obtains $\mu_s'''\left(0\right)=\left(g_s^3\left(z\right)\right)''\left|\right._{z=z_0}$.
Similarly, we will obtain
\begin{gather}\label{DD}
\mu_s\left(0\right)=z_0, \quad \mu_s^{\left(k\right)}\left(0\right)=\frac{{\rm d}^{k-1}g_s^k(z)}{{\rm d}z^{k-1}}\Big|_{z=z_0}, \,\,\,\, k\in\mathbb{N}^+,\,\,\, s=1,2,...,n.
\end{gather}

Given the eigenvalue $u_s-\lambda$ of the matrix $\mathbf{H}$, one can derive the corresponding eigenvector $\mathbf{K}_s$ as
\begin{gather}
\mathbf{K}_s=\left(1, \frac{a_1}{\mu_s+b_1}, \frac{a_2}{\mu_s+b_2}, \cdots, \frac{a_n}{\mu_s+b_n}\right)^\mathrm{T},\quad s=0,1,...,n.
\end{gather}
such that one has $\mathbf{H}=\mathbf{K}\Lambda\mathbf{K}^{-1}$, where $\mathbf{K}=\left(\mathbf{K}_0, \mathbf{K}_1, \cdots, \mathbf{K}_n\right)$ and $\Lambda=\mathrm{diag}\left(\mu_0-\lambda, \mu_1-\lambda, \cdots, \mu_n-\lambda\right)$.

 Let $\Psi=\mathbf{K}\widehat{\Psi}$ in system (\ref{constant}), then one  can find that $\widehat{\Psi}$ solves the system of the linear PDEs
\begin{equation*}
\widehat{\Psi}_x=\mathrm{i}\Lambda\widehat{\Psi},\quad \widehat{\Psi}_t=\mathrm{i}\left[\frac{1}{2}\Lambda^2+\lambda\Lambda-\left(\sum_{j=1}^na_j^2
+\frac{\lambda^2}{2}\right)\mathbb{I}_{n+1}\right]\widehat{\Psi},
\end{equation*}
which has a fundamental matrix solution
\begin{gather*}
\widehat{\Psi}=\exp\left\{\mathrm{i}\Lambda\left[x\mathbb{I}_{n+1}
+\left(\frac{1}{2}\Lambda+\lambda\mathbb{I}_{n+1}\right)t\right]-\mathrm{i}\left(\left\|{\bf a}\right\|_2^2+\frac{\lambda^2}{2}\right)t\right\}.
\end{gather*}
Then the Lax pair (\ref{laxs}, \ref{laxt}) has a fundamental matrix solution $\Phi_0$ under the plane-wave potential $\mathbf{Q}=\mathbf{Q}^\mathrm{bg}$.
Note that the coefficient of each term of the Taylor series
\begin{gather*}
\Phi_0\left(\epsilon\right)=\Phi_0\left(0\right)+\Phi_0'\left(0\right)\epsilon+\frac{\Phi_0''\left(0\right)}{2}\epsilon^2+\cdots
\end{gather*}
can be explicitly derived via Eq.~(\ref{DD}). Based on the fact $g_s'(z)=(n+1)^{-1}\sum_{i=1}^n\left(z+b_i\right)^{-1}g_s(z)$, we have
${\rm rank}(\Phi_0^{\left(i\right)}(0))=1$, and each row of $\Phi_0^{\left(i\right)}(0)$ is linearly dependent with $\left(g_0^i(z_0),g_1^i(z_0), \cdots, g_{n}^i(z_0)\right), i\in\mathbb{N}$. Besides, by the fact $g_s^{n+1}(z_0)=2\lambda_0\prod_{i=1}^n\left(z_0+b_i\right)$, then one find that $\left(g_0^i(z_0), g_1^i(z_0), \cdots, g_{n}^i(z_0)\right)$ is linearly dependent with $\big(g_0^{n+1+i}(z_0), g_1^{n+1+i}(z_0), \cdots, g_{n}^{n+1+i}(z_0)\big)$.

To obtain the general non-trivial vector RW solutions of the $n$-NLS system (\ref{n-NLS}) by the generalized Darboux transformation, one needs to modify the fundamental matrix solution.  Note that $\mathbf{R}_i^\dagger\mathbf{R}_j=0$ for $i\ne j$. Let us define the matrix $\mathbf{R}=\left(\mathbf{R}_0, \mathbf{R}_1, \cdots, \mathbf{R}_n\right)$ and matrix $\mathbf{E}=\mathrm{diag}\left(1, \epsilon^{-1}, \cdots, \epsilon^{-n}\right)$. Then we construct another fundamental matrix solution of  Eqs. \eqref{laxs} and \eqref{laxt} $\widehat{\Phi}$ as
\begin{gather}\label{modify}
\widehat{\Phi}=\Phi_0\mathbf{RE},
\end{gather}
which has the following explicit Taylor expansion
\begin{gather}
\widehat{\Phi}\left(\epsilon\right)=\sum_{s=0}^{N-1}\widehat{\Phi}_s\epsilon^{s\left(n+1\right)}+\mathcal{O}\left(\epsilon^{N\left(n+1\right)}\right),
\end{gather}
Let $\bsm{\gamma}$ be an $(n+1)$-dimensional column vector in the form
\begin{gather}
\bsm{\gamma}=\left(\sum_{s=0}^{N-1}\bsm{\gamma}_s\epsilon^{s\left(n+1\right)}\right)\bsm{\beta}.
\end{gather}
Then $\widetilde{\Phi}:=\widehat{\Phi}\bsm{\gamma}$ is a fundamental vector solution of the system (\ref{initial}), which has the following explicit Taylor expansion
\begin{gather}\label{modify-1}
\widetilde{\Phi}\left(\epsilon\right)=\sum_{k=0}^{N-1}\widetilde{\Phi}_k\epsilon^{k\left(n+1\right)}.
\end{gather}

Below, we give the proof of the intensity expression \eqref{RW-mf-N}.
As we all known that the generalized Darboux transformation is the limit of the classical Darboux transformation.  In the process of the proof, we consider the classical Darboux transformation firstly and then we perform the technique of the limits. Given $N$ mutually different $\epsilon_1, \epsilon_2, \cdots, \epsilon_N$ and the corresponding $\lambda_1=\left(1+\epsilon_1^{n+1}\right),  \lambda_2=\left(1+\epsilon_2^{n+1}\right), \cdots, \lambda_N=\left(1+\epsilon_N^{n+1}\right)$. Let $\widehat{\mathbf{Y}}_j$ be the $j$-th row of the matrix $\widehat{\mathbf{Y}}:=\left(\widetilde{\Phi}\left(\epsilon_1\right), \widetilde{\Phi}\left(\epsilon_2\right), \cdots, \widetilde{\Phi}\left(\epsilon_N\right)\right), j=1, 2, \cdots, n+1$ and  the matrix $\widehat{\mathbf{M}}=\left(\widehat{M}_{m, j}\right)_{N, N}$ with each entry $\widehat{M}_{m, j}=\widetilde{\Phi}\left(\epsilon_m\right)^\dag\widetilde{\Phi}\left(\epsilon_j\right)/\left(\lambda_j-\lambda_m^*\right)$. We perform the classical $N$-fold Darboux matrix
\begin{gather*}
\mathbf{T}^{[N]}=\mathbb{I}_{n+1}+\sum_{j=1}^N\frac{\mathbf{C}_j}{\lambda-\lambda_j^*},
\end{gather*}
which converts the system (\ref{initial}) to the invariant form
\begin{gather*}
\left\{
\begin{aligned}
\Phi[N]_x&={\mathbf U}^{[N]}\Phi^{[N]},\quad \mathbf{U}^{[N]}(\lambda; x, t):=\mathrm{i}\left(\lambda\sigma_3+\mathbf{Q}^{[N]}\right) , \\
\Phi[N]_t&={\mathbf V}^{[N]}\Phi^{[N]},\quad \mathbf{V}^{[N]}(x, t; \lambda):=\lambda\mathbf{U}^{[N]}+\frac{1}{2}\sigma_3\left(\mathbf{Q}^{[N]}_x-\mathrm{i}\left(\mathbf{Q}^{[N]}\right)^2\right),
\end{aligned}\right.
\end{gather*}
with
\begin{equation*}
\mathbf{Q}^{[N]}=\begin{pmatrix}
0& \left(\mathbf{p}^{[N]}\right)^\dag  \\
\mathbf{p}^{[N]} & \mathbf{0}
\end{pmatrix}, \quad \mathbf{p}^{[N]}=\left(p_1^{[N]}, p_2^{[N]}, \cdots, p_n^{[N]}\right)^\mathrm{T},
\end{equation*}
By $\mathbf{T}^{[N]}_x+\mathbf{T}^{[N]}\mathbf{U}^{\mathrm{bg}}\left(\lambda\right)=\mathbf{U}^{[N]}\mathbf{T}^{[N]}$ and matching the term $\mathcal{O}\left(\lambda^{-1}\right)$, one yields
\begin{gather*}
\left(\sum_{j=1}^N\mathbf{C}_j\right)_x=\mathrm{i}\left(\mathbf{Q}^{[N]}\sum_{j=1}^N\mathbf{C}_j-\sum_{j=1}^N\mathbf{C}_j\mathbf{Q}^\mathrm{bg}\right).
\end{gather*}
By $\mathbf{T}^{[N]}_t+\mathbf{T}^{[N]}\mathbf{V}^{\mathrm{bg}}\left(\lambda\right)=\mathbf{V}^{[N]}\mathbf{T}^{[N]}$ and matching the term $\mathcal{O}\left(1\right)$, one yields
\begin{gather*}
\left(\mathbf{Q}^{[N]}\right)^2=\mathbf{Q}^\mathrm{bg}+\mathrm{i}\left(\mathbf{Q}^\mathrm{bg}-\mathbf{Q}^{[N]}\right)_x+2\sigma_3\left(\mathbf{Q}^{[N]}\sum_{j=1}^N\mathbf{C}_j-\sum_{j=1}^N\mathbf{C}_j\mathbf{Q}^\mathrm{bg}\right).
\end{gather*}
Combining the above two equations, one deduces
\begin{gather*}
\left(\mathbf{Q}^{[N]}\right)^2=\mathbf{Q}^\mathrm{bg}+\mathrm{i}\left(\mathbf{Q}^\mathrm{bg}-\mathbf{Q}^{[N]}\right)_x-2\mathrm{i}\sigma_3\left(\sum_{j=1}^N\mathbf{C}_j\right)_x.
\end{gather*}
By extracting the diagonal entries, we derives
\begin{gather*}
\|\mathbf{p}^{[N]}\|_2^2=\|\mathbf{a}\|_2^2-2\mathrm{i}\frac{\partial}{\partial x}\left(\sum_{j=1}^N\mathbf{C}_j\right)_{1, 1}, \quad \left|p_j^{[N]}\right|^2=\left|a_j\right|^2+2\mathrm{i}\frac{\partial}{\partial x}\left(\sum_{j=1}^N\mathbf{C}_j\right)_{j+1, j+1}.
\end{gather*}
Note that
\begin{gather*}
\sum_{j=1}^N\mathbf{C}_j=-\widehat{\mathbf{Y}}\widehat{\mathbf{M}}^{-1}\widehat{\mathbf{Y}}^\dag, \quad \left(\frac{\widetilde{\Phi}\left(\epsilon_m\right)^\dag\widetilde{\Phi}\left(\epsilon_j\right)}{k_j-k_m^*}\right)_x=\mathrm{i}\left(\widetilde{\Phi}\left(\epsilon_m\right)^\dag\sigma_3\widetilde{\Phi}\left(\epsilon_j\right)\right).
\end{gather*}
Then one obtains
\begin{gather*}
\|\mathbf{p}^{[N]}\|_2^2=\left\|\mathbf{a}\right\|_2^2+\frac{\partial^2}{\partial x^2} \ln (\mathrm{det}\,\widehat{\mathbf{M}}),
\end{gather*}
Taking the limits $\epsilon_l \to 0\, (l=1, 2, \cdots, N$), we have the $q_j^{[N]}=\lim_{\epsilon_l \to 0}p_j^{[N]}$. Thus we complete the proof.
\end{proof}

From the intensity of the vector rogue wave $\mathbf{q}^{[N]}$ in the above proof, the high-order rogue wave solutions also have the following  conservation laws:
\begin{gather}
\int_{-\infty}^{+\infty}\left(\|\mathbf{q}^{[N]}\|_2^2-\left\|\mathbf{a}\right\|_2^2\right)\mathrm{d}x=0.
\end{gather}

\subsection{The governing polynomials and dynamical behaviors of higher-order RWs}

The approach of governing polynomial can also be used to study the dynamical behaviors of the above-mentioned high-order vector rational rogue wave solutions. Without loss of generality, we take $\zeta=1$ for convenience. From the above examples of the first-order vector rational RW solutions in Sec. 3, one can notice that the governing polynomials play a key and important role, which control the central locations of RWs in the leading terms of asymptotic expressions. Besides, the governing polynomial of the first-order rogue wave solution is free and one can adjust the coefficients  based on our purpose. However, one can find that the governing polynomials of high-order vector RW solutions are not completely free and some coefficients are restricted. \\


{\it Example 4.2.1.\, } For the second-order vector RW solutions of the single NLS equation, from Theorem. \ref{Nth-RW}, the free parameters are $\bsm{\gamma}_0\bsm{\beta}$ and $\bsm{\gamma}_1\bsm{\beta}$, denoted by
\begin{gather}
\bsm{\gamma}_0\bsm{\beta}=\left(u_1, \mathrm{i}\right)^\mathrm{T}, \quad \bsm{\gamma}_1\bsm{\beta}=\left(u_2, \mathrm{i}u_3\right)^\mathrm{T}, \quad u_1, u_2, u_3\in\mathbb{R}.
\end{gather}
To perform the method of governing polynomial, we take
\begin{gather}
u_1=-\frac{2}{3}\kappa_2h+\kappa_{2, 0}, \quad u_2=\frac{4}{3}\kappa_0h^3+4\sum_{s=0}^2\kappa_{0, s}h^{2-s}, \quad u_3=0.
\end{gather}
Since $u_2$ and $u_3$ have the same order in $\mathrm{det}\,\mathbf{M}$, then we can take $u_3=0$. Under the transformation $x=\widehat{x}+\Re\left(z\right)h, t=\widehat{t}+\Im\left(z\right)h$, we have the asymptotic expansion
\begin{gather}
-\frac{9\,\mathrm{det}\,\mathbf{M}}{256h^6}=\left|\mathcal{F}_3(z)\right|^2+\mathcal{O}\left(\frac{1}{h}\right), \quad h\to +\infty,
\end{gather}
where
\begin{gather}
\mathcal{F}_3(z)=z^3+\kappa_2z^2+\kappa_1z+\kappa_0, \quad \kappa_1=\frac{\kappa_2^2}{3}.
\end{gather}

One can find that $\mathcal{F}_3(z)$ has no three distinct and simple real roots or three distinct and simple roots with same real part, thus there do not exist the line-typed structures along the $x$-axis or $t$-axis. As $\kappa_0=1, \kappa_2=0$, $\mathcal{F}_3$ has three simple roots $\mathrm{e}^{(2s+1)\pi/3}, s=0, 1, 2$, in which we have the asymptotic behavior
\begin{gather}
q_1^{[2]}=q_1^\mathrm{bg}\left(1+\sum_{s=0}^2R_1\left(x-x_s, t-t_s\right)\right)+\mathcal{O}\left(\frac{1}{h^2}\right), \quad h\to +\infty,
\end{gather}
where $x_s=(h+\kappa_{0, 0})\cos[(2s+1)\pi/3]+(\kappa_{2, 0}-1)/2, \, t_s=(h+\kappa_{0, 0})\sin[(2s+1)\pi/3]$.  In particular, as $\kappa_{2, 0}=1, \kappa_{0, 0}=0$, three fundamental RWs located at $(h\cos[(2s+1)\pi/3], h\sin[(2s+1)\pi/3])$


\vspace{0.1in}
{\it Example 4.2.2.} For the $2$-NLS equation with $N=2$, according to the degree of the denominator of $q_j^{[2]}$, we have the following two cases.

\begin{itemize}

\item{} For the first case,  the vanishing  third entry of  $\bsm{\gamma}_0\bsm{\beta}$ is considered and $\bsm{\gamma}_0\bsm{\beta}, \bsm{\gamma}_1\bsm{\beta}$ are denoted by
\begin{gather}
\bsm{\gamma}_0\bsm{\beta}=\left(\sqrt[3]{2}u_1,\, \mathrm{i},\, 0\right)^\mathrm{T}, \quad
\bsm{\gamma}_1\bsm{\beta}=\left(\sqrt[3]{2}u_2, \,\mathrm{i}2^{2/3}u_3,\, u_4/(2\sqrt[3]{2})\right)^\mathrm{T}, \quad u_1, u_2, u_3,u_4\in\mathbb{R}.
\end{gather}
To derive the governing polynomial, we take
\begin{gather}
u_1=-\kappa_3h+\kappa_{3, 0}, \, u_2=-\kappa_0h^4+24\sum_{s=0}^3\kappa_{0, s}h^{3-s}, \, u_3=0, \, u_4=\left(\kappa_3^2-\kappa_2\right)h^2+24\left(\kappa_{2, 0}h+\kappa_{2, 1}\right).
\end{gather}
Since $u_3$ and $u_4$ have the same order in $\mathrm{det}\,\mathbf{M}$, then we can take $u_3=0$. Under the transformation $x=\widehat{x}+\Re\left(z\right)h, \, t=\widehat{t}+\Im\left(z\right)h$, we have the asymptotic expansion
\begin{gather}
-\frac{\mathrm{e}^{-2\left(\widehat{x}+\Re\left(z\right)h\right)}\mathrm{det}\,\mathbf{M}}{162\sqrt[3]{2}h^8}
=\left|\mathcal{F}_4(z)\right|^2+\mathcal{O}\left(\frac{1}{h}\right), \quad h\to +\infty,
\end{gather}
where
\begin{gather}
\mathcal{F}_4(z)=z^4+2\kappa_3z^3+\kappa_2z^2+\kappa_1z+\kappa_0, \quad \kappa_1=\kappa_3\left(\kappa_2-\kappa_3\right).
\end{gather}

One can study the asymptotic behavior of the RWs based on the roots of the governing polynomial $\mathcal{F}_4$. Here, we just give   a kind of straightline-shaped wave structure. As $\kappa_0=9,\, \kappa_2=-10,\, \kappa_3=0$, the governing polynomial $\mathcal{F}_4$ has four simple roots $z=\pm 1, \pm 3$. At last, we obtain the asymptotic formula
\begin{gather}
q_j^{[2]}=q_j^\mathrm{bg}\left[1+\sum_{\delta=0}^1R_j\left(x-x_\delta, t\right)+R_j\left(x-\widetilde{x}_\delta, t\right)\right]+\mathcal{O}\left(\frac{1}{h^2}\right), \quad h\to +\infty,
\end{gather}
where $x_\delta=(-1)^\delta[h-3(\kappa_{0, 0}+\kappa_{2, 0})/2]-5/4,\, \widetilde{t}_\delta=(-1)^\delta[3h+(\kappa_{0, 0}+9\kappa_{2, 0})/2]+5/4$.

\item{} For the second case, the non-vanishing  third entry of  $\bsm{\gamma}_0\bsm{\beta}$ is considered and $\bsm{\gamma}_0\bsm{\beta}, \bsm{\gamma}_1\bsm{\beta}$ are denoted by
\begin{gather}
\bsm{\gamma}_0\bsm{\beta}=\left(2^{5/3}u_1, \mathrm{i}u_22^{4/3}, 1\right)^\mathrm{T}, \quad \bsm{\gamma}_1\bsm{\beta}=\left(2^{2/3}u_3/5, \mathrm{i}2^{1/3}u_4/5, u_5\right)^\mathrm{T}, \quad u_1, u_2, u_3,u_4, u_5\in\mathbb{R}.
\end{gather}
To derive the governing polynomial, we take
\begin{gather}\label{2nls-xishud}
u_1=\left(2\kappa_5^2-\kappa_4\right)h^2, \quad u_2=\kappa_5h, \quad u_3=-\kappa_1h^5, \quad u_4=\left[\kappa_2-10\left(\kappa_4^2+4\kappa_5^4-4\kappa_4\kappa_5^2\right)\right]h^4, \quad u_5=0.
\end{gather}
Since $u_4$ and $u_5$ have the same orders in $\mathrm{det}\,\mathbf{M}$, then we can take $u_5=0$. Under the transformation $x=\widehat{x}+\Re\left(z\right)h,\, t=\widehat{t}+\Im\left(z\right)h$, we have the asymptotic expansion
\begin{gather}
-\frac{25\mathrm{e}^{-2\left(\widehat{x}+\Re\left(z\right)h\right)}\mathrm{det}\,\mathbf{M}}{2^{2/3}1296h^8}
=\left|\mathcal{F}_6(z)\right|^2+\mathcal{O}\left(\frac{1}{h}\right), \quad h\to +\infty,
\end{gather}
where
\begin{gather}
\mathcal{F}_6(z)=z^6+6\kappa_5z^5+5\kappa_4z^4+\kappa_3z^3+\kappa_2z^2+\kappa_1z+\kappa_0,
\end{gather}
with \begin{gather*}
\kappa_3=20\kappa_5\left(\kappa_4-\kappa_5^2\right), \quad
\kappa_0=-80\kappa_5^6+120\kappa_4\kappa_5^4+\left(2\kappa_2-60\kappa_4^2\right)\kappa_5^2+\kappa_1\kappa_5+10\kappa_4^3-\kappa_2\kappa_4.
\end{gather*}
One can  add the low-order terms in Eq. \eqref{2nls-xishud} for the study the asymptotic formula in a standard way  based on the roots of the governing polynomial $\mathcal{F}_6$. Here we omit them.

\end{itemize}

\subsection{High-order vector RW solutions with the maximal amplitude}

Similarly, in this subsection,  we consider  two cases of the maximal amplitudes of the obtained higher-order vector rational RWs:
\begin{gather}
\mathcal{GA}[N]:=\max_{\bsm{\gamma},x, t}\frac{\|\mathbf{q}^{[N]}\|_1}{\left\|\mathbf{a}\right\|_1}, \quad
\mathcal{A}_j[N]:=\max_{\bsm{\gamma}, x, t}\left|q_j^{[N]}/a_j\right|.
\end{gather}
In the following proposition, one can verify that the $N$-th-order vector rogue wave solution (\ref{high-RW}) attains the maximal amplitude at the origin $\left(0, 0\right)$  for some chosen $\bsm{\gamma}$.

\begin{prop}
As  $\bsm{\gamma}_0=\widehat{\Phi}_0\left(0, 0\right)^{-1}, \, \bsm{\gamma}_m=-\widehat{\Phi}_0\left(0, 0\right)^{-1}\sum_{s=0}^{m-1}\widehat{\Phi}_{m-s}\left(0, 0\right)\bsm{\gamma}_s \,(1\le m \le N-1)$,
\begin{itemize}
\item $\|{\bf q}^{[N]}\|_1/\|{\bf a}\|_1$ attains the maxima at   $\bsm{\beta}=\bsm{\eta}$ and $\left(x, t\right)=\left(0, 0\right)$, that is
\begin{gather}\label{N-GAA}
\mathcal{GA}[N]=1+\frac{N\left(n+1\right) \sqrt{n}}{\sum_{j=1}^n \csc{\omega_j}};
\end{gather}
\item $\left|q_j^{[N]}/a_j\right|$ attains the maxima  at  $\bsm{\beta}=\bsm{\eta}_j$ and $\left(x, t\right)=\left(0, 0\right)$, that is, $\mathcal{A}_j[N]=1+N\left(n+1\right) \sin\omega_j$.
\end{itemize}
\begin{proof}
For two cases of the amplitudes, we just verify the case of $\mathcal{GA}[N]$, and the case of $\mathcal{A}_j[N]$ can be shown in a similar way.
In the process of the proof, we firstly consider the classical Darboux transformation, and then we perform the technique of the limits. Given $N$ mutually different $\epsilon_1, \epsilon_2, \cdots, \epsilon_N$ and the corresponding $\lambda_1=\left(1+\epsilon_1^{n+1}\right),  \lambda_2=\left(1+\epsilon_2^{n+1}\right), \cdots, \lambda_N=\left(1+\epsilon_N^{n+1}\right)$, one obtains $\widetilde{\Phi}\left(0, 0; \epsilon_s\right)=\bsm{\eta}\, (s=1, 2, \cdots, N)$. Denote the potentials derived by the classical Darboux transformation by $\mathbf{p}^{[1]}, \mathbf{p}^{[2]}, \cdots, \mathbf{p}^{[N]}$. By the transformation between the potential $\mathbf{p}^{[1]}$ and $\mathbf{q}^\mathrm{bg}$,
\begin{gather*}
p_j^{[1]}=q_j^\mathrm{bg}-\frac{4\mathrm{i}\Im\left(\lambda_1\right)\left(\widetilde{\Phi}\left(\epsilon_1\right)\widetilde{\Phi}
\left(\epsilon_1\right)^\dag\right)_{j+1, 1}}{\widetilde{\Phi}\left(\epsilon_1\right)^\dag\widetilde{\Phi}\left(\epsilon_1\right)},
\end{gather*}
one deduces $\|\mathbf{p}^{[1]}\|_1$ attains the maximum at $\left(0, 0\right)$ and $\|\mathbf{p}^{[1]}\left(0, 0\right)\|_1=\left|\sum_{j=1}^na_j+2\Im\left(\lambda_1\right)\sqrt{n}\right|$. By the one-fold Darboux transformation
\begin{gather*}
\mathbf{T}_1\left(\epsilon\right)=\mathbb{I}_{n+1}+\frac{\lambda_1^*-\lambda_1}{\lambda-\lambda_1^*}\frac{\widetilde{\Phi}\left(\epsilon_1\right)\widetilde{\Phi}\left(\epsilon_1\right)^\dag}{\widetilde{\Phi}\left(\epsilon_1\right)^\dag\widetilde{\Phi}\left(\epsilon_1\right)},
\end{gather*}
one obtains the wave function of the first iteration $\widetilde{\Phi}[1]$ at $\left(0, 0; \epsilon_2\right)$ is $\widetilde{\Phi}[1]\left(0, 0; \epsilon_2\right)=\left(1+\frac{\lambda_1^*-\lambda_1}{\lambda_2-\lambda_1^*}\right)\bsm{\eta}$. By the transformation between the potential $\mathbf{p}^{[2]}$ and $\mathbf{p}^{[1]}$,
\begin{gather*}
p_j^{[2]}=p_j^{[1]}-\frac{4\mathrm{i}\Im\left(\lambda_2\right)\left(\widetilde{\Phi}[1]\left(\epsilon_2\right)\widetilde{\Phi}[1]\left(\epsilon_2\right)^\dag\right)_{j+1, 1}}{\widetilde{\Phi}[1]\left(\epsilon_2\right)^\dag\widetilde{\Phi}[1]\left(\epsilon_2\right)},
\end{gather*}
one deduces $\|\mathbf{p}^{[2]}\|_1$ attains the maximum at $\left(0, 0\right)$ and $\|\mathbf{p}^{[2]}\left(0, 0\right)\|_1=\left|\sum_{j=1}^na_j+2\Im\left(\lambda_1+\lambda_2\right)\sqrt{n}\right|$. Continuing  the process, one can find the wave function at each iteration $\widetilde{\Phi}[s]$ at $\left(0, 0; \epsilon_{s+1}\right), s=1, 2, \cdots, N-1$  is always linearly dependent with $\bsm{\eta}$. Therefore,  the next iterated potential $\|\mathbf{p}^{[s+1]}\|_1$ attains the maximum at $\left(0, 0\right)$ and $\|\mathbf{p}^{[s+1]}\left(0, 0\right)\|_1=\left|\sum_{j=1}^na_j+2\left(\sum_{l=1}^{s+1}\Im\left(\lambda_l\right)\right)\sqrt{n}\right|$. After $N$ iterations, $\|\mathbf{p}^{[N]}\|_1$ attains the maximum at $\left(0, 0\right)$ and
\begin{gather*}
\|\mathbf{p}^{[N]}\left(0, 0\right)\|_1=\left|\sum_{j=1}^na_j+2\left(\sum_{l=1}^{N}\Im\left(\lambda_l\right)\right)\sqrt{n}\right|.
\end{gather*}
As $\epsilon_l\to 0$, one has $\lambda_l \to\lambda_0$. Let $\epsilon_l \to 0, l=1, 2, \cdots, N$, we yield
\begin{gather*}
\|\mathbf{q}^{[N]}\left(0, 0\right)\|_1=\left|\sum_{j=1}^na_j+2N\Im\left(\lambda_0\right)\sqrt{n}\right|,
\end{gather*}
which derives Eq. (\ref{N-GAA}). Then we complete the proof.
\end{proof}
\end{prop}

\begin{figure}[!t]
\centering
\includegraphics[scale=0.4]{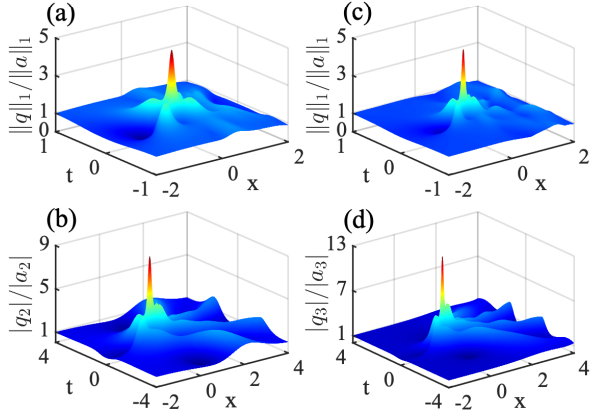}
\caption{The second-order rogue wave solution with $\mathcal{GA}[2]$ for (a) $3$-NLS  equation and (c) $5$-NLS  equation. The second-order rogue wave solution with $\mathcal{A}_{\frac{n+1}{2}}[2]$ for (b) $3$-NLS  equation and (d) $5$-NLS  equation.}
\label{GAAhigh}
\end{figure}

It follows from the above proof that $N$-th-order vector rogue wave $\mathbf{q}^{[N]}$ with $\bsm{\beta}=\bsm{\eta}$ is time-reversal. We display the dynamical structures of the second-order rogue waves with $\mathcal{GA}[2]$ for $3$-NLS equation (Fig. \ref{GAAhigh}(a)) and $5$-NLS equation (Fig. \ref{GAAhigh}(c)).

 For the rogue wave with $j$-th component of the maximal amplitude, $\bsm{\eta}_j$ satisfies $\bsm{\eta}_j=\mathbf{J}\bsm{\eta}_j$ iff $n$ is odd and $j=\frac{n+1}{2}$. Then as $n$ is odd, the rogue wave $\mathbf{q}^{[1]}$ with $\bsm{\beta}=\bsm{\eta}_{j_c}$ is time-reversal and its $j_c$-th component attains maximum $2n+3$, where $j_c=\left(n+1\right)/2$. We display the dynamical structures of the second-order rogue waves with $\mathcal{GA}_{j_c}[1]$ for $3$-NLS equation (Fig. \ref{GAAhigh}(b)) and $5$-NLS equation (Fig. \ref{GAAhigh}(d)). As $n$ is even and $\bsm{\beta}=\left(1, 0, \cdots, 0, \mathrm{i}/\sqrt{2},  \mathrm{i}/\sqrt{2}, 0, \cdots, 0\right)^\mathrm{T}$, the $N$-th-order  rogue wave $\mathbf{q}^{[N]}$  is also time-reversal, where the non-zero $ \mathrm{i}/\sqrt{2}$ is the $\left(n/2+1\right)$- and $\left(n/2+2\right)$-th entries. In this case, $\mathbf{q}^{[N]}$ is time-reversal and
\begin{gather}
\left|q_{\frac{n}{2}}^{[N]}(0, 0)/a_{\frac{n}{2}}\right|=\left|q_{\frac{n}{2}+1}^{[N]}(0, 0)/a_{\frac{n}{2}+1}\right|=1+\frac{\sqrt{2}}{2}N(n+1)\sin\omega_{\frac{n}{2}}.
\end{gather}

\section{Conclusions and discussions}

In conclusion, we have explicitly found the existence condition of $\left(n+1\right)$-multiple eigenvalues for $n$-component nonlinear Sch\"odinger equation, which play a fundamental role in the study of using the Darboux transformation to construct vector rational rogue waves.
 Base on the proper plan-wave solution $\mathbf{q}^\mathrm{bg}$ for the $\left(n+1\right)$-multiple eigenvalues,
 we first presented the explicit formulae of both first-order and high-order vector rational RW solutions. For the first-order vector rational
 RW solutions, a complete classification is proposed based on the degree $\ell$ of $\mathbf{A}\bsm{\beta}$. When $\ell=1$, we classify the first-order vector RWs into three types of fundamental RWs (bright, dark, and four-pated RWs). However, when $\ell\geq 2$, the so-called {\it new governing polynomials} $\mathcal{F}_\ell(z)$'s are introduced to study the asymptotic behaviors of first-order vector RWs as some parameters tend to infinity. The leading term of asymptotic formula is the linear superposition of some fundamental RWs, whose central location is related to the governing polynomial $\mathcal{F}_\ell(z)$, and can be deduced. Two kinds of amplitude constants $\mathcal{GA}[1], \mathcal{A}_j[1]$ are also studied and the first-order vector rogue wave solutions with $\bsm{\beta}=\bsm{\eta}$ and $\bsm{\beta}=\bsm{\eta}_j$ have the properties:
\bee\begin{array}{l}
\mathcal{GA}[1]=1+\dfrac{\left(n+1\right)\sqrt{n}}{\sum_{j=1}^n \csc{\omega_j}}=\mathcal{O}\left(\dfrac{\sqrt{n}}{\ln n}\right), \quad \mathrm{as}\quad n\to\infty,  \vspace{0.1in} \\
\mathcal{A}_j[1]=1+(n+1)\sin\omega_j.
\end{array}
\ene

Be recalling the idea of the generalized Darboux transformation, we find the explicit zeros of the polynomial of  degree $(n+1)$ with the parameter $\epsilon$ is not necessary and one just need to yield the coefficients of the Taylor expansion at $\epsilon=0$. By performing the technique of inverse function, we present the explicit formula of high-order vector rogue wave solutions. As the direct application, two kinds of constants $\mathcal{GA}[N], \mathcal{A}_j[N]$ related to the amplitude, are also studied. The $N$-th-order rogue wave solution with $\bsm{\beta}=\bsm{\eta}$ and $\bsm{\beta}=\bsm{\eta}_j$ has the properties:
\begin{gather}
\mathcal{GA}[N]=1+\frac{N\left(n+1\right) \sqrt{n}}{\sum_{j=1}^n \csc{\omega_j}}, \quad \mathcal{A}_j[N]=1+N\left(n+1\right) \sin\omega_j.
\end{gather}

The condition of the multiple root plays the fundamental and essential role of the existence of the rogue wave solutions. In this paper, we present a religious  rogue wave theory for the $n$-component NLS system with an $(n+1)$-multiple root.  Recalling the Lemma \ref{root-condition}, the condition of $(n+1)$-multiple root is determined based on $n$ mutually different $b_j's$. In the following we give the general existence condition of the multiple root according to different cases of the mutuality  of $b_j's$. Without loss of generality, suppose $b_1, b_2, \cdots, b_m$ are mutually different and $b_{m+1}, b_{m+2}, \cdots, b_n\in \{b_1, b_2, \cdots, b_m\}$, where $1\le m\le n$. We define $m$ mutually exclusive index sets $\Lambda_1, \Lambda_2, \cdots, \Lambda_m$ such that $b_i=b_j$ for $i\in \Lambda_j$. As the plane wave solution $\mathbf{q}^\mathrm{bg}$ in Eq. (\ref{plane}) is chosen as
\begin{gather}
b_j=\frac{2\Im\left(\lambda_0\right)}{m+1}\cot\frac{j\pi}{m+1}-2\Re\left(\lambda_0\right),  \quad \sum_{i\in\Lambda_j}a_i^2=\left(\frac{2\Im\left(\lambda_0\right)}{m+1}\csc \frac{j\pi}{m+1} \right)^2, \quad j=1, 2, \cdots, m,
\end{gather}
the characteristic polynomial $D(z)$ in Eq. (\ref{Dx}) has the following form of multiple root
\begin{gather}
D\left(x\right)=\left(x-z_0\right)^{m+1}\prod_{j=1}^m\left(x+b_j\right)^{\left|\Lambda_j\right|-1}, \quad z_0=\frac{2\mathrm{i}\Im\left(\lambda_0\right)}{m+1}+2\Re\left(\lambda_0\right),
\end{gather}
where $\left|\Lambda_j\right|$ denotes the amount of the entries in the index set $\Lambda_j$. With the general existence condition and the process of the paper, one can study the semi-rational solutions.

The idea used in this paper can also be extended to other $n$-component integrable nonlinear wave systems such as the higher-order n-component NLS system (e.g., n-component Hirota system, n-component fifth-oder NLS system), n-component mKdV system, n-component complex mKdV system, n-component KP system, and etc. These issues will be given in another literature.

The recent studies on the RW theory in the integrable system shows that the RW solutions almost possess the uniform structure. For the fixed order RW solutions, one of them will attain the maximum peak, the others can be decomposed into the asymptotic linear superposition of lower order RWs. Simultaneously, there exist infinite many hierarchy of RW solutions similar as the soliton hierarchy. And the infinite order RW are constructed recently from the Riemann-Hilbert representation of Darboux transformation \cite{BilmanLM-20}. The RW solutions can not be involved in the traditional inverse scattering transform since the RW solutions will have the same scattering data as the background solution or the so-called spectral singularity \cite{BilmanM-19}. The recent proposed robust inverse scattering transform can be used to deal with the spectral singularity of RWs on the non-vanishing background \cite{BilmanM-19}.

On the other hand, as for the non-integrable equation, up to now there is no systematic construction on the RW solution for the non-integrable model with MI under the non-vanishing background except for the non-integrable defocusing NLS equation with time-dependent potential possessing the RW solution~\cite{dnlsp}. The existence or non-existence of RWs on the non-integrable systems with MI is still a puzzle up to now \cite{Strauss77} which is similar as the existence or nonexistence of two-soliton solution in the non-integrable KdV equation \cite{Martel-11,Martel-11-1}. It is widely believed that the mechanism of RWs is MI. The RWs were grown in the background of MI. Thus the orbital stability theory will not adapt for the studies of RWs \cite{GrillakisSS-87}. Thus how to develop a theory to study the stabilities of RWs is also open for us. All in all, further studies on the RWs will not only impulse the development of integrable systems but also bring the new topic to the theoretical analysis and numerical studies of non-integrable dispersive equations in the fields of applied mathematics and mathematical physics.



\section*{Appendix A: The proof of the Lemma 1.}
\addcontentsline{toc}{section}{Appendix A}

\begin{proof}
	Since the coefficient matrices $\mathbf{U}$ and $\mathbf{V}$ have  the symmetric relationships:
	\begin{equation}
	\mathbf{U}(\lambda;x,t)=-\mathbf{U}^{\dag}(\lambda^*;x,t),\quad \mathbf{V}(\lambda;x,t)=-\mathbf{V}^{\dag}(\lambda^*;x,t),
	\end{equation}
	thus it follows that the matrix function $\Phi^{\dag}(\lambda^*;x,t)$ satisfies the adjoint Lax pair:
	\begin{equation}\label{eq:adjoint-lax}
	-\Phi_x^{\dag}(\lambda^*;x,t)=\Phi^{\dag}(\lambda^*;x,t) \mathbf{U}(\lambda;x,t),\quad
	-\Phi_t^{\dag}(\lambda^*;x,t)=\Phi^{\dag}(\lambda^*;x,t) \mathbf{V}(\lambda;x,t)
	\end{equation}
if $\Phi(\lambda;x,t)$ satisfies the Lax pair \eqref{laxs} and \eqref{laxt}.
	On the other hand, the inverse matrix function $\Phi^{-1}(\lambda;x,t)$ also satisfies the adjoint Lax pair \eqref{eq:adjoint-lax}. By the uniqueness and existence of differential equations and $\Phi(\lambda;0,0)=\mathbb{I}$, we complete the proof.
\end{proof}

\section*{Appendix B: The proof of the Lemma 2: }
\addcontentsline{toc}{section}{Appendix B}

\begin{proof}
	The inverse matrix function $\mathbf{N}^{-1}(\lambda;x,t)$ satisfies the adjoint linear spectral problem
	\begin{equation}
	\frac{\partial}{\partial x}\left(\mathbf{N}^{-1}(\lambda;x,t)\right)=-\mathbf{N}^{-1}(\lambda;x,t) \mathbf{U}(\lambda;x,t),
	\end{equation}
	and then we have the following stationary zero-curvature equations:
	\begin{equation}\label{eq:station-zero}
	\frac{\partial}{\partial x}\Theta(\lambda;x,t)=\left[\mathbf{U}(\lambda;x,t),\Theta(\lambda;x,t)\right],
	\end{equation}
	by $\Theta(\lambda;x,t)\equiv\mathbf{N}(\lambda;x,t)\sigma_3\mathbf{N}^{-1}(\lambda;x,t)=\mathbf{M}(\lambda;x,t)\sigma_3\mathbf{M}^{-1}(\lambda;x,t).$
	The above coefficients $\Theta_j(x,t)$ can be determined recursively:
	\begin{equation}\label{eq:theta-re}
	\begin{split}
	\Theta_{j+1}^{\rm off}&=-\frac{1}{2}\sigma_3\left(\ii\,\Theta_{j,x}^{\rm off}+\left[\mathbf{Q}, \Theta_j^{\rm diag}\right]\right), \\
	\Theta_{j,x}^{\rm diag}&=\ii \left[\mathbf{Q}, \Theta_j^{\rm off}\right].\\
	\end{split}
	\end{equation}
	On the other hand, it is readily to see that
	\begin{equation}
	\Theta^2(\lambda;x,t)=\mathbb{I},
	\end{equation}
	which implies that
	\begin{equation}\label{eq:theta-diag}
	\Theta_j^{\rm diag}=-\frac{\sigma_3}{2}\sum_{s=1}^{j-1}\left(\Theta_j^{\rm diag}\Theta_{s-j}^{\rm diag}+\Theta_j^{\rm off}\Theta_{s-j}^{\rm off}\right),\qquad j\geq 2; \qquad \Theta_1^{\rm diag}=0.
	\end{equation}
	Through the first equation of \eqref{eq:theta-re} and equation \eqref{eq:theta-diag}, we complete the proof.
\end{proof}

\section*{Appendix C.\, MI analysis of the plane waves}
\addcontentsline{toc}{section}{Appendix C}

Here we utilize the squared eigenfunction method to construct the solutions of linearized vector matrix Hirota equations. To this purpose, we depart from the stationary zero curvature equation:
\begin{equation}\label{eq:station-zero}
\Psi_x=[\mathbf{U},\Psi],\qquad \Psi_t=[\mathbf{V},\Psi],\qquad \Psi=\begin{pmatrix}
\Psi_{11} &\Psi_{12} \\
\Psi_{21} & \Psi_{22} \\
\end{pmatrix},
\end{equation}
where the matrices $\Psi_{11}$, $\Psi_{12}$, $\Psi_{21}$ and $\Psi_{22}$ have the shape $1\times 1$, $1\times n$,  $n\times 1$ and $n\times n$ respectively. Then we write the corresponding component form:
\begin{equation}\label{eq:station-x}
\begin{split}
\Psi_{11,x}=&\ii\left(\mathbf{q}^{\dag} \Psi_{21}-\Psi_{12}\mathbf{q}\right), \\
\Psi_{12,x}=&\ii\left(2\lambda \Psi_{12}+\mathbf{q}^{\dag}\Psi_{22}-\Psi_{11}\mathbf{q}^{\dag}\right), \\
\Psi_{21,x}=&\ii\left(\mathbf{q}\Psi_{11}-\Psi_{22}\mathbf{q}-2\lambda \Psi_{21}\right), \\
\Psi_{22,x}=&\ii\left(\mathbf{q}\Psi_{12}-\Psi_{21}\mathbf{q}^{\dag}\right).
\end{split}
\end{equation}
Further, taking the second order derivative of $\Psi_{12}$ and $\Psi_{21}$ with respect to $x$, together with the above equations \eqref{eq:station-x} we obtain that
\begin{equation}\label{eq:station-xx}
\begin{split}
\Psi_{12,xx}=&-2\left[\left(\lambda \mathbf{q}^{\dag}-\frac{\ii}{2}\mathbf{q}_x^{\dag}\right)\Psi_{22}-\Psi_{11}\left(\lambda\mathbf{q}^{\dag}-\frac{\ii}{2}\mathbf{q}_x^{\dag}\right)+2\lambda^2\Psi_{12}\right] \\
&-\mathbf{q}^{\dag}\left(\mathbf{q}\Psi_{12}-\Psi_{21}\mathbf{q}^{\dag}\right)+\left(\mathbf{q}^{\dag}\Psi_{21}-
\Psi_{12}\mathbf{q}\right)\mathbf{q}^{\dag}, \\
\Psi_{21,xx}=&2\left[(\lambda\mathbf{q}+\frac{\ii}{2}\mathbf{q}_x)\Psi_{11}-\Psi_{22}(\lambda\mathbf{q}+\frac{\ii}{2}\mathbf{q}_x)-2\lambda^2\Psi_{21}\right]\\
&-\left(\mathbf{q}^{\dag} \Psi_{21}-\Psi_{12}\mathbf{q}\right)+\left(\mathbf{q}\Psi_{12}-\Psi_{21}\mathbf{q}^{\dag}\right)\mathbf{q}.
\end{split}
\end{equation}

Similarly, for the $t$-part of stationary zero-curvature equations, we write the corresponding component form:
\begin{equation}\label{eq:station-t}
\begin{split}
\Psi_{11,t}&=V_{12}\Psi_{21}-\Psi_{12}V_{21}, \\
\Psi_{12,t}&=V_{12}\Psi_{22}-\Psi_{11}V_{12}+V_{11}\Psi_{12}-\Psi_{12}V_{22}, \\
\Psi_{21,t}&=V_{21}\Psi_{11}-\Psi_{22}V_{21}+V_{22}\Psi_{21}-\Psi_{21}V_{11}, \\
\Psi_{22,t}&=V_{21}\Psi_{12}+V_{22}\Psi_{22}-\Psi_{21}V_{12}-\Psi_{22}V_{22},
\end{split}
\end{equation}
where
\begin{equation}
\begin{split}
V_{11}=&\ii\left(\lambda^2-\frac{1}{2}\mathbf{q}^{\dag}\Lambda\mathbf{q}\right), \quad
V_{12}=\ii\left(\lambda\mathbf{q}^{\dag}\Lambda-\frac{\ii}{2}\mathbf{q}^{\dag}_x\Lambda\right), \\
V_{21}=&\ii\left(\lambda\mathbf{q}+\frac{\ii}{2}\mathbf{q}_x\right), \quad
V_{22}=\ii\left(\lambda^2\mathbb{I}_n+\frac{1}{2}\mathbf{q}\mathbf{q}^{\dag}\Lambda\right). \\
\end{split}
\end{equation}
Combining with the equations \eqref{eq:station-xx} and \eqref{eq:station-t}, we will obtain the linearized  equations:
\begin{equation}\label{eq:linearized-eq}
\begin{split}
\ii \Psi_{21,t}&=-\left[\frac{1}{2}\Psi_{21,xx}+\left(\Psi_{21}\mathbf{q}^{\dag}\mathbf{q}-\mathbf{q}\Psi_{12}\mathbf{q}
+\mathbf{q}\mathbf{q}^{\dag}\Psi_{21}\right)\right],\\
\ii\Psi_{12,t}&=\left[\frac{1}{2}\Psi_{12,xx}+\left(\Psi_{12}\mathbf{q}\mathbf{q}^{\dag}
-\mathbf{q}^{\dag}\Psi_{21}\mathbf{q}^{\dag}+\mathbf{q}^{\dag}\mathbf{q}\Psi_{12}\right)\right].
\end{split}
\end{equation}
Moreover, the symmetry relationship $\Psi_{12}=-\Psi_{21}^{\dag}$ guarantees the above linearized equation to satisfy the linearized multi-component Hirota equations.

We construct the solutions of $\Psi$ by the wavefunction of Lax pair. Suppose we have a vector solution $\phi_i(\lambda)$ for Lax pair, by the symmetric proposition we know that
$\phi_j^{\dag}(\lambda^*)$ satisfies the adjoint Lax pair, which implies that the matrix function
$\Psi=\phi_i(\lambda)\phi_j^{\dag}(\lambda^*)$ solves the stationary zero-curvature equations \eqref{eq:station-zero}, where $\phi_i(\lambda)=[\phi_{i,1}(\lambda), \phi_{i,2}^{\T}(\lambda)]^{\T}$. Similarly, we know that $\Psi=\phi_j(\lambda^*)\phi_i^{\dag}(\lambda)$ also solves the equations \eqref{eq:linearized-eq}. Since the equations \eqref{eq:linearized-eq} are unrelated with the spectral parameters $\lambda$. Thus the solutions
\begin{equation}\label{eq:solution-linearized}
\phi_{i,2}(\lambda)\phi_{j,1}^*(\lambda^*)-\phi_{j,2}(\lambda^*)\phi_{i,1}^*(\lambda)
\end{equation}
solves the linearized multi-component NLS equations automatically.

In what follows, we consider how to use the above procedures to solve linear stability analysis for the multi-component Hirota equation with the plane wave solution \eqref{plane}. Suppose we consider the perturbation form:
\begin{equation}\label{eq:perturb}
\Psi_{21}=\begin{pmatrix}
\left(g_1\ee^{\ii\xi( x+\zeta t)}+f_1^*\ee^{-\ii\xi(x+\zeta^* t)}\right)\ee^{\ii(b_1x+c_1t)}\\
\left(g_2\ee^{\ii\xi( x+\zeta t)}+f_2^*\ee^{-\ii\xi(x+\zeta^* t)}\right)\ee^{\ii(b_2x+c_2t)}\\
\vdots\\
\left(g_n\ee^{\ii\xi( x+\zeta t)}+f_n^*\ee^{-\ii\xi(x+\zeta^* t)}\right)\ee^{\ii(b_nx+c_nt)}
\end{pmatrix},
\end{equation}
and $\phi_i(\lambda)=KL_i(\lambda)\ee^{\ii[\xi_i(\lambda)-\lambda]x+\eta_i(\lambda) t}$ and $\phi_j(\lambda)=KL_j(\lambda)\ee^{\ii[\xi_j(\lambda)-\lambda]x+\eta_j(\lambda) t}$ provides a solution
\begin{equation}
g_k=\frac{a_i}{\xi_i(\lambda)+b_k},\qquad f_k=-\frac{a_i}{\xi_j(\lambda)+b_k}
\end{equation}
where $k=1,2,\cdots, n$ and
\begin{equation}
\xi=\xi_i(\lambda)-\xi_j(\lambda)\in \mathbb{R},\qquad \zeta=\frac{\gamma}{2}(\xi_i(\lambda)+\xi_j(\lambda))
\end{equation}
the parameter $\xi$ determines the perturbation frequency and the parameter $\zeta$ determines the gain index. As for the fixed $\xi$, the parameter $\xi_i(\lambda)$ can be determined by the following equations
\begin{equation}\label{eq:dispersion-1}
1+\sum_{k=1}^n\frac{s_ka_k^2}{(\xi_i(\lambda)+b_k)((\xi_i(\lambda)-\xi+b_k))}=0.
\end{equation}
For the multi-component NLS equations, the gain index $\zeta=\xi_i(\lambda)-\frac{\xi}{2}$ will solve the equations
\begin{equation}
1+\sum_{k=1}^n\frac{s_ka_k^2}{(\zeta+b_k)^2-\frac{\xi^2}{4}}=0.
\end{equation}

\section*{Appendix D.\, Asymptotic behaviors of $\mathbf{q}^{[1]_4}$}
\addcontentsline{toc}{section}{Appendix D}

\quad {\it Case 1.} We consider four simple roots of the governing polynomial $\mathcal{F}_4(z)$. For instance,  as $\kappa_0=9, \kappa_2=-10, \kappa_1=\kappa_3=0$ are taken for the line-typed structure along $t$-axis, four roots of $\mathcal{F}_4$ are $z=\pm 1, \pm 3$.  As $\kappa_0=9, \kappa_2=10, \kappa_1=\kappa_3=0$ are taken for the line-typed structure along $x$-axis, four roots of $\mathcal{F}_4$ are $z=\pm \mathrm{i}, \pm 3\mathrm{i}$.  As $\kappa_0=1, \kappa_1=\kappa_2=\kappa_3=0$ are taken for the square-typed structure, four roots of $\mathcal{F}_4$ are $z=\mathrm{e}^{\left(2s-1\right)\pi\mathrm{i}/4}, s=1, 2, 3, 4$. Then we yield the formula of asymptotic behavior
 \begin{gather}
q_j^{[1]_4}=q_j^\mathrm{bg}\bigg[1+\sum_{\delta_1=0}^1\sum_{\delta_2=0}^1R_j\left(x-x_{\delta_1, \delta_2}, t-t_{\delta_1, \delta_2}\right)
\bigg]\mathrm{e}^{-2\mathrm{i}\omega_j}+\mathcal{O}\left(\frac{1}{h}\right), \,\, h\to +\infty.
\end{gather}
with the linear superposition of four single rogue waves in the leading term, where
\begin{itemize}
\item $x_{\delta_1, 0}=(-1)^{\delta_1}(h+(\kappa_{0, 0}+\kappa_{2, 0})/16\zeta)+(\kappa_{1, 0}+\kappa_{3, 0}-8)/16\zeta, \, x_{\delta_1, 1}=(-1)^{\delta_1}(3h-(9\kappa_{2, 0}+\kappa_{0, 0})/48\zeta)-(\kappa_{1, 0}+9\kappa_{3, 0}+8)/16\zeta, \, t_{\delta_1, \delta_2}=0$ for the case $\kappa_0=9, \kappa_2=-10, \kappa_1=\kappa_3=0$  (see Figs. \ref{RW4}(a, b));
\item $x_{\delta_1, 0}=-(\kappa_{1, 0}-\kappa_{3, 0}+4)/16\zeta, \, t_{\delta_1, 0}=(-1)^\delta(h+(\kappa_{0, 0}-\kappa_{2, 0})/16\zeta)/\zeta, \, x_{\delta_1, 1}=(\kappa_{1, 0}-9\kappa_{3, 0}+36)/16\zeta\, t_{\delta_1, 1}=(-1)^\delta(3h-(\kappa_{0, 0}-9\kappa_{2, 0})/48\zeta)/\zeta$ (see Fig. \ref{RW4}(c, d));
\item $x_{\delta_1, \delta_2}=(-1)^{\delta_1}(h/\sqrt{2}+\sqrt{2}(\kappa_{0, 0}-\kappa_{2, 0})/8\zeta)+(1-\kappa_{3, 0})/4\zeta, \, t_{\delta_1, \delta_2}=(-1)^{\delta_2}(h/\sqrt{2}+\sqrt{2}(\kappa_{0, 0}-\kappa_{2, 0})/8\zeta)/\zeta+(-1)^{\delta_2}(\kappa_{1, 0}+3)/4\zeta$ (see Fig. \ref{RW4}(e, f)).
\end{itemize}

{\it Case 2.} For the second case, we consider a  double root and two simple roots of the governing polynomial $\mathcal{F}_4$, which arrange a line-typed structure. For instance,  as  $\kappa_2=-1, \kappa_0=\kappa_1=\kappa_3=0$ are taken, then the governing polynomial $\mathcal{F}_4$ has a double root $z=0$ and two simple root $z=\pm 1$. Under the constraint $\kappa_{0, 0}=0$, we deduce the formula of asymptotic behavior
\begin{gather}
q_j^{[1]_4}=q_j^\mathrm{bg}\left[1+p_{j, 2}\left(x, t\right)+\sum_{\delta=0}^1R_j\left(x-x_\delta, t\right)\right]\mathrm{e}^{-2\mathrm{i}\omega_j}
+\mathcal{O}\left(\frac{1}{h}\right), \,\, h\to +\infty,
\end{gather}
where $x_\delta=(-1)^\delta(h-\kappa_{2, 0}/2\zeta)-(\kappa_{3, 0}+\kappa_{1, 0}+1)/2\zeta$ and the parameters of $p_{j, 2}$ are $\alpha_0=\kappa_{0, 1}, \alpha_1=\kappa_{1, 0}\mathrm{i}, \alpha_2=2\zeta$.
The leading term is linear superposition of two single rogue waves located at $\left(\pm h-1/2\zeta, 0\right)$ and a double rogue wave (see Figs. \ref{RW4a}(a, b)).

{\it Case 3.}
For the third case, we consider a two real double roots of the governing polynomial $\mathcal{F}_4$.  For instance,  we can take $\kappa_2=-2, \kappa_1=1, \kappa_0=\kappa_3=0$, then the governing polynomial $\mathcal{F}_4$ has two double real roots  $z=\pm 1$. Under the  constraint  $\kappa_{3, 0}+\kappa_{1, 0}=\kappa_{2, 0}+\kappa_{0, 0}=0$, we derive  the formula of asymptotic behavior
\begin{gather}
q_j^{[1]_4}=q_j^\mathrm{bg}\bigg[1+\sum_{\delta=0}^1p_{j, 2}\left(x+\left(-1\right)^\delta h, t\right)+\bigg]\mathrm{e}^{-2\mathrm{i}\omega_j}
+\mathcal{O}\left(\frac{1}{h}\right), \,\, h\to +\infty.
\end{gather}
where the parameters of $p_{j, 2}$ are $\alpha_0=\kappa_{0, 0}+\kappa_{1, 0}, \alpha_1=-2(\kappa_{0, 1}+\kappa_{1, 1}+\kappa_{2, 1}+\kappa_{0, 0}^2)\mathrm{i}, \alpha_2=4\zeta$.
The leading term is linear superposition of two double rogue waves (see Figs. \ref{RW4a}(c, d)).

{\it Case 4.}
For the fourth case, we consider a simple root and a triple root of $\mathcal{F}_4$ and the  corresponding rogue waves  arrange in a line. As we take $\kappa_3=1, \kappa_0=\kappa_1=\kappa_2=0$, the governing polynomial $\mathcal{F}_4$ has a simple roots $z=-1$ and a triple root $z=0$. Under the constraint $\kappa_{0, 0}=\kappa_{1, 0}=0$, we deduce the formula of  asymptotic behavior
\begin{gather}
q_j^{[1]_4}=q_j^\mathrm{bg}\left[1+p_{j, 3}\left(x, t\right)+R_j\left(x-x_0, t\right)\right]\mathrm{e}^{-2\mathrm{i}\omega_j}
+\mathcal{O}\left(\frac{1}{h}\right), \,\, h\to +\infty,
\end{gather}
where $x_0=(2\kappa_{2, 0}-2\kappa_{3, 0}-1)/2\zeta-h$ and the parameters of $p_{j, 3}$ are $\alpha_0=\kappa_{0, 2}, \alpha_1=\mathrm{i}\kappa_{1, 1}, \alpha_2=-2\kappa_{2, 0}, \alpha_3=-6\mathrm{i}\zeta$. The leading term is linear superposition of a single rogue wave located at $\left(x_0, 0\right)$ and a triple rogue wave (see Fig. \ref{RW4a}(e, f)).

{\it Case 5.}
For the fifth case, we consider a quadruple root of $\mathcal{F}_4$.  For the convenience of the representation of  the leading term of asymptotic formula, we define
\begin{gather}
p_{j, 4}=-\frac{2\mathrm{i}\zeta}{b_j-\mathrm{i}\zeta}\frac{\left(n+1\right)L_jL_0^*-\mathbf{L}^\dag \mathbf{L}}{\mathbf{L}^\dag \mathbf{L}},\quad j=1, 2, \cdots, n,
\end{gather}
where $\mathbf{L}, L_j's, L_0$ are defined in Eqs. \eqref{L} and \eqref{VL} with $\ell=4$. As we take $\kappa_0=\kappa_1=\kappa_2=\kappa_3=0$, the governing polynomial $\mathcal{F}_4$ has a quadruple root $z=0$.   To obtain the rogue wave solution in the leading term of the asymptotic formula, the constraint $\kappa_{0, 0}=\kappa_{0, 1}=\kappa_{0, 2}=\kappa_{1, 0}=\kappa_{1, 1}=\kappa_{2, 0}=0$ is posed.  Then $q_j^{[1]_4}=q_j^\mathrm{bg}\left(1+p_{j, 4}\right)\mathrm{e}^{-2\mathrm{i}\omega_j}$ with $\alpha_0=\kappa_{0, 3}, \alpha_1=\mathrm{i}\kappa_{1, 2}, \alpha_2=-2\kappa_{2, 1}, \alpha_3=-6\kappa_{3, 0}\mathrm{i}, \alpha_4=24\zeta$. For other quadruple  root of $\mathcal{F}_4$, we also have the formula of asymptotic behavior as $d\to +\infty$. For instance,  as $\kappa_0=1, \kappa_1=-4, \kappa_2=6, \kappa_3=-4$, the governing polynomial $\mathcal{F}_4=\left(z-1\right)^4$, which has a quadruple root $z=1$. Under the constraint $\kappa_{0, 0}=\kappa_{0, 1}=\kappa_{0, 2}=\kappa_{1, 0}=\kappa_{1, 1}=\kappa_{2, 0}=0$, we derive the formula of asymptotic behavior
\begin{gather}
q_j^{[1]_4}=q_j^\mathrm{bg}\left[1+p_{j, 4}\left(x-h, t\right)\right]\mathrm{e}^{-2\mathrm{i}\omega_j}+\mathcal{O}\left(\frac{1}{h}\right).
\end{gather}
where the parameters of $p_{j, 4}$ are $\alpha_0=\kappa_{0, 3}, \alpha_1=\mathrm{i}\kappa_{1, 2}, \alpha_2=-2\kappa_{2, 1}, \alpha_3=-6\kappa_{3, 0}\mathrm{i}, \alpha_4=24\zeta$.

\vspace{0.1in}
\noindent {\bf Acknowledgements}
\vspace{0.1in}

G.Z.  acknowledges support from the China Postdoctoral Science Foundation under Grant No. 2019M660600, L. L. is supported by the National Natural Science Foundation of China (Grant No. 11771151), the Guangzhou Science and Technology Program of China (Grant No. 201904010362), and the Fundamental Research Funds for the Central Universities of China (Grant No. 2019MS110). Z.Y.  acknowledges support from the
National Natural Science Foundation of China (Grant Nos.11925108 and 11731014).


\begin{thebibliography}{99}

\bibitem{nls67} D. Benney and A. C. Newell, The propagation of nonlinear
wave envelopes, J. Math. Phys. 46 (1967) 133-139.

\bibitem{nls67b} L. A. Ostrowskii, Propagation of wave packets and space-time self-focussing in a nonlinear medium,
 Sov. Phys. JETP 24 (1967) 797-800.

\bibitem{nls68} V. E. Zakharov, Stability of periodic waves of finite amplitude on the surface of a deep fluid,
J. Appl. Mech. Tech. Phys. 9 (1968) 190-194.

\bibitem{nlsb1} G. P. Agrawal, {\it Nonlinear Fiber Optics} (Academic Press, New York, 1995).

\bibitem{nlsb2} A. Hasegawa and Y. Kodama, {\it Solitons in Optical Communications} (Oxford University Press, England, 1995).

\bibitem{nlsb3} N. Akhmediev and A. Ankiewicz, {\it Solitons: Nonlinear Pulses and Beams} (Chapman and Hall, London, 1997).

 \bibitem{nlsb4} L. Pitaevskii and S. Stringari, {\it Bose-Einstein Condensation} (Oxford University Press, Oxford, 2003).

 \bibitem{nlsb5} Y. S. Kivshar and G. P. Agrawal, {\it Optical Solitons: From Fibers to Photonic Crystals} (Academic Press, New York, 2003).

 \bibitem{nlsb6} B. A. Malomed, D. Mihalache, F. Wise, and L. Torner, J. Opt. B: Quantum Semiclassical Opt. 7 (2005) R53.


\bibitem{frw} Z. Yan, Financial rogue waves, Commun. Theor. Phys. 54 (2010) 947.


\bibitem{nlsbook1} C. Sulem and P. L. Sulem. {\it The Nonlinear Schr\"odinger Equation: Self-Focusing and Wave Collapse} (Springer, New York, 1999).

\bibitem{nlsbook2} G. Fibich, {\it The Nonlinear Schr\"odinger Equation: Singular Solutions and Optical Collapse} (Springer, New York, 2015).


\bibitem{nlsist} V. E. Zakharov and A. B. Shabat, Exact theory of two-dimensional self-focusing and one-dimensional self-modulation of waves in nonlinear media, Sov. Phys. JETP 34 (1972) 62.

 \bibitem{ist} M. J. Ablowitz and P. A. Clarkson, {\it Solitons, Nonlinear Evolution Equations and Inverse Scattering} (Cambridge University Press,
Cambridge, 1991).

 \bibitem{dt} V. B. Matveev and V. Matveev, {\it Darboux Transformations and Solitons} (Springer-Verlag, New York, 1991).

\bibitem{Peregrine1983} D. Peregrine,  Water waves, nonlinear Schr\"odinger equations and their solutions, J. Aust. Math. Soc. B, Appl. Math.  25 (1983) 16.


\bibitem{nail09a} N. Akhmediev, A. Ankiewicz, and M. Taki, Waves that appear from nowhere and disappear without a trace, Phys. Lett. A 373 (2009) 675.

\bibitem{nail09b} N. Akhmediev, A. Ankiewicz, and J. M. Soto-Crespo, Rogue waves and rational solutions of the nonlinear Schr\"odinger equation, Phys. Rev. E 80 (2009) 026601.

\bibitem{Guo2012} B. Guo, L. Ling, Q. Liu, Nonlinear Schr\"{o}dinger equation: generalized Darboux transformation and rogue wave solutions, Phys. Rev. E 85 (2012) 026607.

\bibitem{rwbook17} B. Guo, L. Tian, Z. Yan, L. Ling, and Y. Wang, {\it Rogue Waves:  Mathematical Theory and Applications in Physics}
( De Gruyter, Berlin, 2017).

\bibitem{rogon} Z. Yan, Nonautonomous ``rogons" in the inhomogeneous nonlinear Schr\"odinger equation with variable coefficients, Phys. Lett. A
 374 (2010) 672-679.

\bibitem{orw} D. R. Solli, C. Ropers, P. Koonath, B. Jalali, Optical rogue waves, Nature  450 (2007) 1054.

\bibitem{ps} B. Kibler, J. Fatome, C. Finot, G. Millot, F. Dias, G. Genty, N. Akhmediev, and J. M. Dudley, The Peregrine soliton in nonlinear fibre optics, Nature Phys. 6 (2010) 790.


\bibitem{wrw} A. Chabchoub, N. P. Hoffmann, and N. Akhmediev, Rogue wave observation in a water wave tank, Phys. Rev. Lett.
106 (2011) 204502.

\bibitem{prw} H. Bailung, S. K. Sharma, and Y. Nakamura, Observation of Peregrine solitons in a multicomponent plasma
with negative ions, Phys. Rev. Lett. 107 (2011) 255005.

\bibitem{Kharif2003} C. Kharif and E. Pelinovsky, Physical mechanisms of the rogue wave phenomenon, Eur. J. Mech. B Fluids 22 (2003) 603-634.

\bibitem{Pelinovsky2016} E. Pelinovsky and C. Kharif (eds) {\it Extreme Ocean Waves} (2nd ed.) (Springer, New York, 2016).

\bibitem{mi} T. B. Benjamin and J. E. Feir, The disintegration of wave trains on deep water, J. Fluid Mech. 27 (1967) 417.

\bibitem{mi2} V. E. Zakharov and L. A. Ostrovsky, Modulational instability: The beginning, Physica (Amsterdam) 238D (2009) 540.


\bibitem{hirota} A. Ankiewicz, J. M. Soto-Crespo, and N. Akhmediev, Discrete rogue waves of the Ablowitz-Ladik and Hirota equations,
Phys. Rev. E 81 (2010) 046602.

\bibitem{al} X. Y. Wen and  Z. Yan, Modulational instability and dynamics of multi-rogue wave solutions for the
discrete Ablowitz-Ladik equation, J. Math. Phys. 59 (2018) 073511.

\bibitem{ss} U. Bandelow and N. Akhmediev,Persistence of rogue waves in extended nonlinear Schr\"odinger equations: integrable Sasa-Satsuma case,
Phys. Lett. A 376 (2012) 1558.

\bibitem{ss2} S. Chen, Twisted rogue-wave pairs in the Sasa-Satsuma equation, Phys. Rev. E 88 (2013) 023202.

\bibitem{gi} S. Xu and J. He, The rogue wave and breather solution of the Gerdjikov-Ivanov equation, J. Math. Phys. 53 (2012) 063507.

\bibitem{ds1}Y. Ohta and J. K. Yang, Rogue waves in the Davey-Stewartson I equation, Phys. Rev. E 86 (2012) 036604.

\bibitem{dnls}H. N. Chan, K. W. Chow, D. J. Kedziora, R. H. J. Grimshaw, and E. Ding, Rogue wave modes for a derivative nonlinear Schr\"odinger model,
Phys. Rev. E 89 (2014) 032914.

\bibitem{dnls2} X. Y. Wen, Y.  Yang, and Z. Yan, Generalized perturbation (n, M)-fold Darboux transformations and multi-rogue-wave structures
for the modified self-steepening nonlinear Schr\"odinger equation, Phys. Rev. E 92 (2015) 012917.

\bibitem{qnls} A. Chowdury, D. J. Kedziora, A. Ankiewicz, and N. Akhmediev,Breather solutions of the integrable quintic nonlinear Schr\"odinger equation and their interactions, Phys. Rev. E 91 (2015) 022919.

\bibitem{5nls} Y. Yang, Z. Yan, and B. A. Malomed, Rogue waves, rational solitons, and modulational instability in an integrable
fifth-order nonlinear Schr\"odinger equation, Chaos 25 (2015) 103112.

\bibitem{Bludov2010} Y. V. Bludov, V. Konotop, N. Akhmediev, Vector rogue waves in binary mixtures of Bose-Einstein condensates, Eur. Phys. J. Special Topics 185 (2010) 169-180.

\bibitem{guo2011} B.-L. Guo, L.-M. Ling, Rogue wave, breathers and bright-dark-rogue solutions for the coupled Schr\"{o}dinger equations, Chin. Phys. Lett. 28 (2011) 110202.

\bibitem{yan2011} Z. Yan, Vector financial rogue waves, Phys. Lett. A 375 (2011) 4274-4279.

\bibitem{Baronio2012} F. Baronio, A. Degasperis, M. Conforti, S. Wabnitz, Solutions of the vector nonlinear Schr\"{o}dinger equations: evidence for deterministic rogue waves, Phys. Rev. Lett. 109 (2012) 044102.

\bibitem{Tao2012} Y. Tao, J. He, Multisolitons, breathers, and rogue waves for the Hirota equation generated by the Darboux transformation, Phys. Rev. E 85 (2012) 026601.

\bibitem{Li2013} L. Li, Z. Wu, L. Wang, J. He, High-order rogue waves for the Hirota equation, Ann. Phys. 334 (2013) 198-211.

\bibitem{He2013} J. He, H. Zhang, L. Wang, K. Porsezian, A. Fokas, Generating mechanism for higher-order rogue waves, Phys. Rev. E 87  (2013) 052914.

\bibitem{Ling2014} L. Ling, B. Guo, L.-C. Zhao, High-order rogue waves in vector nonlinear Schr\"{o}dinger equations, Phys. Rev. E 89  (2014) 041201.

\bibitem{Baronio2014} F. Baronio, M. Conforti, A. Degasperis, S. Lombardo, M. Onorato, S. Wabnitz, Vector rogue waves and baseband modulation
    instability in the defocusing regime, Phys. Rev. Lett. 113  (2014) 034101.

\bibitem{Zhao2012} L.-C. Zhao, J. Liu, Localized nonlinear waves in a two-mode nonlinear fiber, J. Opt. Soc. Am. B 29 (2012) 3119-3127.

\bibitem{Chen2013} S. Chen, L.-Y. Song, Rogue waves in coupled Hirota systems, Phys. Rev. E 87 (2013) 032910.

\bibitem{Ling2016} L. Ling, L.-C. Zhao, B. Guo, Darboux transformation and classification of solution for mixed coupled nonlinear
Schr\"{o}dinger equations, Commun. Nonlinear Sci. Numer. Simul. 32 (2016) 285-304.

\bibitem{Zhao2014} L.-C. Zhao, G.-G. Xin, Z.-Y. Yang, Rogue-wave pattern transition induced by relative frequency, Phys. Rev. E 90 (2014) 022918.

\bibitem{Zhao2016} L.-C. Zhao, B. Guo, L. Ling, High-order rogue wave solutions for the coupled nonlinear Schr\"{o}dinger equations-II, J. Math. Phys. 57  (2016) 043508.

\bibitem{Chen2015} S. Chen, D. Mihalache, Vector rogue waves in the Manakov system: diversity and compossibility, J. Phys. A 48 (2015) 215202.

\bibitem{LingZ-19} L. Ling, L. C. Zhao, Modulational instability and homoclinic orbit solutions in vector nonlinear Schr\"{o}dinger equation. Commun. Nonlinear Sci. Numer. Simul.  72 (2019) 449-471.

\bibitem{wen2015} X. Y. Wen and Z. Yan, Modulational instability and higher-order rogue waves with parameters modulation in a
coupled integrable AB system via the generalized Darboux transformation,  Chaos 25 (2015) 123115.


\bibitem{wen2016} X. Y. Wen and Z. Yan, B. A. Malomed, Higher-order vector discrete rogue-wave states in the coupled
Ablowitz-Ladik equations: Exact solutions and stability, Chaos 26 (2016) 123110.

\bibitem{zhang2017} G. Zhang, Z. Yan, X. Y. Wen, and Y. Chen, Interactions of localized wave structures and dynamics in the defocusing coupled nonlinear Schr\"odinger equations, Phys. Rev. E 95 (2017) 042201.

\bibitem{zhang2017b} G. Zhang, Z. Yan, and X. Y. Wen, Modulational instability, beak-shaped rogue waves, multi-dark-dark solitons and
dynamics in pair-transition-coupled nonlinear Schr\"odinger equations, Proc. R. Soc. A 473 (2017) 20170243.

\bibitem{zhang2019} G. Zhang, Z. Yan, L. Wang, The general coupled Hirota equations: modulational instability and higher-order
vector rogue wave and multi-dark soliton structures, Proc. R. Soc. A 475 (2019) 20180625.

\bibitem{Zhao2013} L.-C. Zhao, J. Liu, Rogue-wave solutions of a three-component coupled nonlinear Schr\"{o}dinger equation, Phys. Rev. E 87 (1)
    (2013) 013201.

\bibitem{Baronio2013} F. Baronio, M. Conforti, A. Degasperis, S. Lombardo, Rogue waves emerging from the resonant interaction of three waves, Phys. Rev. Lett. 111 (2013) 114101.

\bibitem{zhang2018}  G. Zhang, Z. Yan, Three-component nonlinear Schr\"odinger equations: Modulational instability, $N$th-order vector rational and
semi-rational rogue waves, and dynamics, Commun. Nonlinear Sci. Numer. Simulat. 62 (2018) 117-133.


\bibitem{zhang2018b} G. Zhang, Z. Yan, and X. Y. Wen, Three-wave resonant interactions: Multi-dark-dark-dark solitons,
breathers, rogue waves, and their interactions and dynamics, Physica D 366 (2018) 27-42.

\bibitem{Ablowitz2004} M. J. Ablowitz, B. Prinari, and A. Trubatch, {\em Discrete and
 	Continuous Nonlinear Schr\"odinger Systems} (Cambridge University Press, Cambridge, 2004).


\bibitem{manakov} S. V. Manakov, On the theory of two-dimensional stationary self-focusing of electromagnetic waves,
Sov. Phys.-JETP 38 (1974) 248-253.

\bibitem{nnls1} A. C. Scott, The vibrational structure of Davydov solitons, Phys. Scr. 25 (1982) 651-658.

\bibitem{nnls2} C. R. Menyuk, Nonlinear pulse propagation in birefringent optical fibers, IEEE J. Quantum Electron. 23 (1987) 174.

\bibitem{nnls3} C. Yeh and L. Bergman,Enhanced pulse compression in a nonlinear fiber by a wavelength division
multiplexed optical pulse, Phys. Rev. E 57 (1998) 2398.

\bibitem{nnls4} N. Akhmediev, W. Kr\'olikowski, and A. W. Snyder, Partially coherent solitons of variable shape, Phys. Rev. Lett. 81 (1998) 4632.

\bibitem{Faddeev} L. Faddeev, L. Takhtajan, {\it Hamiltonian Methods in the Theory of Solitons} (Springer, Berlin, 1987).

\bibitem{KedzioraAA-11} D. J. Kedziora, A. Ankiewicz and N. Akhmediev, Circular rogue wave clusters, Phys. Rev. E 84 (2011) 056611.

\bibitem{KedzioraAA-13}D. J. Kedziora, A. Ankiewicz and N. Akhmediev, Classifying the hierarchy of nonlinear-Schr\"{o}dinger-equation rogue-wave solutions, Phys. Rev. E 88 (2013) 013207.

\bibitem{Yang-20} B Yang and J-K Yang, Universal patterns of rogue waves, arXiv: 2009.06060v1.

\bibitem{BilmanLM-20} D. Bilman, L. Ling, P. D. Miller, Extreme superposition: Rogue waves of infinite order and the Painleve-III hierarchy, Duke Math. J. 169 (2020) 671-760.

\bibitem{Ling2017} L. Ling, L.-C. Zhao, Z. Yang, B. Buo, Generation mechanisms of fundamental rogue wave spatial-temporal structure, Phys. Rev. E 96 (2017) 022211.


\bibitem{HeZWPF-13} J. He, H. Zhang, L. Wang, K. Porsezian, and A. Fokas,  Generating mechanism for higher-order rogue waves, Phys. Rev. E  87 (2013) 052914.


\bibitem{BilmanM-19} D. Bilman, P. D. Miller, A robust inverse scattering transform for the focusing nonlinear Schr\"odinger equation, Comm. Pure Appl. Math. 72 (2019) 1722-1805.



\bibitem{Terng1998} C.-L. Terng, K. Uhlenbeck, B\"{a}cklund transformations and loop group actions, Comm. Pure. Appl. Math. 53 (2000) 1.

\bibitem{KrausBK-15} D. Kraus, G. Biondini, G. Kova\v{c}i\v{c}, The focusing Manakov system with nonzero boundary conditions, Nonlinearity 28 (2015) 3101.











\bibitem{zlyk2020} G. Zhang, L. Ling, Z. Yan, V. V. Konotop, Parity-time-symmetric vector rational rogue wave solutions in any
  $n$-component nonlinear Schr\"odinger models, preprint (2020).



\bibitem{dnlsp} L. Wang and Z. Yan, Rogue wave formation and interactions in the defocusing nonlinear Schr\"odinger equation with external potentials, Appl. Math. Lett. 111 (2021) 106670.

\bibitem{Strauss77} W. A. Strauss, Existence of solitary waves in higher dimensions, Comm. Math. Phys. 55 (1977) 149–162.

\bibitem{Martel-11} Y. Martel and F. Merle, Description of two soliton collision for the quartic gKdV equation, Ann. of Math.  174 (2011) 757–857.

\bibitem{Martel-11-1} Y. Martel and F. Merle, Inelastic interaction of nearly equal solitons for the quartic gKdV equation, Invent. Math. 183 (2011) 563-648.


\bibitem{GrillakisSS-87} M. Grillakis, J. Shatah and W. Strauss, Stability theory of solitary waves in the presence of symmetry. I,  J. Funct. Anal. 74 (1987) 160-197.

\end{thebibliography}
\end{document}